%% file: main.tex
\documentclass{LMCS} 
\usepackage[all,arc]{xy} 
\usepackage{enumerate,hyperref}
\usepackage{epsfig}  
\usepackage{amssymb} 
\usepackage{mathpartir} 
\input{macro}
\newenvironment{oldthm}[1]{
\let\oldthethm\thethm%
\renewcommand{\thethm}{\ref{#1}}%
\begin{thm}%
}
{%
\end{thm}
\addtocounter{thm}{-1}%
\let\thethm\oldthethm%
}
\newenvironment{oldprop}[1]{
\let\oldthethm\thethm%
\renewcommand{\thethm}{\ref{#1}}%
\begin{prop}%
}
{%
 \end{prop}
\addtocounter{thm}{-1}%
\let\thethm\oldthethm%
}
\newenvironment{oldlem}[1]{
\let\oldthethm\thethm%
\renewcommand{\thethm}{\ref{#1}}%
\begin{lem}%
}
{%
\end{lem}
\addtocounter{thm}{-1}%
\let\thethm\oldthethm%
}
\def\doi{7 (2:7) 2011}
\lmcsheading%
{\doi}
{1--53}
{}
{}
{Oct.~10, 2010}
{May\phantom.~16, 2011}
{}

\begin{document}

\title[Symbolic and Asynchronous Semantics via Normalized
Coalgebras]{Symbolic and Asynchronous Semantics\\ via Normalized
Coalgebras\rsuper*}

\author[F.~Bonchi]{Filippo Bonchi\rsuper a}   
\address{{\lsuper a}ENS Lyon, Universit\'e de Lyon, LIP (UMR 5668 CNRS ENS Lyon UCBL INRIA)} 
\email{filippo.bonchi@ens-lyon.fr}  
\thanks{{\lsuper a}This work was carried out during the tenure of an ERCIM ``Alain
Bensoussan'' Fellowship Programme}   

\author[U.~Montanari]{Ugo Montanari\rsuper b} 
\address{{\lsuper b}Department of Informatics, University of Pisa}  
\email{ugo@di.unipi.it}  
\thanks{{\lsuper b}Research supported in part by IST-FP7-FET open-IP project \textsc{ASCENS} }    

\titlecomment{{\lsuper*}This paper is an extended version of \cite{CALCO09}.}

\begin{abstract}
The operational semantics of interactive systems is usually
described by labeled transition systems. Abstract semantics (that is
defined in terms of bisimilarity) is characterized by the final
morphism in some category of coalgebras.
Since the behaviour of interactive systems is for many reasons
infinite, symbolic semantics were introduced as a mean to define
smaller, possibly finite, transition systems, by employing symbolic
actions and avoiding some sources of infiniteness.
Unfortunately, symbolic bisimilarity has a different shape with
respect to ordinary bisimilarity, and thus the standard coalgebraic
characterization does not work.
In this paper, we introduce its coalgebraic models.

We will use as motivating examples two asynchronous formalisms: open
Petri nets and asynchronous pi-calculus. Indeed, as we have shown in
a previous paper, asynchronous bisimilarity can be seen as an
instance of symbolic bisimilarity.
\end{abstract}

\maketitle

\section*{Introduction}

\noindent A compositional interactive system is usually defined as a labelled
transition system (\lts) where states are equipped with an algebraic
structure. Abstract semantics is often defined as bisimilarity. Then
a key property is that ``bisimilarity is a congruence'', i.e., that
abstract semantics respects the algebraic operations.

Universal Coalgebra \cite{Rut96} provides a categorical framework
where the behaviour of dynamical systems can be characterized as
\emph{final semantics}. More precisely, if $\coalg{B}$ (i.e., the
category of $\Fun{B}$-coalgebras and $\Fun{B}$-cohomomorphisms for a
certain endofunctor $\Fun{B}$) has a final object, then the behavior
of a $\Fun{B}$-coalgebra is defined as a final morphism.
Intuitively, a final object is a universe of abstract behaviors and
a final morphism is a function mapping each system in its abstract
behavior.
Ordinary \lts s can be represented as coalgebras for a suitable
functor. Then, two states are bisimilar if and only if they are
identified by a final morphism. The image of a certain \lts\ through
a final morphism is its minimal representative (with respect to
bisimilarity), which in the finite case can be computed via
the partition refinement algorithm \cite{KannelakisSmolka}. Existence and construction of the minimal
transition system is a key property of the coalgebraic
approach. It allows to model check efficiently for several
properties by eliminating redundant states once and for all. In fact
most model checking logics are {\em adequate}, namely either a
formula holds in both the given system and in its minimal
representative or it does not hold in both of them.


When bisimilarity is not a congruence, the abstract semantics is
defined either as the \emph{largest congruence contained in
bisimilarity} \cite{RobinBook} or as the \emph{largest bisimulation
that is also a congruence} \cite{ms92}.
%
In this paper we focus on the latter and we call it \emph{saturated
bisimilarity} ($\satbis$). Indeed it coincides with ordinary
bisimilarity on the \emph{saturated transition system} that is
obtained from the original $\lts$ by adding the transition
$p\atr{c}{a}{}q$, for every context $c$, whenever $c(p)\tr{a}q$.

Many interesting abstract semantics are defined in this way. For
example, since late and early bisimilarity of the $\pi$-calculus
\cite{Pi} are not preserved under substitution (and thus under input
prefixes), in \cite{Sang96} Sangiorgi introduces \emph{open
bisimilarity} as the largest bisimulation on $\pi$-calculus agents
which is closed under substitutions. Other noteworthy examples are
asynchronous $\pi$-calculus \cite{AmadioCONCUR96,Honda91}, mobile
ambients calculus \cite{Ambient,MerroNardelli} and (explicit
\cite{WischikG04}) fusion calculus \cite{Fusion}.
The definition of saturated bisimilarity as ordinary bisimilarity on
the saturated \lts\ often makes infinite the portion of \lts\
reachable by any nontrivial agent and, in any case, 
is very inefficient, since it introduces a large number of
additional states and transitions. Inspired by Hennessy and Lin
\cite{HennessyLin}, who introduced a {\em symbolic} semantics of
value passing calculi, Sangiorgi defines in \cite{Sang96} a symbolic
transition system and symbolic bisimilarity that efficiently
characterizes open bisimilarity. After this, many formalisms have
been equipped with a symbolic semantics.

In \cite{FOSSACS08}, we have introduced a general model that
describes at an abstract level both saturated and symbolic
semantics. In this abstract setting, a symbolic transition
$p\atr{c}{\alpha}{\beta}p'$ means that $c(p)\tr{\alpha}p'$ and $c$
is a smallest context that allows $p$ to performs such a transition.
Moreover, a certain \emph{derivation relation $\vdash$} amongst the
transitions of a system is defined: $p\atr{c_1}{\alpha_1}{}p_1
\vdash p\atr{c_2}{\alpha_2}{}p_2$ means that the latter transition
is a logical consequence of the former. In this way, if all and only
the saturated transitions are logical consequences of symbolic
transitions, then saturated bisimilarity can be retrieved via the
symbolic \lts.

Unfortunately, the ordinary bisimilarity over the symbolic
transition system differs from saturated bisimilarity. Symbolic
bisimilarity is thus defined with an asymmetric shape: in the
bisimulation game, when a player proposes a transition, the opponent
can answer with a move with a different label. For example in the
open $\pi$-calculus, a transition $p\atr{[a=b]}{\tau}{}p'$ can be
matched by $q \tr{\tau}q'$. Moreover, the bisimulation game does not
restart from $p'$ and $q'$, but from $p'$ and $q'\{b / a\}$.

\bigskip

For this reason, ordinary coalgebras fail to characterize symbolic
bisimilarity. Here, we provide coalgebraic models for it by relying
on the framework of \cite{FOSSACS08}.

\bigskip

Consider the example of open bisimilarity discussed above. The fact
that open bisimulation does not relate the arriving states $p'$ and
$q'$, but $p'$ and $q'\{b / a\}$, forces us to look for models
equipped with an algebraic structure.
In~\cite{TP97}, \emph{bialgebras} are introduced as a both algebraic
and coalgebraic model, while an alternative approach based on {\em
structured coalgebras}, i.e., on coalgebras in categories of
algebras, is presented in~\cite{CGH98a}. In this paper we adopt the
latter and we introduce $\coalg{\PDA}$ (Section
\ref{sec:CIScoalgebra}), a category of structured coalgebras where
the saturated transition system can be naively modeled in such a way
that $\satbis$ coincides with the kernel of a final morphism.
Then, we focus only on those $\PDA$-coalgebras whose sets of
transitions are closed w.r.t.\ the derivation relation $\vdash$.
These form the category of \emph{saturated coalgebras}
$\coalg{\DMIS}$ (Section \ref{sec:saturatedcoalgebra}) that is
(isomorphic to) a covariety of $\coalg{\PDA}$. Thus, it has a final
object and bisimilarity coincides with the one in $\coalg{\PDA}$.

In order to characterize symbolic bisimilarity, we introduce the
notions of \emph{redundant transition} and \emph{semantically
redundant transition}. Intuitively, a transition
$p\atr{c_2}{\alpha_2}{}q$ is redundant if there exists another
transition $p\atr{c_1}{\alpha_1}{}p_1$ that logically implies it,
that is $p\atr{c_1}{\alpha_1}{}p_1 \vdash p\atr{c_2}{\alpha_2}{}q$;
it is semantically redundant, if it is ``redundant up to
bisimilarity'', i.e.,  $p\atr{c_1}{\alpha_1}{}p_1 \vdash
p\atr{c_2}{\alpha_2}{}p_2$ and $q$ is bisimilar to $p_2$. Now, in
order to retrieve saturated bisimilarity by disregarding redundant
transitions, we have to remove from the saturated transition system
not only all the redundant transitions, but also the semantically
redundant ones. This is done in the category of \emph{normalized
coalgebras} $\coalg{\DMIN}$ (Section \ref{sec:normalizedcoalgebra}).
These are defined as coalgebras without redundant transitions. Thus,
by definition, a final coalgebra in $\coalg{\DMIN}$ has no
\emph{semantically} redundant transitions. 

We prove that $\coalg{\DMIS}$ and $\coalg{\DMIN}$ are isomorphic
(Section \ref{sec:ISO}). This means that a final morphism in the
latter category still characterizes $\satbis$, but with two
important differences w.r.t.\ $\coalg{\DMIS}$. First of all, in a
final $\DMIN$-coalgebra, there are no semantically redundant
transitions. Intuitively, a final $\DMIN$-coalgebra is a universe of
\emph{abstract symbolic behaviours} and a final morphism maps each
system in its abstract symbolic behaviour. Secondly, minimization in
$\coalg{\DMIN}$ is feasible, while in $\coalg{\DMIS}$ is not,
because saturated coalgebras have all the redundant transitions.
Minimizing in $\coalg{\DMIN}$ coincides with a \emph{symbolic
minimization algorithm} that we have introduced in \cite{ESOP09}
(Section \ref{sec:Algo}).
The algorithm shows another peculiarity of normalized coalgebras:
minimization relies on the algebraic structure. Since in bialgebras
bisimilarity abstracts away from this, we can conclude that our
normalized coalgebras are not bialgebras. This is the reason why we
work with structured coalgebras.

\bigskip

As motivating examples we will show open Petri nets
\cite{Kindler,OpenPN} (Section \ref{openpetrinets}) and asynchronous
$\pi$-calculus \cite{Honda91,AmadioCONCUR96} (Section
\ref{sec:Asyn}). In \cite{FOSSACS08}, we have shown that
asynchronous bisimilarity \cite{AmadioCONCUR96} is an instance of
symbolic bisimilarity. Indeed, in the definition of asynchronous
bisimulation, the input transition $p\tr{a(b)}p'$ can be matched
either by $q\tr{a(b)}q'$ or by $q\tr{\tau}q'$. In the latter case,
the bisimulation game does not restart from $p'$ and $q'$ but from
$p'$ and $q'|\outp{a}{b}$. Thus our framework will provide, as
lateral result, also a coalgebraic model for asynchronous
bisimilarity that, as far as we know, has never been proposed so
far.

In Section \ref{sec:CIS} and \ref{secCoa} we report the framework of
\cite{FOSSACS08} and we recall the basic notions on (structured)
coalgebras. In Section \ref{sec:swc} we introduce a further example
aimed at clarifying the whole framework (by avoiding all the
technical details of open Petri nets and asynchronous $\pi$). All
proofs are in Appendix.

\paragraph{\bf{Previous works.}}
Our work relies on the framework introduced in \cite{FOSSACS08} and
on the minimization algorithm in \cite{ESOP09}. In this work we
focus on the coalgebraic characterization of them that appeared in
\cite{CALCO09}. The present paper extends \cite{CALCO09} by (1)
introducing the example of asynchronous $\pi$-calculus, (2) by
adding all the proofs, (3) by explaining in full details the
relationship with the minimization algorithm in \cite{ESOP09}.
Normalized coalgebras have been previously introduced in \cite{BM07}
for giving a coalgebraic characterization of the theory of reactive
systems by Leifer and Milner \cite{RobinCONCUR00}.

\section{Asynchronous $\pi$-calculus}\label{sec:Asyn}

\noindent Asynchronous $\pi$-calculus has been introduced in
\cite{Honda91}
for modeling distributed systems interacting via \emph{asynchronous}
message passing. Differently from the synchronous case, where
messages are sent and received at the same time, in the asynchronous
communication, messages are sent and travel through some media until
they reach the destination. Therefore sending messages is non
blocking (i.e., a process can send messages even if the receiver is
not ready to receive), while receiving is blocking (processes must
wait until the message has arrived). This asymmetry is reflected on
the observations: since sending is non blocking, receiving is
unobservable.

In this section, we introduce asynchronous $\pi$-calculus and two
definitions of bisimilarity ($\sim^1$ and $\sim^a$) that, as proved
in \cite{AmadioCONCUR96}, coincide. In Section \ref{sec:CIS}, we
will show that the first is an instance of our general definition of
saturated bisimilarity (Definition \ref{def:saturated}) while the
second of symbolic bisimilarity (Definition \ref{def:Symbolic}).

\bigskip

Let $\names$ be a set of \emph{names} (ranged over by $a,b,c \dots$)
with $\tau \notin \names$. The set of $\pi$-processes is defined by
the following grammar:
$$p::= \outp{a}{b}, \;\;  p_1| p_2 ,\;\; \res{a}{p}, \;\; !g,
\;\; m \;\;\;\;\;\;\;m::= \nil,\;\; \alpha .p,\;\;
m_1+m_2\;\;\;\;\;\;\; \alpha ::=\inp{a}{b},\;\;\tau$$
The main difference with the ordinary $\pi$-calculus \cite{Pi} is
that here output prefixes are missing. The occurrence of an
unguarded $\bar{a}b$ can be thought of as message $b$ that is
available on some communication media named $a$. This message is
received whenever it disappears, i.e., it is consumed by some
process performing an input. Thus the action of sending happens when
$\bar{a}b$ becomes unguarded.

 Considering $\inp{a}{b}.p$ and $\res{b}{p}$, the occurrences of $b$ in $p$ are
 bound. An occurrence of a name in a process is \emph{free}, if it is not
 bound. The set of \emph{free names} of $p$ (denoted by $\fn{p}$) is the set of names that have a
 free occurrence in the process $p$. The process $p$ is
 $\alpha$-equivalent to $q$ (written $p\equiv_{\alpha}q$), if they
 are equivalent up to $\alpha$-renaming of bound occurrences of
 names.
 The operational semantics of $\pi$-calculus is a transition system labeled
 on actions $Act= \{\inp{a}{b}, \outp{a}{b}, \boutb{a}{b}, \tau
 \mid a,b \in \names \}$ (ranged over by $\mu$) where $b$ is a
 \emph{bound name} (written $b \in \bn{\mu}$) in $\inp{a}{b}$ and
 $\boutb{a}{b}$. In all the other cases $a$ and $b$ are free in
 $\mu$ ($a,b \in \fn{\mu}$). By $\nm{\mu}$ we denote the set of both
 free and bound names of $\mu$.

\begin{table}[t]
 \centering
 \begin{tabular}{lll}
\rulelabel{tau} $\tau.p \tr{\tau} p$ &
\rulelabel{in} $\inp{a}{b}.p\tr{\inp{a}{c}} p\subs{c}{b}$ &
\rulelabel{out} $\outp{a}{b}\tr{\outp{a}{b}} \nil$
\\[1em]
\rulelabel{com} $\deduz{p \tr{\outp{a}{b}} p' \quad q
\tr{\inp{a}{b}} q'}{p | q \tr{\tau} p'| q'}$ &
\rulelabel{sum} $\deduz{p \tr{\mu} p'}
    {p +q \tr{\mu} p'}$&
\rulelabel{par} $\deduz{p \tr{\mu} p'}
    {p | q \tr{\mu} p' | q}$
    {\scriptsize$\bn{\mu} \cap \fn{q}= \emptyset$}
\\[1em]
    \rulelabel{opn} $\deduz{p \tr{\outp{a}{b}}p'}{\res{b}{p} \tr{\boutb{a}{b}} p'}$
    {\scriptsize$b \neq a$} &
    \rulelabel{rep} $\deduz{m | !m \tr{\mu} q}
    {!m \tr{\mu} q}$&
    \rulelabel{cls} $\deduz{p \tr{\boutb{a}{b}} p' \quad q \tr{\inp{a}{b}} q'}
    {p| q \tr{\tau} \res{b}{p' | q'}}${\scriptsize$b\notin \fn{q}$}
\\[1em]
    \rulelabel{res} $\deduz{p\tr{\mu}p'}{\res{b}{p} \tr{\mu}\res{b}{p'}}$
    {\scriptsize$b\notin \nm{\mu}$}&
    %
%
\end{tabular}
 \caption{Operational semantics of asynchronous $\pi$-calculus. \label{tableAsynchronousPiTransitions}}
 \end{table}

The labeled transition system (\lts) is inductively defined by the
rules in Table \ref{tableAsynchronousPiTransitions}, where we have
omitted the symmetric version of the rules $\textsc{sum}$,
$\textsc{par}$, $\textsc{com}$ and \textsc{cls} and where we
consider processes up to $\alpha$-equivalence, i.e., we have
implicitly assumed the rule
$$\inferrule{p\tr{\mu}q\;\;\; p\equiv_{\alpha}p'}{p'\tr{\mu}q}\text{.}$$



The main difference with the synchronous case is in the notion of
\emph{observation}. Since sending messages is non-blocking, then an
external observer can just send messages to a system without knowing
if they will be received or not. For this reason the receiving
action is not observable and the abstract semantics is defined
disregarding input transitions.


As in the case of the standard $\pi$-calculus, in the bisimulation
game we have to take care of the bound names in output actions.
Indeed, when a process $p\tr{\boutb{a}{b}}p'$, the name $b$ is
initially bound in $p$ and becomes free in $p'$. Thus, in order to
avoid name-clashes, in the bisimulation game when comparing $p$ and
$q$, we require $b$ to be \emph{fresh}, namely, different from all
the free names of $p$ and $q$. In the following definitions, by
``$\bn{\mu}$ is fresh'' we mean that if $\mu$ has a bound name, then
it is fresh.

\begin{defi}[$o\tau$-Bisimilarity]\label{defotau}
A symmetric relation $R$ is an \emph{$o\tau$-bisimulation} iff,
whenever $pRq$:
\begin{enumerate}[$\bullet$]
\item if $p\tr{\mu}p'$ where $\mu$ is not an input action and
$\bn{\mu}$ is fresh, then $\exists q'$ such that $q\tr{\mu}q'$ and
$p'Rq'$.
\end{enumerate}
We say that $p$ and $q$ are $o\tau$-bisimilar (written $p
\sim^{o\tau}q$) if and only if there exists an $o\tau$-bisimulation
relating them.
\end{defi}

Note that $\inp{a}{x}.\outp{y}{x}\sim^{o\tau}
\inp{a}{x}.\outp{d}{x}$, even if the two processes are really
different when they are put in parallel with a process
$\outp{a}{b}$.
In order to obtain an abstract semantics preserved under parallel
composition, we proceed analogously to saturated bisimilarity (that
we will show in Definition \ref{def:saturated}), i.e., at any step
of the bisimulation we put the process in parallel with all possible
outputs.
\begin{defi}[1-Bisimilarity]\label{def:1bis}
A symmetric relation $R$ is an \emph{1-bisimulation} iff, $\forall
\outp{a}{b}$, whenever $pRq$,
\begin{enumerate}[$\bullet$]
\item if $\outp{a}{b} | p\tr{\mu}p'$ where $\mu$ is not an input
action and $\bn{\mu}$ is fresh, then $\exists q'$ such that
$\outp{a}{b}| q\tr{\mu}q'$ and $p'Rq'$.
\end{enumerate}
We say that $p$ and $q$ are $1$-bisimilar (written $p \sim^{1}q$) if
and only if there exists an $1$-bisimulation relating them.
\end{defi}

The above definition is not very efficient since it considers a
quantification over all possible output in parallel.
Instead of considering all possible output contexts, we could also
consider the input actions. This leads to the following notion of
syntactic bisimulation.

\begin{defi}[Syntactic Bisimilarity]\label{def:syntacticAsy}
A symmetric relation $R$ is a \emph{syntactic bisimulation} iff,
whenever $pRq$:
\begin{enumerate}[$\bullet$]
\item if $p\tr{\mu}p'$ where
$\bn{\mu}$ is fresh, then $\exists q'$ such that $q\tr{\mu}q'$ and
$p'Rq'$.
\end{enumerate}
We say that $p$ and $q$ are syntactic bisimilar (written $p
\sim^{SYN}q$) if and only if there exists a syntactic bisimulation
relating them.
\end{defi}
Note that syntactic bisimilarity is strictly included into
$1$-bisimilarity. Indeed,
$$\tau + \inp{a}{b}.\outp{a}{b} \sim^1 \tau \text{, but } \tau + \inp{a}{b}.\outp{a}{b} \not \sim^{SYN} \tau \text{.}$$
The former equivalence can be understood by observing that both
processes can perform a $\tau$ transition in any possible context
and, when inserted into the context $ - | \outp{a}{x}$, both can
perform a $\tau$ transition going into $\outp{a}{x}$. More
generally, it holds that for all processes $p \sim^1 q\sim^1 r$:
$$\tau.p + \inp{a}{b}.(\outp{a}{b}| q) \sim^1 \tau.r$$
For instance, by taking $q=r=\nil$ and $p=\res{y}{\outp{y}{a}}$
(that is $1$-bisimilar to $\nil$, since both cannot move), we have
that $\tau.\res{y}{\outp{y}{a}} + \inp{a}{b}.\outp{a}{b} \sim^1
\tau.\nil$. Their \lts s are shown in Figure \ref{fig:processes}(A).

In order to efficiently characterize $\sim^1$, without considering
all possible contexts, we have to properly tackle the input
transitions.

\begin{defi}[Asynchronous
Bisimilarity]\label{def:AsynchronousBisimulation} A symmetric
relation $R$ is an \emph{asynchronous bisimulation} iff whenever
$pRq$,
\begin{enumerate}[$\bullet$]
\item if $ p\tr{\mu}p'$ where $\mu$ is not an input
action and $\bn{\mu}$ is fresh, then $\exists q'$ such that $
q\tr{\mu}q'$ and $p'Rq'$,
\item if
$p\tr{\inp{a}{b}}p'$, then $\exists q'$ such that either
$q\tr{\inp{a}{b}}q'$ and $p'Rq'$, or $q\tr{\tau}q'$ and $p'R(q'|
\outp{a}{b})$.
\end{enumerate}
We say that $p$ and $q$ are asynchronous bisimilar (written $p
\sim^{a} q$) if and only if there is an asynchronous bisimulation
relating them.
\end{defi}

%

\begin{figure}[t]
\begin{tabular}{c|c}
\begin{tabular}{c}
$\xymatrix@R=6pt@C=5.5pt{& & \dots \\& &  \outp{a}{b} \ar@(l,u)[rrd]^(.7){\outp{a}{b}} \\
\tau.\res{y}{\outp{y}{a}}+ \inp{a}{b}.\outp{a}{b}
\ar[rr]|(.7){\inp{a}{a}} \ar[rrd]_{\tau} \ar[rru]|(.7){\inp{a}{b}}
\ar[rruu]|(.7){\inp{a}{c}} \ar[rrd]_{\tau}& & \outp{a}{a}
\ar[rr]^{\outp{a}{a}} && \nil && \tau.\nil \ar[ll]_{\tau} \\
& & \res{y}{\outp{y}{a}} && \res{y}{\outp{y}{a}}|\outp{a}{a}
\ar[ll]_{\outp{a}{a}}}$
\\
(A)
\\
\hline
\\
$\xymatrix@R=6pt@C=5.5pt{ & & \dots \\& & \outp{a}{b}_2 \ar[rr]^{-,
\outp{a}{b}} & & \nil_2 &&\\ \tau.\res{y}{\outp{y}{a}}+
\inp{a}{b}.\outp{a}{b}_1 \ar@(ur,l)[rru]^{-|\outp{a}{b},\tau}
\ar@(u,l)[rruu]^{\dots} \ar[rr]^(.7){-|\outp{a}{a},\tau}
\ar[rrd]_{-, \tau} & & \outp{a}{a}_1
\ar[rr]^{-, \outp{a}{a}} & & \nil_1 &&  \tau.\nil_1 \ar[ll]_{-, \tau} \\
& & \res{y}{\outp{y}{a}}_1 && \res{y}{\outp{y}{a}}|\outp{a}{a}_1 \ar[ll]_{-, \outp{a}{a}} & &  \\
& & \res{y}{\outp{y}{a}}_2 && \res{y}{\outp{y}{a}}|\outp{a}{b}_2
\ar[ll]_{-,\outp{a}{b}} }$\\
(B)
\end{tabular}
&
\begin{tabular}{c}
$\xymatrix@R=6pt@C=5.5pt{
\dots & & \dots && \dots \\
\\
\outp{a}{c}_3 \ar[uu]|{\dots} \ar[uurr]|(0.4){\dots} && \outp{a}{d}_4 \ar[uu]|{\dots} \ar[uurr]|(0.4){\dots}\ar[uull]|(0.4){\dots} && \dots \\
\\
\tau.\nil_1 \ar[dd]_(0.4){-, \tau} \ar[ddrr]^{-|\outp{a}{b},\tau} \ar@(r,u)[ddrrrr]|(0.8){-|\outp{a}{b}|\outp{a}{c},\tau} \ar[rrrr]^{\dots} \ar[uu]^(0.4){-|\outp{a}{c},\tau} \ar[uurr]|{-|\outp{a}{d},\tau} \ar[uurrrr]|{\dots}&& &&\dots\\
\\
\nil_1 \ar[dd]|{\dots} \ar[ddrr]|(0.4){\dots} && \outp{a}{b}_2 \ar[ddll]|(0.4){\dots} \ar[dd]|{\dots} \ar[ddrr]|(0.4){\dots} && \outp{a}{b}|\outp{a}{c}_3\ar[dd]|{\dots} \ar[ddll]|(0.4){\dots}\\
\\
\dots && \dots && \dots }$
\\
(C)
\end{tabular}
\end{tabular}\caption{(A) Part of the infinite \lts\ of
$\tau.\res{y}{\outp{y}{a}}+ \inp{a}{b}.\outp{a}{b}$ and the \lts\ of
$\tau.\nil$. (B) The symbolic transition system $\approxAsy$ of
$\tau.\res{y}{\outp{y}{a}}+ \inp{a}{b}.\outp{a}{b}_1$ and
$\tau.\nil_1$. (C) Part of the infinite saturated transition system
of $\tau.\nil_1$.}\label{fig:processes}
\end{figure}
%

For instance, the symmetric closure of the following relation is an
asynchronous bisimulation.
$$R=\{(\tau.\res{y}{\outp{y}{a}}+ \inp{a}{b}.\outp{a}{b}, \tau.\nil),
\; (\res{y}{\outp{y}{a}}, \nil)
\}\cup\{(\outp{a}{x},\res{y}{\outp{y}{a}} | \outp{a}{x}) \mid x\in
\names\}$$ In \cite{AmadioCONCUR96}, it is proved that
$\sim^1=\sim^a$. In Section \ref{sec:CIS} we will show that this
result is an instance of a more general theorem (Theorem
\ref{theo:main}), since $\sim^1$ is an instance of saturated
bisimilarity and $\sim^a$ is an instance of symbolic bisimilarity.
The main contribute of this paper is to give coalgebraic
characterization to saturated and symbolic semantics and thus we
will characterize both $\sim^1$ and $\sim^a$ via coalgebras.
\section{Open Petri nets}\label{openpetrinets}
\noindent Differently from process calculi, Petri nets do not have a widely
known interactive behavior. Indeed they model concurrent systems
that are closed, in the sense that they do not interact with the
environment. \emph{Open nets} \cite{Kindler,OpenPN} are P/T Petri
nets \cite{PetriNets} that can interact by exchanging tokens on
\emph{input} and \emph{output places}.

%

Given a set $X$, we write $X^{\oplus}$ for the free commutative
monoid over $X$. A multiset $m\in X^{\oplus}$ is a finite function
from $X$ to $\nat$ (the set of natural numbers) that associates a
multiplicity to every element of $X$. Given two multisets $m_1$ and
$m_2$, $m_1 \oplus m_2$ is defined as $\forall x \in X$, $m_1\oplus
m_2(x)=m_1(x)+m_2(x)$. We write $\emptyset$ to denote respectively
both the empty set and the empty multiset. In order to make lighter
the notation we will use $aab$ to denote the multiset $\{a,a,b\}$.
Sometimes we will use $a^nb^m$ to denote the multisets containing
$n$ copies of $a$ and $m$ copies of $b$.

%

\begin{defi}[Open Net]
An \emph{open net} is a tuple $N=(S,T, \pr, \po, l, I,O)$ where $S$
is the set of places, $T$ is the set of transitions (with $S\cap T=
\emptyset$), $\pr,\po: T\to S^{\oplus}$ are functions mapping each
transition to its pre- and post-set, $l: T \to \Lambda$ is a
labeling function ($\Lambda$ is a set of labels) and $I, O\subseteq
S$ are the sets of input and output places. A \emph{marked open net}
(shortly, \emph{marked net}) is pair $\<\onet{N},m\>$ where
$\onet{N}$ is an open net and $m \in S^{\oplus}$ is a marking.
\end{defi}
%
\begin{figure}[t]
\centering
\begin{tabular}{c}
\epsfig{file=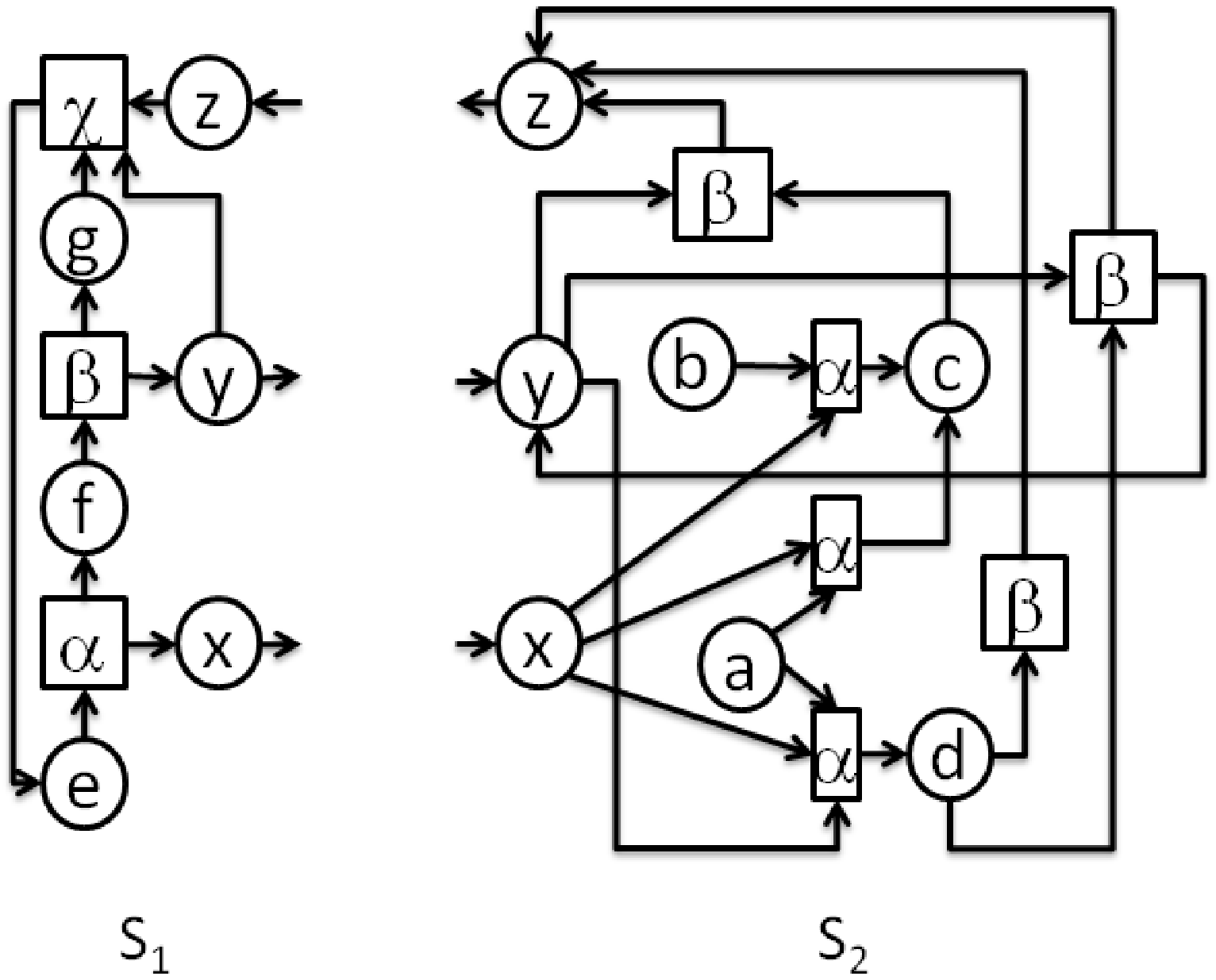, width=0.4\textwidth}
\end{tabular}
\begin{tabular}{c}
\hspace{20pt}(A)$\xymatrix@R=5pt@C=7pt{
a \ar[rr]^{+x} \ar[dd]_{+y}  & & ax \ar[rr]^{+x} \ar[dd]_{+y} \ar[rd]^{\alpha} & & axx \ar[rr]^{+x} \ar[dd]_>{+y} \ar[rd]^{\alpha} & & \dots\\
& & & c \ar[rr]^(.2){+x} \ar[dd]_{+y}& & \dots & \\
\dots & & \dots & & \dots  & & \\
&&& \dots}$\\
\hline (B)$\xymatrix@R=7pt@C=14pt{ b \ar[rr]^{x, \alpha} & & c
\ar[rr]^{y, \beta} & & z & cy \ar[l]_{\emptyset, \beta}
\\
\\
a \ar[rr]_(.6){xy, \alpha} \ar[rruu]|{x, \alpha} & & d
\ar[rruu]|{\emptyset, \beta} \ar[rr]_{y, \beta} & & zy } $
\end{tabular}
\caption{$S_1$ and $S_2$ are two open Petri nets. (A) Part of the
infinite transition system of $\<S_2,a\>$. (B) The symbolic
transition system of $\<S_2,a\>$, $\<S_2,b\>$ and
$\<S_2,cy\>$.}\label{fig:OpenNets1s2}
\end{figure}
It is worth noting that standard P/T Petri nets can be thought of as
open nets whose sets $I$ and $O$ are empty. Figure
\ref{fig:OpenNets1s2} shows two open nets where, as usual, circles
represents places and rectangles transitions (labeled with $\alpha,
\beta, \chi$). Arrows from places to transitions represent $\pr$,
while arrows from transitions to places represent $\po$. Input
places are denoted by ingoing edges, while output places are denoted
by outgoing edges. Thus in $S_1$, $x$ and $y$ are output places,
while $z$ is the only input place. In $S_2$, it is the converse. The
\emph{parallel composition} of two nets is defined by attaching them
on their input and output places. As an example, we can compose
$S_1$ and $S_2$ by attaching them through $x, y$ and $z$.

\begin{table}
 \begin{tabular}{ccc}
    \rulelabel{tr} $\deduz{t \in T \quad \lambda(t)=l\quad m =\pre{t}\oplus c}{\onet{N},m \tr{l} \onet{N},\post{t}\oplus c }$
    &
    \rulelabel{in} $\deduz{i \in I_{\onet{N}}}{\onet{N},m\tr{+i}\onet{N}, m\oplus i}$
    &
    \rulelabel{out} $\deduz{o \in O_{\onet{N}} \quad o \in m}{\onet{N},m\tr{-o}\onet{N}, m\ominus o}$
 \end{tabular}
 \caption{Operational Semantics of marked open nets. \label{tableOpenPetrinet}}
 \end{table}

%

The operational semantics of marked open nets is expressed by the
rules on Table \ref{tableOpenPetrinet}, where we use $\pre{t}$ and
$\post{t}$ to denote $\pr(t)$ and $\po(t)$ and we avoid putting
bracket around the marked net $\<N,m\>$, in order to make lighter
the notation. The rule \rulelabel{tr} is the standard rule of P/T
nets (seen as multisets rewriting), while the other two are specific
of open nets. The rule \rulelabel{in} states that in any moment a
token can be inserted inside an input place and, for this reason,
the \lts\ has always an infinite number of states. The rule
\rulelabel{out} states that when a token is in an output place, it
can be removed. Figure \ref{fig:OpenNets1s2}(A) shows part of the
infinite transition system of $\<S_2,a\>$.

The abstract semantics is defined in \cite{BaldanCEHK07} as the
standard bisimilarity (denoted by $\sim^N$) and it is a congruence
under the parallel composition outlined above. This is due to the
rules \rulelabel{in} and \rulelabel{out}, since they put a marked
net in all the possible contexts. If we consider just the rule
\rulelabel{tr}, then bisimilarity fails to be a congruence. Thus
also for open nets, the canonical definition of bisimulation
consists in inserting the system in all the possible contexts and
observing what happens.
\begin{figure}[t]
\begin{tabular}{c|c}
\begin{tabular}{c}
\epsfig{file=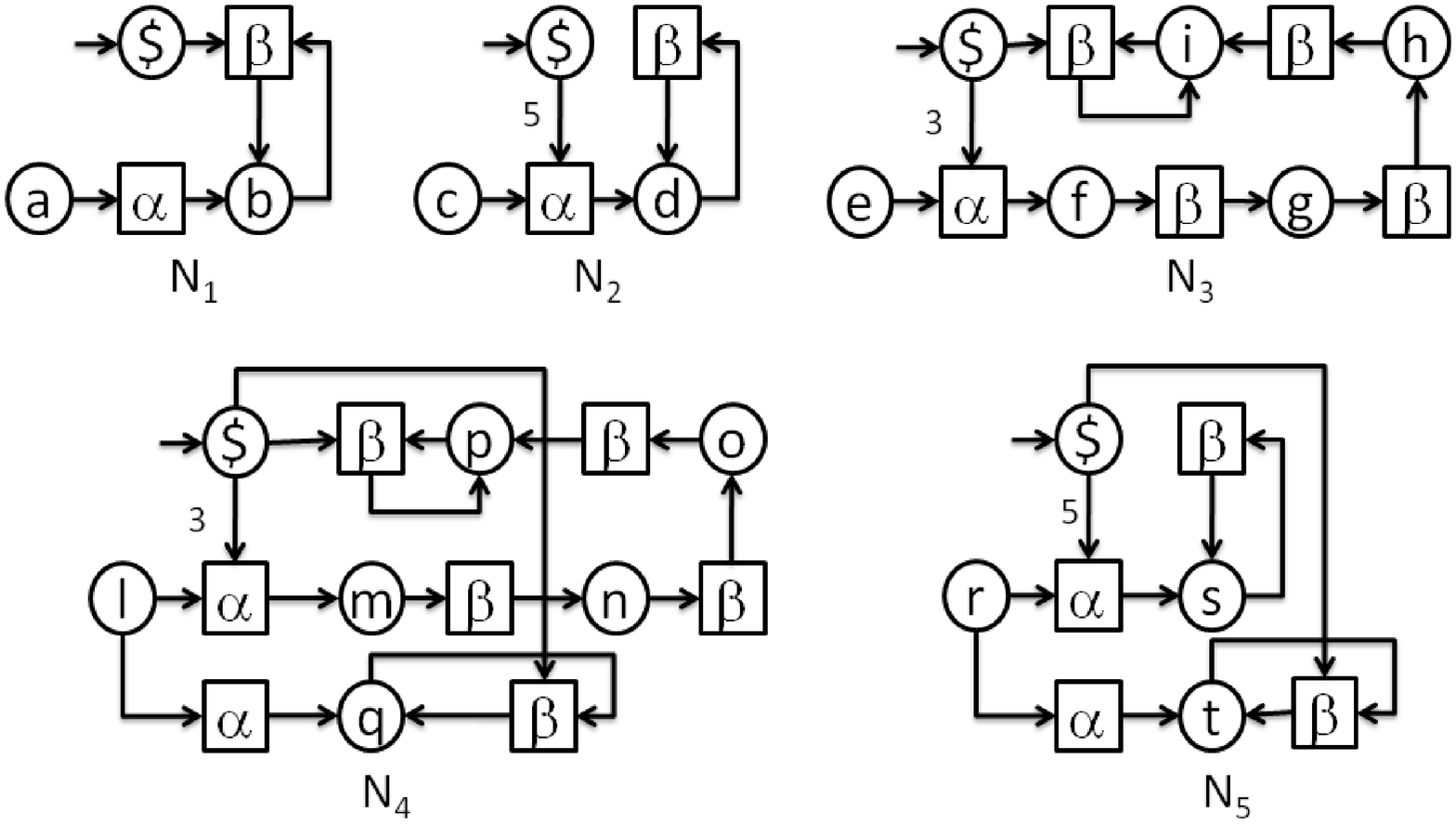, width=0.6\textwidth}
\end{tabular}&
\begin{tabular}{c}
$\xymatrix@R=5pt@C=7pt{ a \ar[rr]^{+ \$ } \ar[dd]_{\alpha}  & & a\$
\ar[rr]^{+ \$ } \ar[dd]_{\alpha}  & & a\$^2 \ar[rr]^{+\$}
\ar[dd]_{\alpha} & &
\dots\\\\
b \ar[rr]^{+ \$ } & & b\$ \ar[rr]^{+\$} \ar@/^1pc/[ll]|{\beta} & &
\dots }$\\
(A)\\
$\xymatrix@R=7pt@C=8pt{ a \ar[dd]|{\emptyset, \alpha } \ar[rrdd]|{ \$, \alpha } \ar[rrrrdd]|{\$^2, \alpha } \ar[rrrrrdd]^{\$^3, \alpha }\\\\
b \ar@(d,l)|{\$,\beta}  \ar@/_1pc/[rr]_{\$^2,\beta}
\ar@/_1pc/[rrrr]_{\$^3,\beta} \ar@/_/[rrrrr]_>>{\$^4,\beta}& & b\$
\ar[ll]|{\emptyset, \beta} && b\$^2 \ar[ll]|{\emptyset, \beta}
& \dots&}$\\
(B)
\end{tabular}
\\
\hline
\end{tabular}
\begin{center}
$\xymatrix@R=5pt@C=7pt{ a \ar[rr]^{\emptyset, \alpha} & & b
\ar@(r,u)_{\$,\beta}}$
$\xymatrix@R=5pt@C=7pt{ c \ar[rr]^{\$^5, \alpha} & & d
\ar@(r,u)_{\emptyset,\beta}}$
$\xymatrix@R=5pt@C=7pt{ e \ar[rr]^{\$^3, \alpha} & & f
\ar[rr]^{\emptyset, \beta} && g\ar[rr]^{\emptyset, \beta} && h
\ar[rr]^{\emptyset, \beta} && i \ar@(r,u)_{\$,\beta}}$
$\xymatrix@R=7pt@C=7pt{ l \ar[rr]^{\$^3, \alpha}
\ar[rrd]_{\emptyset, \alpha}& & m \ar[rr]^{\emptyset, \beta}
&&n\ar[rr]^{\emptyset, \beta} && o \ar[rr]^{\emptyset, \beta} &&p
\ar@(d,r)_{\$,\beta}\\ & & q \ar@(r,u)|{\$,\beta}}$
$\xymatrix@R=7pt@C=7pt{ r \ar@/^/[rrrr]^{\$^5, \alpha}
\ar[rrd]_{\emptyset, \alpha}& & && s \ar@(d,r)_{\emptyset,\beta}\\&
& t\ar@(r,u)|{\$,\beta}}$ (C)
\end{center} \caption{The
open nets $N_1$, $N_2$, $N_3$, $N_4$ and $N_5$. (A) Part of the
infinite transition system of $\<N_1,a\>$. (B) Part of the infinite
saturated transition system of $\<N_1,a\>$. (C) The symbolic
transition systems of
$\<N_1,a\>$,$\<N_2,c\>$,$\<N_3,e\>$,$\<N_4,l\>$ and
$\<N_5,r\>$.}\label{fig:OpenNet}
\end{figure}

\bigskip

In the remainder of the paper we will use as running example the
open nets in Figure \ref{fig:OpenNet}. Since all the places have
different names (with the exception of $\$ $), in order to make
lighter the notation, we write only the marking to mean the
corresponding marked net, e.g. $b^2\$ $ means the marked net
$\<N_1,b^2\$ \>$.

The marked net $a$ (i.e., $\<N_1,a\>$) represents a system that
provides a service $\beta$. After the activation $\alpha$, it
provides $\beta$ whenever the client pay one $\$ $ (i.e., the
environment insert a token into $\$ $). The marked net $c$ instead
requires five $\$ $ during the activation, but then provides the
service $\beta$ for free. The marked net $e$, requires three $\$ $
during the activation. For three times, the service $\beta$ is
performed for free and then it costs one $\$ $. It is easy to see
that all these marked nets are not bisimilar. Indeed, a client that
has only one $\$ $ can have the service $\beta$ only with $a$, while
a client with five $\$ $ can have the service $\beta$ for six times
only with $c$.
The marked net $r$ represents a system that offers the behaviour of
both $a$ and $c$, i.e., either the activation $\alpha$ is for free
and then the service $\beta$ costs one, or the activation costs five
and then the service is for free. Also this marked net is different
from all the others.

Now consider the marked net $l$. It offers the behaviour of both $a$
and $e$, but it is equivalent to $a$, i.e., $l \sim^N a$. Roughly,
the behaviour of $e$ is absorbed by the behaviour of $a$. This is
analogous to what happens in the asynchronous $\pi$-calculus where
it holds that $\inp{a}{x}.(\outp{a}{x}
\mid p) + \tau.p \sim^1 \tau.p$. 
%
%
%
\bigskip

The definition of $\sim^N$ involves an infinite transition system
and thus it is often hard to check. As in the case of $\sim^1$ for
the asynchronous $\pi$-calculus, we would like to efficiently
characterize it. In the following we show an efficient
characterization of $\sim^N$, that we have introduced in
\cite{FOSSACS08}. Here and in the rest of the paper, to make simpler
the presentation we restrict to open nets with only input places.
The general case, is completely analogous and can be found in
\cite{FOSSACS08,BonchiThesis}.

First of all, we have to define a \emph{symbolic transition system}
that, analogously to the operational semantics of the asynchronous
$\pi$, performs input-transitions only when needed. We call it
$\approxON$.

Intuitively, the symbolic transition $\onet{N},m
\atr{i}{\lambda}{\approxON}\onet{N},m'$ is possible if and only if
$\onet{N},m\oplus i \tr{\lambda}\onet{N},m'$ and $i$ is the smallest
multiset (on input places) allowing such transition. This transition
system is formally defined by the following rule.
$$\deduz{t \in T \quad l(t)=\lambda\quad m=(m \cap \pre{t})\oplus n
\quad i \subseteq I^{\oplus} \quad \pre{t}=(m \cap \pre{t})\oplus
i}{\onet{N},m \atr{i}{\lambda}{\approxON} \onet{N},\post{t}\oplus
n}$$ The marking $m \cap \pre{t}$ contains all the tokens of $m$
that are needed to perform the transition $t$. The marking $n$
contains all the tokens of $m$ that are not useful for performing
$t$, while the marking $i$ contains all the tokens that $m$ needs to
reach $\pre{t}$. Note that $i$ is exactly the \emph{smallest}
multiset that is needed to perform the transition $t$. Indeed if we
take $i_1$ strictly included into $i$, $m \oplus i_1$ cannot match
$\pre{t}$.
As an example consider the net $N_2$ in Figure \ref{fig:OpenNet}
with marking $cd\$^2$ and let $t$ be the only transition labeled
with $\alpha$. We have that $cd\$^2 \cap \pre{t}=c\$^2$, $n=d$ and
$i=\$^3$. Thus $N_2,cd\$^2 \atr{\$^3}{\alpha}{\approxON}N_2, dd$,
meaning that $cd\$^2$ needs $\$^3$ to perform $\alpha$ and going
into $dd$. Figure \ref{fig:OpenNet}(C) shows some symbolic
transition systems.

Note that analogously to $\sim^{SYN}$ for the asynchronous
$\pi$-calculus, the ordinary definition of bisimilarity on the
symbolic transition systems for nets, does not coincide with
$\sim^N$. Indeed the symbolic transition systems of $a$ and $l$ in
Figure \ref{fig:OpenNet}(C) are not bisimilar, but as discussed
above, $a \sim^N l$. In order to efficiently characterize $\sim^N$,
we have to introduce the following definition.

\begin{defi}[Net-symbolic Bisimilarity]\label{def:symbnet}
A symmetric relation $R$ is a \emph{net-symbolic bisimulation} iff,
whenever $\<N_1,m_1\>\;R\;\<N_2,m_2\>$:
\begin{enumerate}[$\bullet$]
\item if $\<\onet{N_1},m_1\> \atr{i}{\lambda}{\approxON}
\<\onet{N_1},m_1'\>$, then exists a marking $m_2'$ and $\exists
j,k\in I^{\oplus}$ such that:
\begin{enumerate}
\item $i=j\oplus k$,
\item$\<\onet{N_2},m_2\> \atr{j}{\lambda}{\approxON}
\<\onet{N_2},m_2'\>$ and
\item $\<\onet{N_1},m_1'\> \;R\;\<\onet{N_2},m_2' \oplus k \>$.
\end{enumerate}
\end{enumerate}
Two marked nets are net-symbolic bisimilar (written $\sim^{NS}$)
whenever there is a symbolic bisimulation relating them.
\end{defi}
For instance, the symmetric closure of the following relation is a
net-symbolic bisimulation.
$$R=\{(l,a),\; (q,b),\; (m,b\$^3),\;(n,b\$^2),\;(o,b\$),\;(p,b)\}$$
In \cite{FOSSACS08}, we have shown that $\sim^N=\sim^{NS}$. In Section \ref{sec:CIS}, we will show that the former is an instance of
saturated bisimilarity, while the latter is an instance of symbolic
bisimilarity. In Section \ref{sec:saturatedcoalgebra} and
\ref{sec:normalizedcoalgebra}, we will give a coalgebraic
characterization of both $\sim^N$ and $\sim^{NS}$ by mean of
saturated and normalized coalgebras.

\section{A Simple Words Calculus}\label{sec:swc}

\noindent In the next section we will show a theoretical framework
encompassing both asynchronous $\pi$-calculus and open Petri nets.
In this section, we introduce a \emph{simple words calculus} (swc)
as a further instance of the framework presented in the next
section. The aim of this ``toy calculus'' is to provide a more
gentle example of the concepts that will be introduced afterward, by
avoiding all the technicalities that arise with ``real formalisms''
like asynchronous $\pi$-calculus and open Petri nets.

Let $A$ be an alphabet of symbols (ranged over by $a,b,c \dots$) and
$A^*$ be the set of finite words over $A$ (ranged over by $u,v,w
\dots$). We use $\varepsilon$ to denote the empty word and $uv$ to
denote the concatenation of the words $u$ and $v$. The set of
processes is defined by the following grammar (where $u\in A^*$).
$$p::= \nil, \;\;  u.p ,\;\; p_1+p_2$$
A configuration is a pair $u\rhd p$ where $u$ is a word (in $A^*$)
representing some \emph{resources} and $p$ a process (generated by
the above grammar). The set of all configurations (ranged over by
$\gamma_1,\gamma_2, \dots$) is denoted by $Conf$. The algebra
$\WSCA$ has as carrier-set $Conf$ and as operators the words $v\in
A^*$. The function $v_{\WSCA}:Conf \to Conf$ maps each configuration
$u\rhd p$ into $uv\rhd p$. Intuitively, $v_{\WSCA}$ represents a
context where configurations can be inserted: the effect of this
insertion is that of adding $v$ (via word-concatenation) to the
resources of the configuration. This is analogous to asynchronous
$\pi$-calculus and open nets. There, resources are respectively
outputs (in parallel) and tokens (in input places). Moreover, in
those formalisms the environment can arbitrarily add new resources
(via context composition).

Differently from asynchronous $\pi$ and open nets, in swc all the
transitions are labeled with the same observation $\bullet$.
Therefore, we fix the set of \emph{observations} of swc to be
$\actwsc= \{\bullet\}$ (the subscript $\modelswc$ will be useful later to distinguish the observations of swc from those of asynchronous $\pi$ and open nets). 
The operational semantics of swc is given by
the \emph{transition relation} $\trwsc\subseteq Conf \times \actwsc
\times Conf$ defined by the following rules (together with the
symmetric one for $+$).
%
%
$$ uv \rhd u.p \tr{\bullet} uv\rhd p \qquad \deduz{u \rhd p \tr{\bullet} u\rhd p'}
    {u \rhd p + q \tr{\bullet} u \rhd p'}$$
Intuitively, the process $u.p$ needs the resources $u$ in order to
evolve. If $u$ is present in the configuration (as a suffix) then,
$u.p$ becomes $p$. Note that, differently from asynchronous $\pi$
and open nets, the resources $u$ are not consumed, but only ``read''
(we have chosen to give this read-behavior to swc, just for
simplifying the following examples).

\begin{defi}[Saturated Bisimilarity for swc]\label{def:saturatedswc}
Let $R\subseteq Conf \times Conf$ be a symmetric relation. $R$ is a
\emph{saturated bisimulation} iff, $\forall v \in A^*$, whenever
$\gamma_1 \; R\; \gamma_2$:
\begin{enumerate}[$\bullet$]
\item $v_{\WSCA}(\gamma_1)\, R \, v_{\WSCA}(\gamma_2)$,
\item if $\gamma_1 \tr{\bullet} \gamma_1'$, then $\exists \gamma_2'$ such that $\gamma_2 \tr{\bullet} \gamma_2'$ and $\gamma_1' \; R\; \gamma_2'$.
\end{enumerate}
We write $\gamma_1 \satbis \gamma_2$ iff there is a saturated
bisimulation $R$ such that $\gamma_1 \; R\; \gamma_2$.
\end{defi}

For instance, the configurations $ab \rhd ab.\nil$ and $ab \rhd
\varepsilon.\nil$ are saturated bisimilar, because for any word $v$
both $abv \rhd ab.\nil$ and $abv \rhd \varepsilon.\nil$ can only
perform one transition and then stop.
A more interesting example is the following. For all words $u,v\in
A^*$ such that $v=uw$ (i.e., $u$ is a prefix of $v$), it holds that
$$\varepsilon \rhd u.p + v.p \satbis \varepsilon \rhd u.p$$
because for any word $v'\in A^*$, $v'\rhd u.p + v.p$ and $v' \rhd
u.p$ have the same behaviour. For those $v'$ having $u$ as prefix
(i.e., $v'=uw'$), both the configurations can only perform
transitions going into $v' \rhd p$; for those $v'$ where $u$ is not
a prefix, both the configurations stop. As it happens for the
asynchronous $\pi$-calculus and open nets, the behaviour of $v.p$ is
somehow ``absorbed'' by the behaviour of $u.p$.
By joining the two previous examples, we have that:
$$\varepsilon \rhd a.ab.\nil + ab.\varepsilon.\nil \satbis \varepsilon \rhd a.ab.\nil$$
Indeed, for all the words $v'\in A^*$ having $ab$ as a prefix (i.e.,
$v'=abw'$) the configuration $abw' \rhd a.ab.\nil +
ab.\varepsilon.\nil$ can go either in $abw' \rhd ab.\nil$ or in
$abw' \rhd \varepsilon.\nil$, while the configuration $abw' \rhd
a.ab.\nil$ can only go in $abw' \rhd ab.\nil$ that, as shown in our
first example, is bisimilar to $abw' \rhd \varepsilon.\nil$. For all
the other words, the two configuration behave exactly in the same
way.

\bigskip

For simplifying the explanation, it is useful to introduce the
\emph{saturated transition system}: $u \rhd p \atr{v}{\bullet }{S}
u' \rhd p'$ iff $uv \rhd p \tr{\bullet } u' \rhd p'$. It is easy to
see that the standard notion of bisimilarity on this transition
system coincides with $\satbis$. The saturated transition systems of
$\varepsilon \rhd u.p+v.p$ and $\varepsilon \rhd u.p$ are shown in
Figure \ref{fig:wordsaturatedSymbolic}(A). For making lighter the
notation, in that figure and in the following ones we have omitted
the observation $\bullet$. Note that $\varepsilon \rhd u.p+v.p$ and
$\varepsilon \rhd u.p$ perform the same saturated transitions (and
thus they are saturated bisimilar, as discussed above).

\begin{figure}[t]
\begin{tabular}{l|r}
\begin{tabular}{c}
$\xymatrix@R=10pt@C=10pt{
\text{(A)}&& \dots\\
&& uaa \rhd p\\
&& ua \rhd p\\
\varepsilon \rhd u.p+v.p \ar[rr]^(0.6){u} \ar[rru]|{ua}
\ar[rruu]|{uaa} \ar[rruuu]^{\dots} && u \rhd p
&&  \varepsilon \rhd u.p \ar[ll]_{u} \ar[llu]|{ua} \ar[lluu]|{uaa} \ar[lluuu]_{\dots}  \\
 }$
\\
\hline \\
$\xymatrix@R=10pt@C=10pt{\text{(B)} && v \rhd p\\
\varepsilon \rhd u.p+v.p \ar[rr]^(0.6){u} \ar[rru]^{v}  && u \rhd p &&  \varepsilon \rhd u.p \ar[ll]_{u} \\
 }$
\end{tabular}
&
\begin{tabular}{c}
$\xymatrix@R=10pt@C=6pt{ \text{(C)} && \gamma_1 \ar[dd]^(.4){a}
\ar[ddll]_(.4){ab} && \gamma_2 \ar[ddll]_{a}
 \\\\
ab\rhd \varepsilon.\nil \ar[ddrr]^{\varepsilon} & & a\rhd ab.\nil
\ar[dd]^{b}  \\\\
ab \rhd ab.\nil \ar[rr]^(0.55){\varepsilon} && ab \rhd \nil \\\\ 
}$
\end{tabular}
\end{tabular}
\caption{(A) The saturated transition systems of $\varepsilon \rhd
u.p+v.p$ and $\varepsilon \rhd u.p$. (B) The symbolic transition
systems of $\varepsilon \rhd u.p+v.p$ and $\varepsilon \rhd u.p$.
(C) The symbolic transition systems of $\gamma_1=\varepsilon \rhd
a.ab.\nil + ab.\varepsilon. \nil$ and $\gamma_2=\varepsilon \rhd
a.ab.\nil$.}\label{fig:wordsaturatedSymbolic}
\end{figure}

In order to give a more efficient characterization of $\satbis$
(that avoids the quantification over all words $v\in A^*$), we
define a \emph{symbolic transition system} that, like the saturated
transition system, is labeled with pairs $v,\bullet$ (for $v\in
A^*$). The main difference is that a symbolic transition $u\rhd p
\atr{v}{\bullet}{\approxwsc} u'\rhd p'$ is performed only when $v$
is the ``minimal word'' such that $uv\rhd p \tr{\bullet} u'\rhd p'$.
The symbolic transition system $\approxwsc \subseteq Conf \times A^*
\times \actwsc \times Conf$ is defined by the following rules
(together with the symmetric rule for $+$).
$$ uv \rhd u.p \atr{\varepsilon}{\bullet}{\approxwsc} uv\rhd p \qquad u \rhd uv.p \atr{v}{\bullet}{\approxwsc} uv\rhd p \qquad \deduz{u \rhd p \atr{v}{\bullet}{\approxwsc} u'\rhd p'}
    {u \rhd p + q \atr{v}{\bullet}{\approxwsc} u' \rhd p'}$$
In the central rule, the process $uv.p$ needs the resources $uv$ to
evolve. In the configuration, there are only $u$ resources and thus
the process ``takes from the environment'' the word $v$. In the
leftmost rule, all the needed resources ($u$) are already present in
the configuration (as a prefix) and thus the process can evolve
without taking resources from the environment (i.e., by taking
$\varepsilon$). The symbolic transition systems of $\varepsilon \rhd
u.p+v.p$ and $\varepsilon \rhd u.p$ are depicted in Figure
\ref{fig:wordsaturatedSymbolic}(B). Note that the former process can
perform one symbolic transition more than the latter, even if they
perform the same saturated transitions. The symbolic transition
systems of $\gamma_1 = \varepsilon \rhd a.ab.\nil + ab.\varepsilon. \nil$ and
$\gamma_2 = \varepsilon \rhd a.ab.\nil$ are shown in Figure
\ref{fig:wordsaturatedSymbolic}(C).

Note that the standard notion of bisimilarity defined over
$\atr{v}{\bullet}{\approxwsc}$ (hereafter called \emph{syntactic
bisimilarity} and denoted by $\sim^{W}$) is strictly included into
$\satbis$. For example, $\varepsilon \rhd u.p$ and $\varepsilon \rhd
u.p +v.p$ (with $u$ prefix of $v$) are in $\satbis$ but not in
$\sim^{W}$ because $\varepsilon \rhd u.p +v.p
\atr{v}{\bullet}{\approxwsc} v \rhd p$, while $\varepsilon \rhd u.p$
only performs a symbolic transition labeled with $u$. The same holds
for $\varepsilon \rhd a.ab.\nil + ab.\varepsilon.\nil$ and
$\varepsilon \rhd a.ab.\nil$.

In order to capture $\satbis$ by exploiting the symbolic transition
system we need a more elaborated notion of bisimulation that relies
on an \emph{inference system}. For better explaining it, observe
that the following ``monotonicity property'' holds: 
\begin{center}
$\forall v \in
A^*$ and $\forall u\rhd p, u' \rhd p' \in Conf$, if $u \rhd p \tr{\bullet} u'\rhd p'$, then $uv \rhd p \tr{\bullet}
u'v \rhd p'$.
\end{center}
This property states that when adding the resources $v$ to the
original configuration (or, equivalently, when inserting the
configuration into the context $v_{\WSCA}(-)$), all the transitions
of the original configuration are preserved. This is analogous to
what happens in the asynchronous $\pi$-calculus (where putting
outputs in parallel does not inhibit any transition) and in open
Petri nets (where inserting tokens in input places does not inhibit
any transition).

An inference system is a set of rules stating properties like those
just described. For the case of swc, the inference system $\infwsc$
is defined by the following rule (parametric w.r.t.\ $v\in A^*$).
$$\deduz{\gamma  \tr{\bullet} \gamma'}{v_{\WSCA}(\gamma) \tr{\bullet} v_{\WSCA}(\gamma')}$$
This rule just states the above monotonicity property. Moreover, it
induces a \emph{derivation relation} $\vdash_{\infwsc} \subseteq
(Conf \times A^* \times \actwsc \times Conf) \times (Conf \times A^*
\times \actwsc \times Conf)$ as follows: $$\gamma \atr{v}{\bullet}{}
\gamma' \vdash_{\infwsc} \gamma \atr{vw}{\bullet}{}w_{\WSCA}
(\gamma')$$

Consider the saturated transitions of $\varepsilon \rhd u.p +v.p$ in
Figure \ref{fig:wordsaturatedSymbolic}(A) and fix
$\gamma=\varepsilon \rhd u.p +v.p$. We have that $(\gamma
\atr{u}{\bullet}{S} u \rhd p) \vdash_{\infwsc} (\gamma
\atr{ua}{\bullet}{S} ua \rhd p) \vdash_{\infwsc} (\gamma
\atr{uaa}{\bullet}{S} uaa \rhd p) \vdash_{\infwsc} \dots$ More
generally, $\forall w \in A^*$,
$$\gamma \atr{u}{\bullet}{S} u \rhd p
\vdash_{\infwsc} \gamma \atr{uw}{\bullet}{S} uw \rhd p$$ and in the
case of $\gamma = \varepsilon
\rhd u.p +v.p$ in Figure
\ref{fig:wordsaturatedSymbolic}(B), this means that
$$\gamma \atr{u}{\bullet}{\approxwsc} u \rhd p
\vdash_{\infwsc} \gamma \atr{v}{\bullet}{\approxwsc} v \rhd
p\text{.}$$

This is somehow useful to understand the causes of the mismatch
between $\satbis$ and $\sim^W$ (syntactic bisimilarity). First,
observe that symbolic transitions can derive through $\infwsc$ all
and only the saturated transitions (this will be formally shown in
the next section). Then, recall that the configurations $\varepsilon
\rhd u.p +v.p$ and $\varepsilon \rhd u.p$ are in $\satbis$ because
can perform the same saturated transitions, but they are not in
$\sim^W$ because the former can perform the symbolic transition
$\atr{v}{\bullet}{\approxwsc}$. This symbolic transition is
\emph{redundant} since it can be derived from
$\atr{u}{\bullet}{\approxwsc}$ through the inference system
$\infwsc$. More explicitly, all the saturated transitions that can
be derived from $\atr{v}{\bullet}{\approxwsc}$ can also be derived
from $\atr{u}{\bullet}{\approxwsc}$ and thus
$\atr{v}{\bullet}{\approxwsc}$ does not add any meaningful
information about the saturated behaviour of the configuration. We
can avoid this problem by employing the following notion of
bisimulation.

\begin{defi}[Symbolic Bisimilarity for swc]\label{def:symbolicswc}
Let $R\subseteq Conf \times Conf$ be a symmetric relation. $R$ is a
\emph{symbolic bisimulation} iff  whenever $\gamma_1 \; R\;
\gamma_2$:
\begin{enumerate}[$\bullet$]
\item if $\gamma_1 \atr{v}{\bullet}{\approxwsc}\gamma_1'$, then $\exists \gamma_2',\gamma_2''\in Conf, u\in A^*$ s.t. $\gamma_2 \atr{u}{\bullet}{\approxwsc} \gamma_2'$, $\gamma_2
\atr{u}{\bullet}{\approxwsc} \gamma_2' \vdash_{\infwsc} \gamma_2
\atr{v}{\bullet}{} \gamma_2''$ and $\gamma_1' \; R\; \gamma_2''$.
\end{enumerate}
We write $\gamma_1 \symbis \gamma_2$ iff there is a
symbolic bisimulation $R$ such that $\gamma_1 \; R\; \gamma_2$.
\end{defi}
%

For example $\varepsilon \rhd u.p +v.p \symbis \varepsilon \rhd u.p$
(when $v=uw$), because if $\varepsilon \rhd u.p +v.p
\atr{v}{\bullet}{\approxwsc}v\rhd p$, then $\varepsilon \rhd u.p
\atr{u}{\bullet}{\approxwsc}u\rhd p$ and this transition derives
$\varepsilon \rhd u.p \atr{uw}{\bullet}{\approxwsc}w_{\WSCA}(u\rhd
p)$ that is $\varepsilon \rhd u.p \atr{v}{\bullet}{\approxwsc}v\rhd
p$.

For an example of symbolic bisimulation, take $\gamma_1=\varepsilon
\rhd a.ab.\nil + ab.\varepsilon. \nil$ and $\gamma_2=\varepsilon
\rhd a.ab.\nil$ in Figure \ref{fig:wordsaturatedSymbolic}(C) and consider the symmetric closure of the following
relation.
$$R=\{(\gamma_1,\gamma_2),\; (a \rhd ab.\nil,a \rhd ab.\nil ),\; (ab \rhd \varepsilon.\nil, ab \rhd ab.\nil),\; (ab \rhd\nil, ab\rhd\nil)  \}$$
For the last three pairs, it is easy to check that the
configurations satisfy the above requirements. For $(\gamma_1,
\gamma_2)$, this is more interesting: the transition $\gamma_1
\atr{ab}{\bullet}{\approxwsc} ab \rhd \varepsilon. \nil$ can be
matched by $\gamma_2 \atr{a}{\bullet}{\approxwsc} a \rhd ab. \nil$
because, by definition of $\vdash_{\infwsc}$, $\gamma_2
\atr{a}{\bullet}{\approxwsc} a \rhd ab. \nil \vdash_{\infwsc}
\gamma_2 \atr{ab}{\bullet}{} ab \rhd ab. \nil$ and $(ab \rhd
\varepsilon. \nil, ab \rhd ab. \nil)\in R$.

\bigskip

In the next section we will show that $\satbis=\symbis$. Before
concluding this section, it is worth to make a final remark. The
reader would have thought that in order to retrieve $\satbis$ from
the symbolic transition system, one could just remove all the
``redundant transitions'', i.e., all those symbolic transitions
$\gamma\atr{v}{\bullet}{\approxwsc}\gamma''$ such that there exists
another symbolic transition
$\gamma\atr{u}{\bullet}{\approxwsc}\gamma'$ deriving it (in Section
\ref{sec:CSS} this removal will be called \emph{normalization}). It
is important to show that this is not enough to retrieve $\satbis$:
consider the symbolic transition systems of $\gamma_1=\varepsilon
\rhd a.ab.\nil + ab.\varepsilon. \nil$ and $\gamma_2=\varepsilon
\rhd a.ab.\nil$ shown in Figure \ref{fig:wordsaturatedSymbolic}(C).
They have no redundant transitions, but still $\gamma_1 \satbis
\gamma_2$ and $\gamma_1 \not \sim^W \gamma_2$. The transition
$\gamma_1 \atr{ab}{\bullet}{\approxwsc} ab\rhd \varepsilon. \nil$ is
not redundant, because $\gamma_1 \atr{a}{\bullet}{\approxwsc} a\rhd
ab. \nil \not \vdash_{\infwsc} \gamma_1
\atr{ab}{\bullet}{\approxwsc} ab\rhd \varepsilon. \nil$, since
$b_{\WSCA}(a\rhd ab. \nil) = ab \rhd ab.\nil \neq ab\rhd
\varepsilon. \nil$. However, it is \emph{semantically redundant},
because $\gamma_1 \atr{a}{\bullet}{\approxwsc} a\rhd ab. \nil
\vdash_{\infwsc} \gamma_1 \atr{ab}{\bullet}{\approxwsc} ab\rhd ab.
\nil$ and the states $ab \rhd ab.\nil$ and $ab\rhd \varepsilon.
\nil$ are semantically equivalent (i.e., $ab \rhd ab.\nil \satbis
ab\rhd \varepsilon. \nil$).

In order to characterize $\satbis$ through $\approxwsc$, we should
eliminate all the semantically redundant transitions, but this is
impossible without knowing a priori $\satbis$. This is the main
motivation for the introduction of \emph{normalized coalgebras} in
Section \ref{sec:CSS}.

\section{Saturated and Symbolic Semantics}\label{sec:CIS}

\noindent In Section \ref{sec:Asyn} and Section \ref{openpetrinets}, we have
introduced asynchronous $\pi$-calculus and open Petri nets. In both
cases, their abstract semantics is defined in two different ways:
either by inserting the systems into all possible contexts (like
$\sim^1$ and $\sim^N$) or by inserting the system only in those
contexts that are really needed (like $\sim^a$ and $\sim^{NS}$).
Moreover, the latter coincides with the former and thus can be
thought as an efficient characterization of the former.

This sort of ``double definition'' of the abstract semantics recurs
in many formalisms modeling interactive systems, such as mobile
ambients \cite{Ambient}, open $\pi$-calculus \cite{Sang96} and
explicit fusion calculus \cite{WischikG04}.
In \cite{FOSSACS08}, we have introduced a theoretical framework that
generalizes this ``double definition'' and encompasses all the above
mentioned formalisms. In this section we recall this framework by
employing as running examples the simple words calculus, the
asynchronous $\pi$-calculus and open
Petri nets. 

\subsection{Saturated Semantics}

Given a small category $\Cat{C}$, a $\Sig{\Gamma(\Cat{C})}$-algebra
is an algebra for the algebraic specification in Figure
\ref{fig:Sign} where $|\Cat{C}|$ denotes the set of objects of
$\Cat{C}$, $||\Cat{C}||$ the set of arrows of $\Cat{C}$ and, for all
$i,j \in |\Cat{C}|$, $\Cat{C}[i,j]$ denotes the set of arrows from
$i$ to $j$.
Thus, a $\Sig{\Gamma(\Cat{C})}$-algebra $\Alg{X}$ consists of a
$|\Cat{C}|$-sorted family $X=\{X_i \mid i \in |\Cat{C}| \}$ of sets
and a function $c_{\Alg{X}}: X_i \to X_j$ for all $c\in
\Cat{C}[i,j]$. Moreover, these functions must satisfy the equations
in Figure \ref{fig:Sign}: $id_{i_\Alg{X}}$ is the identity function
on $X_i$ and if $d;e=c$ in $\Cat{C}$, $(d;e)_{\Alg{X}}$ is equal to
$c_{\Alg{X}}$.\footnote{Note that $\Sig{\Gamma(\Cat{C})}$-algebras
coincide with functors from $\Cat{C}$ to $\set$ and
$\Sig{\Gamma(\Cat{C})}$-homomorphisms coincide with natural
transformations amongst functors. Thus,
$\Cat{Alg_{\Sig{\Gamma(\Cat{C})}}}$ is isomorphic to
$\set^{\Cat{C}}$(the category of covariant presheaves over
$\Cat{C}$).} Hereafter, we will use $\int X$ to denote the set of the
\emph{elements} of a $\Sig{\Gamma(\Cat{C})}$-algebra $\Alg{X}$,
namely, the disjoint union $\sum_{i\in |\Cat{C}|}X_i$.

The main definition of the framework presented in \cite{FOSSACS08}
is that of \emph{context interactive systems}. In our theory, an
interactive system is a state-machine that can interact with the
environment (contexts) through an evolving interface.
\begin{defi}[Context Interactive System]\label{def:cis}
A \emph{context interactive system} $\sys{I}$ is a quadruple
$\<\mssign,\Alg{X},O,tr\>$ where:
\begin{enumerate}[$\bullet$]
\item $\Cat{C}$ is a small category,
\item $\Alg{X}$ is a $\Sig{\Gamma(\Cat{C})}$-algebra,
\item $O$ is a set of observations,
\item $tr \subseteq \int X\times O \times \int X$ is a labeled transition relation ($p\tr{o}p'$ means $(p,o,p')\in tr$).
\end{enumerate}
\end{defi}
Intuitively, objects of $\Cat{C}$ are interfaces of the system,
while arrows are contexts. Every element $p$ of $X_i$ represents a
state with interface $i$ and it can be inserted into the context
$c\in \Cat{C}[i,j]$, obtaining a new state $c_{\Alg{X}}(p)$ that has
interface $j$. Every state can evolve into a new state (possibly
with different interface) producing an observation $o \in O$.

\bigskip

The abstract semantics of interactive systems is usually defined
through behavioural equivalences. In \cite{FOSSACS08} we proposed a
general notion of bisimilarity that generalizes the abstract
semantics of a large variety of formalisms
\cite{Ambient,AmadioCONCUR96,Sang96,Fusion,ExplicitFusion,CCPI}. The
idea is that two states of a system are equivalent if they are
indistinguishable from an external observer that, in any moment of
their execution, can insert them into some environment and then
observe some transitions.

\begin{defi}[Saturated Bisimilarity]\label{def:saturated}
Let $\isys$ be a context interactive system. Let $R=\{R_i \subseteq
X_i \times X_i \mid i \in \sort \}$ be a $\sort$-sorted family of
symmetric relations. $R$ is a \emph{saturated bisimulation} iff,
$\forall i,j \in \sort$, $\forall c \in \Cat{C}[i,j]$, whenever
$pR_iq$:
\begin{enumerate}[$\bullet$]
\item $c_{\Alg{X}}(p)\, R_j\, c_{\Alg{X}}(q)$,
\item if $p \tr{o} p'$ with $p'\in X_k$ for some $k\in |\Cat{C}|$, then $\exists q' \in X_k$ such that $q \tr{o}q'$ and $p' R_k q'$.
\end{enumerate}
We write $p \satbis_i q$ iff there is a saturated bisimulation $R$
such that $pR_i q$.
\end{defi}
An alternative but equivalent definition can be given by defining
the \emph{saturated transition system} $(\cts)$ as follows: $p
\atr{c}{o}{S}q$ if and only if $c_{\Alg{X}}(p)\tr{o}q$. Trivially
the ordinary bisimilarity over \cts\ coincides with $\satbis$.
\begin{prop}\label{prop:coarsest}
$\satbis$ is the coarsest bisimulation congruence.
\end{prop}

\begin{figure}[t]
\begin{tabbing}
\hspace*{10em} \= \hspace*{1em} \= \hspace*{1em} \= \hspace*{7.5em}
\= \kill
\>\textbf{specification} $\Sig{\Gamma(\Cat{C})}$ = \\
      \>\> \textbf{sorts}\\
      \>\> \> $i$ \> $\forall i \in |\Cat{C}|$\\
      \>\> \textbf{operations}\\
      \>\> \> $c: i \to j$ \> $\forall c \in \Cat{C}[i,j]$\\
      \>\> \textbf{equations}\\
      \>\> \> $id_i(x)=x$\\
      \>\> \> $e(d(x))=c(x)$ \>$\forall d;e=c $
\end{tabbing}
\caption{Algebraic specification
$\Sig{\Gamma(\Cat{C})}$.}\label{fig:Sign}
\end{figure}

%


\paragraph{\bf{A Context Interactive Systems for swc.}}
In Section \ref{sec:swc}, we have introduced a simple words
calculus. Here we show its context interactive system $\modelswc =
\modelwscAll$. Recall that $\varepsilon$ is the empty word and that
$uv$ denote the concatenation of the words $u$ and $v$. The category
$\wscmssign$ is defined as follows:
\begin{enumerate}[$\bullet$]
\item $|\wscmssign| =\{\circ\}$;
\item $\wscmssign[\circ,\circ]=A^*$;
\item $id_{\circ} = \varepsilon$;
\item $\forall u,v\in A^*$, $u;
v= uv$.
\end{enumerate}
The algebra $\WSCA$, the set of observations $\actwsc$ and the
transition relation $\trwsc$ have been already introduced in Section
\ref{sec:swc}. In swc, all the configurations have the same
interface (sort) and thus, in the category $\wscmssign$ there is
only one object. It is easy to see that saturated bisimilarity for
swc (Definition \ref{def:saturatedswc}) is an instance of Definition
\ref{def:saturated}.

\paragraph{\bf{A Context Interactive Systems for open Petri nets.}}
In the following we formally define $\modelON = \modelONAll$ that is
the context interactive system of all open nets (labeled over the
set of labels $\Lambda$). Let $Pl$ be an infinite set. We assume
that the input places of all open nets are taken from $Pl$.
Formally, we assume that if $I$ is the set of input places of an
open net $N$, then $I\in \Pow(Pl)$ (where $\Pow(Pl)$ denotes the
powerset of $Pl$).

The category $\onmssign$ is formally defined as follows:
\begin{enumerate}[$\bullet$]
\item $\onsort =\{I\,|\, I \in \Pow(Pl)\}$;
\item $\forall I, J \in \onsort$, if $I=J$ then $\onmssign[I,J]=
I^{\oplus}$ while, if $I\neq J$ then $\onmssign[I,J]= \emptyset$;
\item $\forall I \in \onsort$, $id_{I} = \emptyset$;
\item $\forall i_1,i_2\in I^{\oplus}$, $i_1;
i_2= i_1 \oplus i_2$.
\end{enumerate}

Intuitively objects are sets of places $I$. Arrows $i:I \to I$ are
multisets of tokens on $I$, while there exists no arrow $i:I \to J$
for $I\neq J$. Composition of arrows is just the sum of multisets
and, obviously, the identity arrow is the empty multiset.

We say that a marked open net $\<\onet{N},m\>$ has interface $I$ if
the set of input places of $\onet{N}$ is $I$. For example the marked
open net $\<N_1,a\>$ has interface $\{\$\}$. Let us define the
$\Sig{\Gamma(\onmssign)}$-algebra $\ONPA$. For any sort $I$, the
carrier set $\ONSET_{I}$ contains all the marked open nets with
interface $I$. For any operator $i\in \onmssign[I,I]$, the function
$i_{\Alg{N}}$ maps $\<\onet{N}, m\>$ into $\<\onet{N}, m\oplus i
\>$.

The transition structure $\ONtr$  (denoted by $\ontr{}$) associates
to a state $\<\onet{N},m\>$ the transitions obtained by using the
rule \rulelabel{tr} of Table \ref{tableOpenPetrinet}.
The saturated transition system of $\<\onet{N_1},a\>$ is shown in
Figure \ref{fig:OpenNet}(B).
\begin{prop}\label{prop:netsaturated}
Let $\<\onet{N_1},m_1\>$ and $\<\onet{N_2},m_2\>$ be two marked nets
both with interface $I$. Thus $\<\onet{N_1},m_1\> \sim^{N}
\<\onet{N_2},m_2\>$ iff $\<\onet{N_1},m_1\>
\sim^{S}_I\<\onet{N_2},m_2\>$.
\end{prop}

\paragraph{\bf{A Context Interactive System for asynchronous
$\pi$.}}\label{subsec:AsySys} We now introduce the context
interactive system $\modelAsy= \modelAsyAll$ for the asynchronous
$\pi$-calculus. First, we assume the set of names $\names$ to be in
one to one correspondence with $\nat_0$ (the set of natural numbers
$\omega$ without the number $0$). In $\modelAsy$, we use numbers in
$\omega_0$ in place of names in $\names$, but for the sake of
readability, in all the concrete examples of processes we use names
$a,b,c, \dots \in \names$ thought of as the natural numbers $1,2,3,
\dots \in \omega_0$. We need such correspondence, because we use the
well order $1<2<3\dots$. Given an $n\in \omega$, it denotes both the
number and the set of numbers in $\omega_0$ smaller or equal than
$n$. For instance, $2$ denotes both the number $2$ and the set
$\{1,2\}$ that correspond, respectively, to the name $b$ and to the
set of names $\{a,b\}$; while $0$ denotes both the number $0$ and
the empty set: the former does not correspond to any name and the
latter corresponds to the empty set of names $\emptyset$. In the
following, we will use the name in $\names$ and numbers in
$\omega_0$ interchangeably. Also, when fixed some sets $n,m \dots$
we will use $i,j$ to range over the elements of these sets.

The category of interfaces and contexts is $\amssign$, formally
defined as follows:
\begin{enumerate}[$\bullet$]
\item $|\amssign|= \nat$;
\item if $m\geq n$, then $\amssign [n,m]$ is the set of contexts generated by $c::= -, \; c |
\outp{i}{j}$, with $i,j\in m$; if $m< n$, then $\amssign
[n,m]=\emptyset$;
\item $\forall
n \in \nat$, $id_n$ is $-\in \amssign[n,n]$;
\item arrows composition is the syntactic composition of contexts.
\end{enumerate}

Note that a context could correspond to several arrows with
different sources and targets. For instance, the context $-|
\outp{1}{2}$ (corresponding to $-| \outp{a}{b}$) is, e.g., both an
an arrow $0\to 2$ and an arrow $1 \to 6$. The composition of the
arrow $-| \outp{1}{2}:0 \to 2$ with $-| \outp{3}{4}:2 \to 5$ is $-|
\outp{1}{2}|\outp{3}{4}: 0 \to 5$.

Let us define the $\aspec$-algebra $\APA$. For every object $n$,
$\asyproc_{n}$ is the set of asynchronous $\pi$-processes $p$ such
that $n \geq max\;\fn{p}$.
Intuitively in asynchronous $\pi$, interfaces are sets of names. A
process with interface $n$ uses only names in $n$ (not all, just
some). Given a process $p$ and a natural number $n\geq max\;
\fn{p}$, we denote with $p_{n}$ the process $p$ with interface $n$.
For instance, there exists several processes corresponding to
$\tau.\nil$: $\tau.\nil_0$, $\tau.\nil_1$, $\dots$ Each of these is
considered different from the others because has a different
interface. This may seem a bit strange, but is quite standard in
categorical semantics of process calculi
\cite{FioreMS02,FioreTuri,Tile} as well as in their graphical
encodings \cite{RobinCONCUR01,Gad:pi,BGK,GM:CGS}.

Extensively, $0$ is the empty interface and $A_0$ is the set of all
$\pi$-processes without free names. The set $A_1$ contains all the
processes with free names in $\{1\}$ (corresponding to $\{a\}$) and
$A_2$ contains all the processes with free names in $\{1,2\}$
(corresponding to $\{a,b\}$) and so on \dots

In order to fully define $\APA$, we still have to specify its
operations $c_{\APA}$ for all $c \in \amssign[n,m]$. Given a process
$p\in \asyproc_{n}$, $c_{\APA}(p)$ is the process with interface $m$
obtained by syntactically inserting $p$ into $c$. For instance,
$\inp{a}{x}.\outp{x}{a}_1$ can be inserted into $-|\outp{b}{c}:1\to
3$ obtaining the process $\inp{a}{x}.\outp{x}{a}|\outp{b}{c}_3$.

Note that, differently from what happens in open nets, an
asynchronous $\pi$-process can dynamically enlarge its interface by
receiving names in input or \emph{extruding} some restricted name.
Name extrusion is an essential feature of the $\pi$-calculus that can be
easily explained by looking at the rule \rulelabel{opn} in Table
\ref{tableAsynchronousPiTransitions}: the name $b$ is local (i.e.,
bound) in $\res{b}{p}$, but it becomes global (i.e., free) whenever
$p$ send it to the environment. In $\modelAsy$, we are going to
assume that processes $p_n$ with interface $n$ always extrude the name
$n+1$: this ensures that the extruded name is fresh (i.e., $n+1\notin \fn{p_n}$).

\medskip

The set of observations is $\actAsy = \{ \outp{i}{j}, \boutb{i}{},
\tau | i,j \in \omega_0 \}$. Note that the input action is not an
observation, since in the asynchronous case it is not observable.
Moreover note that in the bound output, the sent name does not
appear. This is because, any process with sort $n$ will send as
bound output the name $n+1$.


The transition structure $\trAsy$ (denoted by $\asytr{}$) is defined
by the following rules, where $i,j\in \omega_0$ represent in the
premises the corresponding names in $\names$, while in the
conclusion the numbers in $\omega_0$. Moreover the transition
relation in the premise is the one in Table
\ref{tableAsynchronousPiTransitions}.

\begin{center}
\begin{tabular}{cccc}
\inferrule{p\tr{\tau}p'}{p_{n}\asytr{\tau}p'_{n}}&
\inferrule{p\tr{\outp{i}{j}}p'}{p_{n}\asytr{\outp{i}{j}}p'_{n}} &
\inferrule{p\tr{\boutb{i}{n+1}}p'}{p_{n}\asytr{\boutb{i}{}}p'_{n+1}}
\end{tabular}
\end{center}
Note that for $\tau$ and not-bound output, $\fn{p'}\subseteq \fn{p} \subseteq n$, and thus $p'\in A_n$. 
For the case of bound ouput instead, the extruded name $n+1$ could occur free in $p'$. Thus $\fn{p'}\subseteq n+1$ and $p'\in A_{n+1}$.

In our context interactive system $\modelAsy$, processes only
perform $\tau$ and output transitions. The contexts are all the
possible outputs. Therefore is almost trivial to see that saturated
bisimilarity coincides with $\sim^1$. Figure \ref{fig:processes}(C)
shows the saturated transition system of $\tau.\nil_1$.

\begin{prop}\label{prop:asysaturated}
Let $p,q$ be asynchronous $\pi$-processes, and let $n\geq max
\;\fn{p \cup q}$. Then $p\sim^{1} q$ iff $p_{n}\satbis_{n} q_{n}$.
\end{prop}

\subsection{Symbolic Semantics}
Saturated bisimulation is a good notion of equivalence but it is
hard to check, since it involves a quantification over all contexts.
In \cite{FOSSACS08}, we have introduced a general notion of
\emph{symbolic bisimilarity} that coincides with saturated
bisimilarity, but it avoids to consider all contexts. The idea is to
define a symbolic transition system where transitions are labeled
both with the usual observation and also with the minimal context
that allows the transition. First we need to introduce \emph{context
transition systems}.
\begin{defi}[Context Transition System]
Given a category $\Cat{C}$, a $\Sig{\Gamma(\Cat{C})}$-algebra
$\Alg{X}$ and a set of observations $O$, a \emph{context transition
system} $\sym \subseteq \int X\times ||\Cat{C}|| \times O \times
\int X$ is a transition relation labeled with $||\Cat{C}|| \times O$
($p\atr{c}{o}{\sym}p'$ means that $(p,c,o,p')\in \sym$).
\end{defi}
An example of context transition system is $\eta$ defined in Section
\ref{openpetrinets}: each transition is labeled with both  a
multiset of tokens $i$ and an observation $\lambda$.
%
Also the saturated transition system is a context transition
systems. Hereafter, given a context transition system $\beta$, we
will write $\atr{c}{o}{\beta}$ to denote the transitions of $\beta$,
$\atr{c}{o}{S}$ to denote the saturated transitions and
$\atr{c}{o}{}$ (without subscript) to denote the transitions of the
\emph{total context transition system} $t = \int X\times ||\Cat{C}||
\times O \times \int X$.

\begin{defi}[Inference System]
Given a category $\Cat{C}$, a $\Sig{\Gamma(\Cat{C})}$-algebra
$\Alg{X}$ and a set of observations $O$, an inference system
$\infsys$ is a set of rules of the following format, where $i,j\in
\sort$, $o,o'\in O$, $c\in \Cat{C}[i,i']$ and $d \in \Cat{C}[j,j']$.
$$\deduz{p_i\tr{o}q_j}{c(p_i)\tr{o'}d(q_j)}$$
\end{defi}\medskip

\noindent In this rule, $i$, $j$, $o$, $o'$, $c$ and $d$ are constants, while
$p_i$ and $q_j$ are variables ranging over $X_i$ and $X_j$,
respectively. Therefore, the above rule states that all processes
with interface $i$ that perform a transition with observation $o$
going into a state $q_j$ with interface $j$, when inserted into the
context $c$ can perform a transition with the observation $o'$ going
into $d(q_j)$. In other words, this rule is in a (multisorted) SOS
format, where the operators (here, contexts) are unary and there is
only one transition in the premise of the rules. Note that, however,
this kind of rules is not intended to be used for expressing the
operational semantics of a formalism (as in the case of SOS), but instead for describing
``useful properties'' about how contexts modify the behaviour of
systems.

In the following, we write $c\UTtr{o}{o'}d$ to mean a rule like the
above. The rules $c\UTtr{o}{o'}c'$ and $d\UTtr{o'}{o''}d'$
\emph{derive} the rule $ c;d\UTtr{o}{o''}c';d'$ if $ c;d$ and
$c';d'$ are defined. Given an inference system $\infsys$,
$\Phi(\infsys)$ is the set of all the rules derivable from $\infsys$
together with the identities rules ($\forall o \in O$ and $\forall
i,j \in |\Cat{C}|$, $id_i \UTtr{o}{o}id_j$).

\begin{defi}[Derivations]\label{def:der}
Let $\Cat{C}$ be a category, $\Alg{X}$ be a
$\Sig{\Gamma(\Cat{C})}$-algebra, $O$ be a set of observations. An
inference system $\infsys$ defines a \emph{derivation relation}
$\vdash_{\infsys}\subseteq t \times t$ amongst the transitions of
the total context transition system.

We say that $p\atr{c_1}{o_1}{}p_1$ \emph{derives}
$p\atr{c_2}{o_2}{}p_2$ (written $p\atr{c_1}{o_1}{}p_1
\vdash_{\infsys} p\atr{c_2}{o_2}{}p_2$) if there exist $d,e \in
||\Cat{C}||$ such that $d\UTtr{o_1}{o_2}e\in \Phi(\infsys)$,
$c_1;d=c_2$ and $e_{\Alg{X}}(p_1)=p_2$.
\end{defi}

Note that the above definition can be extended to the transitions of
any pairs of context transition systems $\beta_1, \beta_2$:
$p\atr{c_1}{o_1}{\beta_1}p_1 \vdash_{\infsys}
p\atr{c_2}{o_2}{\beta_2}p_2$ iff $p\atr{c_1}{o_1}{}p_1
\vdash_{\infsys} p\atr{c_2}{o_2}{}p_2$.


Until now, context transition systems and inference systems are not
related with the transitions relations $tr$ of context interactive
systems. The following definition makes a link between them.
\begin{defi}[Soundness and Completeness]
Let $\isys$ be a context interactive system, $\beta$ a context
transition system and $\infsys$ an inference system.

We say that $\beta$ and $\infsys$ are \emph{sound} w.r.t.\ $\sys{I}$
iff  \begin{center}if $p\atr{c'}{o'}{\beta}q'$ and
$p\atr{c'}{o'}{\beta}q' \vdash_{\infsys} p\atr{c}{o}{}q$, then
$p\atr{c}{o}{S}q$.\end{center}

We say that $\beta$ and $\infsys$ are \emph{complete} w.r.t.\
$\sys{I}$ iff \begin{center} if $p\atr{c}{o}{S}q$, then there exists
$p\atr{c'}{o'}{\beta}q'$ such that $p\atr{c'}{o'}{\beta}q'
\vdash_{\infsys} p\atr{c}{o}{}q$. \end{center}
\end{defi}

\begin{defi}
Let $\isys$ be a context interactive system, $\beta$ a context
transition system and $\infsys$ an inference system. If $\beta$ and
$\infsys$ are sound and complete w.r.t.\ $\sys{I}$ we say that
$\beta$ is a \emph{symbolic transition system} (\sts\ for short) for
$\sys{I}$.
\end{defi}
For instance, the saturated transition system $\approxON$ (defined
in Section \ref{openpetrinets} for open nets) is a symbolic
transition system (this will be formally stated in Proposition
\ref{prop:symboliOpenNet}). Also the saturated transition system is
a symbolic transition system (take as $T$ the empty inference
system), while the total context transition system is usually not
sound.

A symbolic transition system could be considerably smaller than the
saturated transition system, but still containing all the
information needed to recover $\satbis$. Note that the ordinary
bisimilarity over \sts\ (hereafter called \emph{syntactic
bisimilarity} and denoted by $\sim^{W}$) is usually strictly
included in $\satbis$. As an example consider the marked open nets
$a$ and $l$. These are not syntactically bisimilar, since
$l\atr{\$^3}{\alpha}{\approxON}m$ while $a$ cannot (Figure
\ref{fig:OpenNet}(C)). However, they are saturated bisimilar, since
$\sim^S=\sim^N$. Analogously, the ordinary bisimilarity over the
\lts\ of the asynchronous $\pi$ does not coincide with $\sim^1$:
$\inp{a}{b}.\outp{a}{b} + \tau$ and $\tau$ are $1$-bisimilar, but
not syntactically bisimilar (at the end of this section, we will
show that also the transition system of asynchronous $\pi$ in Table
\ref{tableAsynchronousPiTransitions} is somehow a \sts ).

In literature, several $\sts$ are defined in
\cite{Sang96,Fusion,ExplicitFusion}. In these works, transitions are
labeled with both ``fusions'' of names and the ordinary labels.
Other noteworthy examples are the IPOs and the borrowed contexts of
\cite{RobinCONCUR00} and \cite{BARBARAFOSSACS04}: here all the
transitions are labeled only with the minimal contexts and the
observations can be though as $\tau$s. Also in all these cases,
syntactic bisimilarity is too fine grained.
In order to recover $\sim^S$ through the symbolic transition system
we need a more elaborated definition of bisimulation.

\begin{defi}[Symbolic Bisimilarity]\label{def:Symbolic}
Let $\isys$ be an interactive system, $\infsys$ be a set of rules
and $\sym$ be a context transition system. Let $R=\{R_i \subseteq
X_i \times X_i \mid i \in \sort\}$ be a $\sort$-sorted family of
symmetric relations. $R$ is a \emph{symbolic bisimulation} iff
$\forall i \in \sort$, whenever $pR_iq$:
\begin{enumerate}[$\bullet$]
\item if $p \atr{c}{o}{\sym} p'$,
then $\exists c_1, o_1, q_1', q'$ such that
$q\atr{c_1}{o_1}{\sym}q_1'$ and
$q\atr{c_1}{o_1}{\sym}q_1'\vdash_{\infsys} q\atr{c}{o}{}q'$ and $p'
R_k q'$.
\end{enumerate}
We write $p \symbis_i q$ iff there exists a symbolic bisimulation
$R$ such that $p R_i q$.
\end{defi}

\begin{thm}\label{theo:main}
Let $\sys{I}$ be a context interactive system, $\beta$ a context
transition system and $\infsys$ an inference system. If $\beta$ and
$\infsys$ are sound and complete w.r.t.\ $\sys{I}$, then $\symbis =
\satbis$.
\end{thm}

\paragraph{\bf{Symbolic Semantics for swc.}}
The symbolic transition system $\approxwsc$ and the inference system
$\infwsc$ for swc have already been defined in Section
\ref{sec:swc}. It is also easy to see that symbolic bisimilarity for
swc (Definition \ref{def:symbolicswc}) is an instance of Definition
\ref{def:Symbolic}. Therefore, in order to apply Theorem
\ref{theo:main}, we only need to prove that $\approxwsc$ and
$\infwsc$ are sound and complete.

\begin{prop}\label{prop:wordsoundcomplete}
$\approxwsc$ and $\infwsc$ are sound and complete w.r.t.\
$\modelswc$.
\end{prop}
\begin{cor}[From Theorem \ref{theo:main}]
In swc, $\satbis=\symbis$.
\end{cor}

%

\paragraph{\bf{Symbolic Semantics for open Petri nets.}}
The symbolic transition system for open Petri nets is $\eta$ defined
in Section \ref{openpetrinets}. The inference system $\infON$ is
defined by the following rule parametric w.r.t.\ $\lambda \in
\Lambda$, $I \in \Pow(Pl)$ and $i\in I^{\oplus}$.
$$\deduz{\onet{N},m \tr{\lambda} \onet{N},m'}{\onet{N},m\oplus i \tr{\lambda} \onet{N},m'\oplus i}$$
Its intuitive meaning is that for all possible observations
$\lambda$ and multiset $i$ on input places, if a marked net performs
a transition with observation $\lambda$, then the addition of $i$
preserves this transition.

Now, consider derivations between transitions of open nets. It is
easy to see that
$N,m\atr{i_1}{\lambda_1}{}N,m_1\vdash_{\infON}N,m\atr{i_2}{\lambda_2}{}N,m_2$
if and only if $\lambda_2=\lambda_1$ and there exists a multiset $x$
on the input places of $N$ such that $i_2=i_1\oplus x$ and $m_2=
m_1\oplus x$. For all the nets $N_k$ of Figure \ref{fig:OpenNet},
this just means that for all observations $\lambda$ and for all
multisets $m,n$, we have that
$\<N_k,m\>\atr{\$^i}{\lambda}{\approxON}\<N_k,n\>\vdash_{\infON}\<N_k,m\>\atr{\$^{i+j}}{\lambda}{}\<N_k,n\$^j\>$.
From this observation, it is easy to see that the definition of
net-symbolic bisimilarity is an instance of symbolic bisimilarity.
\begin{prop}\label{prop:netSymbolic}
Let $\<\onet{N_1},m_1\>$ and $\<\onet{N_2},m_2\>$ be two marked nets
both with interface $I$. Thus $\<\onet{N_1},m_1\> \sim^{NS}
\<\onet{N_2},m_2\>$ iff $\<\onet{N_1},m_1\>
\sim^{SYM}_I\<\onet{N_2},m_2\>$.
\end{prop}
Thus, in order to prove that $\sim^N=\sim^{NS}$, we have only to
prove that $\approxON$ and $\infON$ are sound and complete w.r.t.\
$\modelON$ and then apply the general Theorem \ref{theo:main}.

\begin{prop}\label{prop:symboliOpenNet}
$\approxON$ and $\infON$ are sound and complete w.r.t.\ $\modelON$.
\end{prop}

\begin{cor}[From Theorem \ref{theo:main}]
$\sim^N=\sim^{NS}$.
\end{cor}

\paragraph{\bf{Symbolic Semantics for asynchronous
$\pi$.}}\label{sec:SymbolicAsynchronous}
In the case of asynchronous $\pi$-calculus, the ordinary \lts\
closely corresponds to the \sts\ that we are going to introduce. The
transitions labeled with an input $\inp{a}{b}$ are substantially
transitions saying that if the process is inserted into $- |
\outp{a}{b}$, then it can perform a $\tau$. The symbolic transition
system $\approxAsy$ for the asynchronous $\pi$-calculus is defined
by the following rules, where in the premises there are standard
transitions (from Table \ref{tableAsynchronousPiTransitions}),
$i,j\in \omega_0$ represent in the premises the corresponding names
in $\names$, while in the conclusion the numbers in $\omega_0$ and
$- \in \amssign[n,n]$ and $-| \outp{i}{m}\in \amssign[n,n']$.

\begin{center}
$\deduz{p\tr{\tau}p'}{p_n
\atr{-}{\tau}{\approxAsy}p'_n}$
$\deduz{p\tr{\outp{i}{j}}p'}{p_n
\atr{-}{\outp{i}{j}}{\approxAsy}p'_n}$
$\deduz{p\tr{\boutb{i}{n+1}}p'}{p_n
\atr{-}{\boutb{i}{}}{\approxAsy}p'_{n+1}}$  $\deduz{p\tr{\inp{i}{m}}p'\quad n'=
max \{m,n\}}{p_n \atr{-| \outp{i}{m}}{\tau}{\approxAsy}p'_{n'}}$
\end{center}\medskip

\noindent Note that the only non standard rule is the fourth. If, in the
standard transition system a process can perform an input, in the
\sts\ the same process can perform a $\tau$, provided that there is
an output process in parallel. Note that the interface of the
arriving state depends on the received name $m$: if it is smaller
than $n$, then the arriving interface is still $n$, otherwise it is
extended to $m$ (i.e., $max \;\fn{p'} \leq max\; \{m,n\}=n'$).

Part of the \sts\ of $\tau.\res{y}{\outp{y}{a}}+
\inp{a}{b}.\outp{a}{b}_1$ and $\tau.\nil_1$ are shown in Figure
\ref{fig:processes}(B). There and in the following we avoid to
specify the source and the target of the contexts labelling the
transitions, since these can be inferred by the sorts of starting
and arriving states. As well as the ordinary \lts, the symbolic
transition system is infinite, because the input can receive any
possible name in $\names$. It is well known that, instead of
considering all possible input names, it is enough to consider only
the free names and one fresh name (all the other fresh are useless).
By slightly modifying the general definition of the context
interactive system $\modelAsy$, we could have defined a symbolic
context transition system that only receive in input those names
that are strictly needed. We have made a different choice for the
following reasons: (a) the presentation of this modified context
interactive system is a bit more contrived; (b) the actual
presentation is mainly aimed at showing how an input transition
``can be matched'' by a $\tau$ transition (instead of focusing on
finite representation); (c) there exists several other sources of
infiniteness (discussed in Section \ref{sec:conclusion}) that cannot
be trivially tackled by our framework.

\bigskip

Let us define an inference system $\infAsy$ that describes how
contexts transform transitions. Since our contexts are just parallel
outputs, all the contexts preserve transitions. This is expressed by
the following rules parametric w.r.t.\ $n,m\in\omega$, $i,j\in n$,
$c\in \amssign[n,m]$.
\begin{center}
\begin{tabular}{ccc}
\rulelabel{tau$_c$} $\deduz{p_n\tr{\tau}q_n}{c(p_n)\tr{\tau}c(q_n)}$
&
\rulelabel{out$_c$}$\deduz{p_n\tr{\outp{i}{j}}q_n}{c(p_n)\tr{\outp{i}{j}}c(q_n)}$&
\rulelabel{bout$_c$}$\deduz{p_n\tr{\boutb{i}{}}q_{n+1}}{c(p_n)\tr{\boutb{i}{}}c^{+1}(q_{n+1})}$
\end{tabular}
%
\end{center}
Here, $c^{+1}\in \amssign[n+1,m+1]$ is the same syntactic context as
$c$, but with different interfaces.

Derivations amongst transitions of asynchronous $\pi$-processes are
quite analogous to those amongst open Petri nets. Particularly
relevant is the following kind of derivation: for all processes
$p_n,q_n$, for all names $i\in n$ and $j\in m$,
$$p_n \atr{-}{\tau}{\approxAsy}q_n\vdash_{\infAsy} p_n \atr{-|
\outp{i}{j}}{\tau}{}q| \outp{i}{j}_m\text{.}$$
Intuitively, this means that in the original \lts, the $\tau$
transitions derive the input transitions.
Instantiating the general definition of symbolic bisimulation to
$\approxAsy$ and $\infAsy$, we retrieve the definition of
asynchronous bisimulation. Indeed transitions of the form
$p\atr{-}{\mu}{\approxAsy}p'$ (in the original \lts, these
correspond to $\tau$ and output), can be matched only by transitions
with the same label, since the context $-$ is not decomposable.

The transitions $p\atr{-| \outp{i}{j}}{\tau}{\approxAsy}p'$
(corresponding to the input in the original \lts) can be matched
either by $q\atr{-| \outp{i}{j}}{\tau}{\approxAsy}q'$, or by
$q\atr{-}{\tau}{\approxAsy}q'$. In other words, when $p \atr{-|
\outp{i}{j}}{\tau}{\approxAsy}p'$, then $q$ can answer with
$q\atr{- }{\tau}{\approxAsy}q'$, since $q\atr{-|
}{\tau}{\approxAsy}q' \vdash_{\infAsy} q \atr{-|
\outp{i}{j}}{\tau}{}q' | \outp{i}{j}$.

\begin{prop}\label{prop:AsynSymbolic}
Let $p,q$ be asynchronous $\pi$-processes, and let $n\geq max
\;\fn{p \cup q}$. Then $p\sim^{a} q$ iff $p_{n}\symbis_{n} q_{n}$.
\end{prop}

Therefore $\sim^1$ is the saturated bisimulation for $\modelAsy$,
while $\sim^a$ is its the symbolic version. We can employ our
general Theorem \ref{theo:main} to prove that $\sim^1=\sim^a$ by
showing that the \sts\ $\approxAsy$ and the inference system
$\infAsy$ are sound and complete w.r.t.\ $\modelAsy$.

\begin{prop}\label{prop:AsynchronousSoundandComplete}
$\approxAsy$ and $\infAsy$ are sound and complete w.r.t.\
$\modelAsy$.
\end{prop}
%

\begin{cor}[By Theorem. \ref{theo:main}]
$\sim^1=\sim^a$ as shown in \cite{AmadioCONCUR96}.
\end{cor}


%
\section{(Structured) Coalgebras}\label{secCoa}
\noindent In this section we recall the basic notions of the theory of
coalgebras and the coalgebraic characterization of labeled
transition systems and bisimilarity.
\begin{defi}[Coalgebra]
\label{def:coalgebra} Let $\Fun{B}: \Cat{C} \to \Cat{C}$ be an
endofunctor on a category $\Cat{C}$. A {\em $\Fun{B}$-coalgebra\/}
is a pair  $\langle X, \alpha \>$ where $X$ is an object of
$\Cat{C}$ and $\alpha: X \to \Fun{B}(X)$ is an arrow. A {\em
$\Fun{B}$-morphism\/} $f: \< X, \alpha \> \to \<Y, \beta\>$ is an
arrow $f: X \to Y$ of $\Cat{C}$ such that the following diagram
commutes. $\Fun{B}$-coalgebras and $\Fun{B}$-morphisms form the
category $\coalg{B}$.
$$\xymatrix{ X \ar[d]_{\alpha} \ar[r]^f & Y \ar[d]^{\beta}\\
\Fun{B}(X) \ar[r]_{\Fun{B}(f)} & \Fun{B}(Y)}$$
\end{defi}
For instance, labeled transition systems with labels in
$L$ are coalgebras for the functor $\Pow(L\times Id):\set\to \set$,
where $\set$ denotes the category of sets and functions. This functor maps each set $X$ into the set $\Pow(L \times X)$ 
(i.e., the powerset of $L\times X$) and each function $f:X \to Y$ into  $\Pow(L\times f): \Pow(L\times X) \to  \Pow(L\times Y)$ 
that, for all $A \in \Pow(L\times X)$, is defined as $\Pow(L\times f)(A)=\{(l,f(x)) \text{ s.t. } (l,x)\in A\}$. Concretely, a \lts\ is
a set of states $X$ together with a transition function
$\alpha:X\to\Pow(L\times X)$ mapping each state into a set of pairs
$(l,x)$ representing transitions with labels $l\in L$ and next state
$x\in X$.
A $\Pow(L \times Id)$-morphism is a ``zig-zag'' morphism, i.e., a
function between the sets of states that both preserves and reflects
the transitions.

We can think of symbolic transition systems as ordinary $\Pow(L
\times Id)$-coalgebras where the labels in $L$ are pairs $(c,o)$
(for $c$ a contexts, and $o$ an observation), but this
representation is somehow inadequate. Figure \ref{fig:wordMorphism}
shows a function between the states space of two $\Pow(L\times
Id)$-coalgebras. This is not a $\Pow(L\times Id)$-morphism since the
transition $\gamma_1 \atr{ab}{\bullet}{} ab \rhd \varepsilon.\nil$
is not preserved. The same holds for the morphisms in Figure
\ref{fig:morphism}: these are not $\Pow(L\times Id)$-morphisms since
the transitions $l\atr{\$^3}{\alpha}{}m$ and
$\tau.\res{y}{\outp{y}{a}}+
\inp{a}{b}.\outp{a}{b}_1\atr{-|\outp{a}{a}}{\tau}{} \outp{a}{a}_1$
are not preserved. In Section \ref{sec:CSS}, we will show the
category of normalized coalgebras where these maps are morphisms.

\begin{figure}[t]
\begin{tabular}{c}
$\xymatrix@R=10pt@C=6pt{
\gamma_2 \ar[rrd]^(.4){a} \ar@(ld,lu)@{.>}[ddddd] \\
\gamma_1 \ar[rr]^(.4){a} \ar[rrd]_(.4){ab} \ar@(ld,lu)@{.>}[dddd]& &
a\rhd ab.\nil \ar[rr]^{b} \ar@(ld,lu)@{.>}[dddd] &&
ab \rhd \nil \ar@{.>}[dddd] &&  ab \rhd ab.\nil \ar[ll]_{\varepsilon} \ar@{.>}[dddd] \\
& & ab\rhd \varepsilon.\nil \ar[rru]_{\varepsilon} \ar@{.>}[dddrrrr] \\
\\
\\
\gamma \ar[rr]^(.4){a} & &  \gamma' \ar[rr]^{b} && \gamma'' &&
b_{\Alg{U}}(\gamma') \ar[ll]_{\varepsilon} }$
\end{tabular}\caption{The dotted arrows represent a map between the states space
of two transition systems. It is not a $\Pow(L\times Id)$-morphism
but it is a morphism in the category of normalized
coalgebra.}\label{fig:wordMorphism}
\end{figure}

\begin{figure}[t]
\begin{center}
\begin{tabular}{c}
$\xymatrix@R=6pt@C=3pt{ l \ar[rr]^{\$^3, \alpha}
\ar[rrd]_{\emptyset, \alpha} \ar@{.>}@(u,u)[rrrrrrrrrrrr]&& m
\ar[rr]^{\emptyset, \beta} \ar@{.>}@(u,u)[rrrrrrrrrrrrrrrrrr] &&
n\$^2\ar[rr]^{\emptyset, \beta} \ar@{.>}@(ur,ul)[rrrrrrrrrrrrrr] &&
o\$ \ar[rr]^{\emptyset, \beta} \ar@{.>}@(ur,ul)[rrrrrrrrrr] &&p
\ar@(r,u)|{\$,\beta} \ar@{.>}@(r,dl)[rrrrrr]&&  && u
\ar[rr]^{\emptyset, \alpha} & & v \ar@(r,u)_{\$,\beta} && w
\ar[ll]^{\emptyset, \beta} & & x \ar[ll]^{\emptyset, \beta} & &
\$^3_{\Alg{Y}}(v) \ar[ll]^{\emptyset, \beta}\\ && q
\ar@(r,d)|{\$,\beta} \ar@{.>}@(ur,dl)[rrrrrrrrrrrru]&& q\$^1
\ar[ll]_{\emptyset,\beta}\ar@{.>}@(dr,d)[rrrrrrrrrrrru] && q\$^2
\ar[ll]_{\emptyset,\beta} \ar@{.>}@(dr,d)[rrrrrrrrrrrru]&& q\$^3
\ar[ll]_{\emptyset,\beta} \ar@{.>}@(r,d)[rrrrrrrrrrrru]&& }$

\\
\hline
\\
$\xymatrix@R=6pt@C=5.5pt{
&& \dots && &&  \\
&& \outp{a}{b}_2 \ar[rr]^{-, \outp{a}{b}} \ar@{.>}@(ur,u)[rrrrrrrr] && \nil_2 \ar@{.>}[rrrr]&&    && q_2 && -|\outp{a}{b}_{\Alg{Z}}(q_1) \ar[ll]_{-, \outp{a}{b}} \\
\tau.\res{y}{\outp{y}{a}}+ \inp{a}{b}.\outp{a}{b}_1
\ar@(ur,l)[rru]^{-|\outp{a}{b},\tau} \ar@(u,l)[rruu]^{\dots}
\ar[rr]^(.7){-|\outp{a}{a},\tau} \ar[rrd]_{-, \tau}
\ar@(l,dr)@{.>}[rrrrrr] && \outp{a}{a}_1 \ar@{.>}@(ur,u)[rrrrrrrrd]
\ar[rr]^{-, \outp{a}{a}} && \nil_1 \ar@{.>}@(r,l)[rrrrd] &&  p_1 \ar[rrd]_{-, \tau} && \\
&& \res{y}{\outp{y}{a}}_1 \ar@(l,dr)@{.>}[rrrrrr] && \res{y}{\outp{y}{a}} | \outp{a}{a}_1 \ar[ll]_{-, \outp{a}{a}} \ar@(l,dr)@{.>}[rrrrrr] &&   && q_1 && -|\outp{a}{a}_{\Alg{Z}}(q_1) \ar[ll]_{-, \outp{a}{a}}  \\
&& \res{y}{\outp{y}{a}}_2 \ar@(rd,d)@{.>}[uuurrrrrr] &&
\res{y}{\outp{y}{a}}| \outp{a}{b}_2 \ar[ll]_{-,\outp{a}{b}}
\ar@(r,d)@{.>}[uuurrrrrr]&&   }$
\end{tabular}
\end{center}\caption{
The dotted arrows represent maps between the states space of
transition systems. Both are not $\Pow(L\times Id)$-morphisms, but
they are morphisms in the category of normalized
coalgebras.}\label{fig:morphism}
\end{figure}

\bigskip

Under certain conditions, $\coalg{B}$ has a \emph{final
coalgebra} (unique up to isomorphism) into which every
$\Fun{B}$-coalgebra can be mapped via a unique $\Fun{B}$-morphism.
The final coalgebra can be viewed as the universe of all possible
$\Fun{B}$-\emph{behaviours}: the unique morphism into the final
coalgebra maps every state of a coalgebra to a canonical
representative of its behaviour. This provides a general notion of
behavioural equivalence (hereafter referred to as bisimilarity): two
$\Fun{B}$-coalgebras are $\Fun{B}$-equivalent iff they are mapped to
the same element of the final coalgebra. Moreover, the image of a
coalgebra through the final morphism is its minimal realization
w.r.t.\ bisimilarity. In the finite case, this can be done via a
minimization algorithm, that for \lts s coincides with
\cite{KannelakisSmolka}.

Unfortunately, due to cardinality reasons, $\Pow(L \times Id)$ does
not have a final object \cite{Rut96}. One satisfactory solution
consists in replacing the powerset functor $\Pow$ by the {\em
countable} powerset functor $\Pc$, which maps a set to the family of
its countable subsets. Then, $\Pc(L \times Id)$-coalgebras are
one-to-one with transition systems with \emph{countable degree}.
Unlike the functor $\Pow(L \times Id)$, the functor $\Pc(L \times
Id)$ admits final coalgebras (Example 6.8 of \cite{Rut96}).

\bigskip

The coalgebraic representation using functor $\Pc(L \times Id)$ is
not completely satisfactory, because the intrinsic algebraic
structure of the states is lost. This calls for the introduction of {\em
structured coalgebras} \cite{CGH01}, i.e., coalgebras for an
endofuctor on a category $\alge{\Gamma}$ of algebras for a
specification $\Gamma$. Since morphisms in a category of structured
coalgebras are also $\Gamma$-homomorphisms, \emph{bisimilarity}
(i.e.\ the kernel of a final morphism) \emph{is a congruence} w.r.t.\
the operations in $\Gamma$.

Moreover, since we would like that the structured coalgebraic model
is compatible with the unstructured, set-based one, we are
interested in functors $\Fun{B^{\Gamma}}:\alge{\Gamma} \to
\alge{\Gamma}$ that are the \emph{lifting} of some functor $\Fun{B}:\set
\to \set$ along the \emph{forgetful} functor
$\Fun{U^{\Gamma}}:\Cat{Alg_{\Gamma}} \to \Cat{Set}$ (i.e., the
following diagram commutes).
$$\xymatrix{ \Cat{Alg_{\Gamma}} \ar[d]_{\Fun{U^{\Gamma}}} \ar[r]^{\Fun{B^{\Gamma}}} & \Cat{Alg_{\Gamma}} \ar[d]^{\Fun{U^{\Gamma}}}\\
\set \ar[r]_{\Fun{B}} & \set}$$

\begin{prop}[From \cite{CGH01}]\label{prop:lifting}
Let $\Gamma$ be an algebraic specification. Let
$\Fun{U^{\Gamma}}:\Cat{Alg_{\Gamma}} \to \Cat{Set}$ be the forgetful
functor. If $\Fun{B^{\Gamma}}:\Cat{Alg_{\Gamma}} \to
\Cat{Alg_{\Gamma}}$ is a lifting of $\Pc(L \times Id)$ along
$\Fun{U^{\Gamma}}$, then \textsc{(1)} $\coalg{\Fun{B_{\Gamma}}}$ has
a final object, \textsc{(2)} bisimilarity is uniquely induced by
$\Pc(L \times Id)$-bisimilarity and \textsc{(3)} bisimilarity is a
congruence.
\end{prop}
In \cite{TP97}, \emph{bialgebras} are used as structures
combining algebras and coalgebras. Bialgebras are richer than
structured coalgebras, in the sense that they can be seen both as
coalgebras on algebras and also as algebras on coalgebras. In
\cite{CGH01}, it is shown that whenever $\Fun{B^{\Gamma}}$ is a
lifting of some $\Fun{B}$, then $\Fun{B^{\Gamma}}$-coalgebras are
also bialgebras.
In Section \ref{sec:normalizedcoalgebra}, we will introduce
\emph{normalized coalgebras} that are structured coalgebras, 
but not bialgebras (i.e., their endofunctor is not the lifting 
of some endofunctor on $\set$). This is our motivation for using structured coalgebras.
\section{Coalgebraic Saturated Semantics}\label{sec:CIScoalgebra}

\noindent Recall the definition of context interactive system (Definition
\ref{def:cis}). Here, and in the rest of the paper we will always
assume to work with a context interactive system
$\sys{I}=\<\mssign,\Alg{X},O,tr\>$ where (a) $||\Cat{C}||$ (the set
of morphisms of the small category $\Cat{C}$) is a countable set and
(b) the transition relation $tr$ has countable degree, i.e., the set
of transitions outgoing from a state is countable. These two
assumptions also guarantee that the saturated transition system has
countable degree.

In this section we introduce the coalgebraic model for the saturated
transition system. First we model it as a coalgebra over
$\set^{|\Cat{C}|}$, i.e., the category of $|\Cat{C}|$-sorted
families of sets and functions. Therefore in this model, all the
algebraic structure is missing. Then we lift it to
$\Cat{Alg_{\Gamma(\Cat{C})}}$ that is the category of
$\Gamma(\Cat{C})$-algebras and $\Gamma(\Cat{C})$-homomorphisms.
%
%
Recall that when $X$ is a $|\Cat{C}|$-sorted family of sets, $\int X
= \sum_{i\in |\Cat{C}|}X_i$.
\begin{defi}
$\PD: \set^{|\Cat{C}|} \to \set^{|\Cat{C}|}$ is defined for each
$|\Cat{C}|$-sorted family of set $X$ and for each $i \in |\Cat{C}|$
as $\PD(X_i) = \Pc (\sum_{j\in |\Cat{C}|} ( \Cat{C}[i,j] \times O
\times \int X))$. Analogously for arrows.
\end{defi}
A $\PD$-coalgebra is a $\Cat{C}$-sorted family $\alpha =
\{\alpha_i:X_i \to \PD(X_i) \mid i \in |\Cat{C}| \}$ of functions
assigning to each $p \in X_i$ a set of transitions $(c,o,q)$ where
$c$ is an arrow of $\Cat{C}$ (context) with source $i$, $o$ is an
observation and $q$ is the arriving state. Note that $q$ can have
any possible sort ($q \in \int X$).

For each $\sys{I} =\<\Cat{C}, \Alg{X}, O, tr\>$, we define the
$\PD$-coalgebra $\<X, \alpha_{\sys{I}} \>$ corresponding to the
\cts, where $\forall i \in |\Cat{C}|, \; \forall p \in X_i$, $(c, o,
q) \in \alpha_{\sys{I}}(p)$ iff $(c_{\Alg{X}}(p), o,q) \in tr$.

\medskip

Now we want to define an endofunctor $\PDA$ on
$\Cat{Alg_{\Gamma(C)}}$ that is a lifting of $\PD$ and such that
$\<\Alg{X}, \alpha_{\sys{I}}\>$ is a $\PDA$-coalgebra. In order to
do that, we must define how $\PDA$ modifies the operations of
$\Gamma(\Cat{C})$-algebras.
%
This is described by the following rule.
\[\deduz{p \atr{c_1}{o}{}q \quad c_1=d;c_2}{d(p) \atr{c_2}{o}{}q}\]
Intuitively, this rule states how to compute the saturated
transitions of $d(p)$ from the saturated transitions of $p$. Indeed,
if $p\satr{d;c_2}{o}q$, then $d;c_2(p)\tr{o}q$ and then
$d(p)\satr{c_2}{o}q$.

Hereafter, in order to make lighter the notation, we will avoid to
specify sorts. We will denote a $\Sig{\Gamma(\Cat{C})}$-algebra
$\Alg{X}$ as $\langle X, d^0_{\Alg{X}}, d^1_{\Alg{X}}, \dots
\rangle$ where $X$ is the $|\Cat{C}|$-sorted carrier set of
$\Alg{X}$ and $d^i_{\Alg{X}}$ is the function corresponding to the
operator $d^i \in ||\Cat{C}||$.

\begin{defi}\label{def:endofunctorPDA}$\PDA :
\Cat{Alg_{\Sig{\Gamma(\Cat{C})}}} \to
\Cat{Alg_{\Sig{\Gamma(\Cat{C})}}}$ maps each $\Alg{X}=\langle X,
d^0_{\Alg{X}}, d^1_{\Alg{X}}, \dots \rangle \in
\Cat{Alg_{\Sig{\Gamma(\Cat{C})}}}$ into $\langle \PD(X),
d^0_{\TPDA(\Alg{X})}, d^1_{\TPDA(\Alg{X})}, \dots  \rangle$ where
$\forall d \in ||\Cat{C}||, \; \forall A \in \PD(X),$
$d_{\TPDA(\Alg{X})} A= \{(c_2,l,x) | (c_1,l,x) \in A\ \text{ and }
c_1=d;c_2 \}$. For arrows, it is defined as $\PD$.
\end{defi}

Intuitively, $\PDA:\Cat{Alg_{\Sig{\Gamma(\Cat{C})}}} \to
\Cat{Alg_{\Sig{\Gamma(\Cat{C})}}}$ can be thought of as an extension
of the functor $\PD: \set^{|\Cat{C}|} \to \set^{|\Cat{C}|}$ to the
category $\Cat{Alg_{\Sig{\Gamma(\Cat{C})}}}$. Each algebra $\Alg{X}$
with ($|\Cat{C}|$-sorted) carrier set $X$ is mapped to an algebra
having as ($|\Cat{C}|$-sorted) carrier set $G(X)$. The elements of
$G(X)$ with sort $i$ are sets of triples $(c_1,o,x)$ (representing
sets of transitions) where $c_1:i \to j$ is an arrow in $\Cat{C}$.
For each arrow $d:i \to k$, there is an operator in $\PDA(\Alg{X})$
$d_{\TPDA(\Alg{X})}:\PD(X_i) \to \PD(X_k)$ that maps each set $A$ of
triples in $\PD(X_i)$ into the set of triples $\{(c_2,l,x) |
(c_1,l,x) \in A\ \text{ and } c_1=d;c_2 \}$ (note that the arrows
$c_2$ have source $k$).

It is worth to note that by definition, $\PDA$ is a lifting of
$\PD$. Thus, by Proposition \ref{prop:lifting}, follows that
$\coalg{\PDA}$ has final object and that bisimilarity is a
congruence.\footnote{Proposition \ref{prop:lifting} holds also for
many-sorted algebras and many sorted-sets \cite{FCT99}.}

In \cite{TP97}, it is shown that every process algebra whose
operational semantics is given by GSOS rules, defines a bialgebra.
In that approach the carrier of the bialgebra is an initial algebra
$T_{\Sigma}$ for a given algebraic signature $\Sigma$, and the GSOS
rules specify how an endofunctor $\Fun{B_{\Sigma}}$ behaves with
respect to the operations of the signature. Since there exists only
one arrow $?_{\Sigma}:T_{\Sigma} \to \Fun{B_{\Sigma}}(T_{\Sigma})$,
to give SOS rules is enough for defining the bialgebra (i.e.,
$\<T_{\Sigma}, ?_{\Sigma}\>$) and then for assuring compositionality
of bisimilarity. Our construction slightly differs from this.
Indeed, the carrier of our coalgebra is $\Alg{X}$, that is not the
initial algebra of $\Cat{Alg_{\Sig{\Sigma(\Cat{C})}}}$. Then there
might exist several or none structured coalgebras with carrier
$\Alg{X}$. In the following we prove that $\alpha_{\sys{I}}:
\Alg{X}\to \PDA(\Alg{X})$ is a $\Gamma(\Cat{C})$-homomorphism.

\begin{thm}\label{theoDMcoalgebra}
$\<\Alg{X}, \alpha_{\sys{I}}\>$ is a $\PDA$-coalgebra.
\end{thm}
Now, since a final coalgebra $\final{\PDA}$ exists in $\coalg{\PDA}$
and since $\<\Alg{X}, \alpha_{\sys{I}}\>$ is a $\PDA$-coalgebra,
there exists a final morphism from $\<\Alg{X}, \alpha_{\sys{I}}\>$.
The kernel of this coincides with $\satbis$, because (a)
$\PDA$-bisimilarity coincides with $\PD$-bisimilarity (by
Proposition \ref{prop:lifting}(2)) and (b) bisimilarity of
$\PD$-coalgebras for the saturated transition system coincides with
saturated bisimilarity.

By \cite{CGH98a}, $\<\Alg{X}, \alpha_{\sys{I}}\>$ is also a
bialgebra (since $\PDA$ is a lifting). In the next section we will
introduce coalgebraic models for symbolic semantics that are
structured coalgebras but not bialgebras.

\section{Coalgebraic Symbolic Semantics}\label{sec:CSS}

\noindent In Section \ref{sec:CIScoalgebra} we have characterized saturated
bisimilarity as the equivalence induced by the final morphism from
$\<\Alg{X}, \alpha_{\sys{I}}\>$ (i.e., the $\PDA$-coalgebra
corresponding to \cts) to $\final{\PDA}$. This is theoretically
interesting, but pragmatically useless. Indeed \cts\ is usually
infinitely branching (or in any case very inefficient), and so is
the minimal model. In this section we use symbolic bisimilarity in
order to give an efficient and coalgebraic characterization of
$\satbis$. We provide a notion of \emph{redundant transitions} and
we introduce \emph{normalized coalgebras} as coalgebras without
redundant transitions. The category of normalized coalgebras
($\coalg{\DMIN}$) is isomorphic to the category of \emph{saturated
coalgebras} ($\coalg{\DMIS}$) that is (isomorphic to) a full
subcategory of $\coalg{\PDA}$ that contains only those coalgebras
``satisfying'' an inference system $\ded{T}$. From the isomorphism
follows that $\satbis$ coincides with the kernel of the final
morphism in $\coalg{\DMIN}$. This provides a characterization of
$\satbis$ really useful: every equivalence class has a canonical
model that is smaller than that in $\coalg{\PDA}$ because normalized
coalgebras have no redundant transitions. Moreover, minimizing in
$\coalg{\DMIN}$ is usually feasible since it abstracts away from
redundant transitions.
\subsection{Saturated Coalgebras}\label{sec:saturatedcoalgebra}
Hereafter we refer to a context interactive system $\isys$ and to an
inference system $\ded{T}$. First, we extend $\vdash_{\ded{T}}$
(Definition \ref{def:der}) with the operators of
$\Gamma(\Cat{C})$-algebras.
\begin{defi}[Extended Derivation]\label{def:exDer}
Let $\Alg{X}$ be a $\Gamma (\Cat{C})$-algebra. A transition $p
\tr{c_1,o_1}q_1$ \emph{derives} a transition $d_{\Alg{X}}(p)
\tr{c_2,o_2}q_2$ in $\Alg{X}$ through $\ded{T}$ (written
$(c_1,o_1,q_1)\vdash_{\ded{T}, \Alg{X}}^d (c_2,o_2,q_2)$) iff there
exist $e, e'\in ||\Cat{C}||$ such that $c_1;e=d;c_2$ and $ e
\Ttr{o_1}{o_2}e' \in \clos{\ded{T}}$ and $e'_{\Alg{X}}(q_1)=q_2$.
\end{defi}
Intuitively, $\vdash_{\ded{T}, \Alg{X}}^d$ allows to derive from the
set of transitions of a state $p$ some transitions of
$d_{\Alg{X}}(p)$. Consider the symbolic transition
$\gamma_1\atr{a}{\bullet}{\approxwsc}a\rhd ab.\nil$ in Figure
\ref{fig:wordsaturatedSymbolic} (C). The derivation $(a, \bullet,
a\rhd ab.\nil) \vdash_{\infwsc,\WSCA}^{a} (\varepsilon, \bullet, a
\rhd ab.\nil) \vdash_{\infwsc,\WSCA}^{b} (b, \bullet, abb \rhd
ab.\nil)$ means that $a_{\WSCA}(\gamma_1)= a\rhd a.ab.\nil
\atr{\varepsilon}{\bullet}{} a \rhd ab.\nil$ and
$ab_{\WSCA}(\gamma_1)= ab\rhd a.ab.\nil \atr{b}{\bullet}{} abb \rhd
ab.\nil$. Note that both the transitions are in the saturated
transition system (by soundness of $\approxwsc$ and $\infwsc$). The
former is also in the symbolic transition system $\approxwsc$, while
the latter is not.

For open nets, take the symbolic transition
$l\atr{\$^3}{\alpha}{\eta} m$ of $\<\onet{N_4},l\>$ in Figure
\ref{fig:OpenNet}. The derivation $(\$^3,\alpha,
m)\vdash_{\ded{T}_{\modelON},\Alg{N}}^{\$^2}(\$,\alpha,
m)\vdash_{\ded{T}_{\modelON},\Alg{N}}^{\$^2}(\$,\alpha, m\$^2)$
means that $l\$^2\atr{\$}{\alpha}{}m$ and
$l\$^4\atr{\$}{\alpha}{}m\$^2$. Note that both the transitions are
in the saturated transition system (by soundness of $\eta$ and
$\ded{T}_{\modelON}$). The former is also in the symbolic transition
system $\eta$, while the latter is not.

Analogously for $\tau.\res{y}{\outp{y}{a}}+ \inp{a}{b}.\outp{a}{b}_1
\atr{-|\outp{a}{a}}{\tau}{\approxAsy}\outp{a}{a}$. The derivation
$(-|\outp{a}{a},\tau,
\outp{a}{a}_1)\vdash_{\ded{T}_{\modelAsy},\Alg{A}}^{-|\outp{a}{a}}(-,\tau,
\outp{a}{a}_1)\vdash_{\ded{T}_{\modelAsy},\Alg{A}}^{-}(-|\outp{a}{b},\tau,
\outp{a}{a}|\outp{a}{b}_2)$ means that $\tau.\res{y}{\outp{y}{a}}+
\inp{a}{b}.\outp{a}{b} | \outp{a}{a}_1\atr{-}{\tau}{}\outp{a}{a}_1$
and $\tau.\res{y}{\outp{y}{a}}+ \inp{a}{b}.\outp{a}{b} |
\outp{a}{a}_1\atr{-|\outp{a}{b}}{\tau}{}\outp{a}{a}|\outp{a}{b}_2$.
Note that both the transitions are in the saturated transition
system (by soundness of $\approxAsy$ and $\ded{T}_{\modelAsy}$). The
former is also in the symbolic transition system $\approxAsy$, while
the latter is not.

\begin{defi}[Sound Inference System]
An inference system $\ded{T}$ is \emph{sound} w.r.t.\ a
$\PDA$-coalgebra $\<\Alg{X}, \alpha\>$ (or viceversa, $\<\Alg{X},
\alpha\>$ \emph{satisfies} $\ded{T}$) provided that whenever $(c,o,q)\in
\alpha(p)$ and $(c,o,q) \vdash_{\ded{T}, \Alg{X}}^d (c',o',q')$, then
$(c',o',q')\in \alpha(d_{\Alg{X}}(p))$.
\end{defi}
For example, $\<\WSCA,\alpha_{\modelswc}\>$ (i.e., the
$\PDA$-coalgebra corresponding to the \cts\ of swc) satisfies
$\infwsc$, while the coalgebra corresponding to the symbolic
transition system $\approxwsc$ does not.
%
Analogously for the coalgebra $\<\Alg{N},\alpha_{\modelON}\>$ of
open nets and the coalgebra $\<\Alg{A},\alpha_{\modelAsy}\>$ of
asynchronous $\pi$-calculus.
Hereafter we use $\vdash_{\ded{T},\Alg{X}}$ to mean
$\vdash_{\ded{T},\Alg{X}}^{id}$.
\begin{defi}[Saturated Set]\label{def:saturatedset}
Let $\Alg{X}$ be a $\Gamma(\Cat{C})$-algebra. A set $A\in \PD(X)$ is
\emph{saturated} in $\ded{T}$ and $\Alg{X}$ if it is closed w.r.t.\
$\vdash_{\ded{T},\Alg{X}}$. The set $\PS(X)$ is the subset of
$\PD(X)$ containing all and only the saturated sets in $\ded{T}$ and
$\Alg{X}$.
\end{defi}

\begin{defi}\label{def:DMIS}
$\DMIS: \Cat{Alg_{\Sig{\Gamma(\Cat{C})}}} \to
\Cat{Alg_{\Sig{\Gamma(\Cat{C})}}}$ maps each $\Alg{X}=\langle X,
d^0_{\Alg{X}}, d^1_{\Alg{X}}, \dots \rangle \in
\Cat{Alg_{\Sig{\Gamma(\Cat{C})}}}$ into $\DMIS(\Alg{X}) = \langle
\PS(X), d^0_{\DMIS(\Alg{X})}, d^1_{\DMIS(\Alg{X})}, \dots  \rangle$
where $\forall d \in ||\Cat{C}||, \; \forall A \in \PD (X),$
$d_{\DMIS(\Alg{X})} A=  \{(c_2,o_2,x_2) \text{ s.t. }
(c_1,o_1,x_1)\in A \text{ and } (c_1,o_1,x_1) \vdash_{\ded{T},
\Alg{X}}^d (c_2,o_2,x_2) \}\text{.}$ For arrows, it is defined as
$\PD$.
\end{defi}
There are two differences w.r.t.\ $\PDA$. First, we require that all
the sets of transitions are saturated. Then the operators are
defined by using the relation $\vdash_{T,\Alg{X}}^d$.

Notice that $\DMIS$ cannot be regarded as a lifting of any
endofunctor over $\set^{|\Cat{C}|}$. Indeed the definition of
$\PS(X)$ depends on the algebraic structure $\Alg{X}$. For this
reason we cannot use Proposition \ref{prop:lifting}.

Now, let $\iota_{\Alg{X}}:\PS(X) \to \PD(X)$ be the inclusion
function. In Appendix \ref{proofSatu} it is proved that it also a
$\Gamma(\Cat{C})$-homomorphism $\iota_{\Alg{X}}:\DMIS(\Alg{X}) \to
\PDA(\Alg{X})$ and that it extends to a natural transformation.
\begin{lem}\label{lemma:inatural}
Let $\iota$ be the family of morphisms $\iota =
\{\iota_{\Alg{X}}:\DMIS(\Alg{X}) \to \PDA(\Alg{X}),\; \forall
\Alg{X}\in |\Cat{Alg_{\Sig{\Gamma(\Cat{C})}}}|\}$. Then $\iota:\DMIS
\Rightarrow \PDA$ is a natural transformation.
\end{lem}
It is well-known that every natural transformation between
endofunctors induces a functor between the corresponding categories
of coalgebras \cite{Rut96}. In our case, $\iota:\DMIS \Rightarrow
\PDA$ induces the functor $\Fun{I}: \coalg{\DMIS} \to \coalg{\PDA}$
that maps each $\DMIS$-coalgebra $\alpha: \Alg{X} \to
\DMIS(\Alg{X})$ into the $\PDA$-coalgebra $\alpha; \iota_{\Alg{X}}:
\Alg{X} \to \PDA(\Alg{X})$.

Let $\coalg{\PDA^I}$ be the full subcategory of $\coalg{\PDA}$
containing the $\PDA$-coalgebras $\alpha: \Alg{X} \to \PDA(\Alg{X})$
that factor through $\iota_{\Alg{X}}$, i.e., those
$\alpha=\alpha';\iota_{X}$ for some $\Gamma(\Cat{C})$-homomorphisms
$\alpha':\Alg{X} \to \DMIS(\Alg{X})$. It is trivial to see that this
category is isomorphic to $\coalg{\DMIS}$.

%

In order to prove the existence of final object in $\coalg{\DMIS}$,
we show that $\coalg{\PDA^I}$ is the full subcategory of
$\coalg{\PDA}$ containing all and only the coalgebras satisfying
$\ded{T}$. More precisely, we show that $|\coalg{\PDA^I}|$ is a
\emph{covariety} of $\coalg{\PDA}$.
\begin{lem}\label{lemma:Icorrect}
Let $\<\Alg{X}, \alpha\>$ be a $\PDA$-coalgebra. Then it is in
$|\coalg{\PDA^{I}}|$ iff it satisfies $\ded{T}$.
\end{lem}

\begin{prop}\label{prop:covariety}
$|\coalg{\PDA^I}|$ is a covariety of $\coalg{\PDA}$.
\end{prop}
From this follows that we can construct a final object in
$\coalg{\PDA^I}$ as the biggest subobject of $\final{\PDA}$
satisfying $\ded{T}$. Thus the kernel of final morphisms in
$\coalg{\PDA^I}$ coincides with the kernel of final morphisms in
$\coalg{\PDA}$. This argument extends to $\coalg{\DMIS}$, since it
is isomorphic to $\coalg{\PDA^I}$.

If $\ded{T}$ is \emph{sound} w.r.t.\
$\<\Alg{X},\alpha_{{\sys{I}}}\>$, then the latter is in
$|\coalg{\PDA^I}|$, i.e., $\alpha_{{\sys{I}}}=\alpha'_{{\sys{I}}};
\iota_{\Alg{X}}$. Note that $\<\Alg{X},\alpha'_{{\sys{I}}}\>$
corresponds through the isomorphism to
$\<\Alg{X},\alpha_{{\sys{I}}}\>$ (namely,
$\Fun{I}(\<\Alg{X},\alpha'_{{\sys{I}}}\>)=\<\Alg{X},\alpha_{{\sys{I}}}\>$).
Thus, by assuming $\ded{T}$ to be sound w.r.t.\
$\<\Alg{X},\alpha_{{\sys{I}}}\>$, we have that the kernel of final
morphism from $\<\Alg{X},\alpha'_{{\sys{I}}}\>$ in $\coalg{\DMIS}$
coincides with $\satbis$.

It is worth to give an intuition about $\final{\DMIS}$, the final
coalgebra of $\coalg{\DMIS}$. One can roughly thinks of
$\final{\PDA}$ (the final coalgebra of $\coalg{\PDA}$) as the
standard final coalgebra of transition systems (with labels in
$||\Cat{C}||\times O$), i.e., the coalgebra of all
\emph{synchronization trees}. The final coalgebra of $\coalg{\DMIS}$
is the biggest subcoalgebra of $\final{\PDA}$ containing all and
only those synchronization trees that are sound w.r.t.\ $\ded{T}$.
Note that $\final{\DMIS}$ is not a ``convenient semantics domain''
since all the set of transitions of a given state are saturated. In
the next subsection, we are going to show the category of normalized
coalgebras, where the final coalgebra contains only few ``essential''
symbolic transitions.
\subsection{Normalized Coalgebras}\label{sec:normalizedcoalgebra}
In this subsection we introduce normalized coalgebras, in order to
characterize $\satbis$ without considering the whole \cts\ and by
relying on the derivation relation $\vdash_{\ded{T},\Alg{X}}$. The
following observation is fundamental to explain our idea.
\begin{lem}\label{lemma:normalizedtrivial}
Let $\Alg{X}$ be a $\Gamma(\Cat{C})$-algebra. For all triples
$(c_1,o_1,p_1), (c_2,o_2,p_2)\in \PDA(\Alg{X})$, if
$(c_1,o_1,p_1)\vdash_{\ded{T}, \Alg{X}}(c_2,o_2,p_2)$ then
$p_2=e_{\Alg{X}}(p_1)$ for some $e\in ||\Cat{C}||$. Moreover $\forall
q_1\in \int X$, $(c_1,o_1,q_1)\vdash_{\ded{T},
\Alg{X}}(c_2,o_2,e_{\Alg{X}}(q_1))$.
\end{lem}
Consider a $\PDA$-coalgebra $\<\Alg{X},\gamma\>$ and the equivalence
$\sim^{\gamma}$ induced by the final morphism. Suppose that $p
\atr{c_1}{o_1}{\gamma}p_1$ and $p
\atr{c_2}{o_2}{\gamma}e_{\Alg{X}}(p_1)$ such that
$(c_1,o_1,p_1)\vdash_{\ded{T}, \Alg{X}}(c_2,o_2,e_{\Alg{X}}(p_1))$.
If $\<\Alg{X},\gamma\>$ satisfies $\ded{T}$ (i.e., it is a
$\DMIS$-coalgebra), we can forget about the latter transition.
Indeed, for all $q\in \int X$, if $q \atr{c_1}{o_1}{\gamma}q_1$ then also
$q \atr{c_2}{o_2}{\gamma}e_{\Alg{X}}(q_1)$ (since
$\<\Alg{X},\gamma\>$ satisfies $\ded{T}$) and if $p_1 \sim^{\gamma}
q_1$, then also $e_{\Alg{X}}(p_1) \sim^{\gamma} e_{\Alg{X}}(q_1)$
(since $\sim^{\gamma}$ is a congruence). Thus, when checking
bisimilarity, we can avoid to consider those transitions that are
derivable from others. We call such transitions \emph{redundant}.

A wrong way to efficiently characterize $\sim^{\gamma}$ by
exploiting $\vdash_{\ded{T},\Alg{X}}$, consists in removing all the
redundant transitions from $\<\Alg{X},\gamma\>$ obtaining a new
coalgebra $\<\Alg{X}, \beta\>$ and then computing $\sim^{\beta}$
(i.e., the ordinary bisimilarity on $\<\Alg{X}, \beta\>$). When
considering $\<\Alg{X},\alpha_{\sys{I}} \>$ (i.e., the
$\PDA$-coalgebra corresponding to \cts), this roughly means to build
a symbolic transition system and then computing the ordinary
bisimilarity over this. But, as we have seen in Section
\ref{sec:CIS}, the resulting bisimilarity (denoted by $\sim^W$) does
not coincide with the original one.
Generally, this happens when
\begin{enumerate}[(1)]
\item $p \atr{c_1}{o_1}{\beta}p_1$ and $p \atr{c_2}{o_2}{\beta}p_2$ with
$(c_1,o_1,p_1)\vdash_{\ded{T}, \Alg{X}} (c_2,o_2,e_{\Alg{X}}(p_1))$
and
\item $e_{\Alg{X}}(p_1)\neq p_2$, but
\item $e_{\Alg{X}}(p_1) \sim^{\gamma} p_2$.
\end{enumerate}
Notice that $p \atr{c_2}{o_2}{\beta}p_2$ is not removed, because it
is not considered redundant since $e_{\Alg{X}}(p_1)$ is different
from $p_2$ (even if semantically equivalent). A transition as the
latter is called \emph{semantically redundant} and it causes the
mismatch between $\sim^{\beta}$ and $\sim^{\gamma}$. Indeed, take a
process $q$ that only performs $q \atr{c_1}{o_1}{\beta}q_1$ with
$p_1\sim^{\gamma} q_1$. Clearly $p \not \sim^{\beta} q$, but $p
\sim^{\gamma}q$. Indeed $(c_1,o_1,q_1)\vdash_{\ded{T},
\Alg{X}}(c_2,o_2,e_{\Alg{X}}(q_1))$ and thus $q
\atr{c_2}{o_2}{\gamma}e_{\Alg{X}}(q_1)$ (since $\<\Alg{X},\gamma\>$
satisfies $\ded{T}$) and $p_2 \sim^{\gamma}
e_{\Alg{X}}(p_1)\sim^{\gamma} e_{\Alg{X}}(q_1)$ (since
$\sim^{\gamma}$ is a congruence).

As an example consider the symbolic transition system of $\gamma_1$
(Figure \ref{fig:wordsaturatedSymbolic}(C)). We have that (1) $\gamma_1
\atr{a}{\bullet}{\approxwsc}a\rhd ab.\nil$ and $\gamma_1
\atr{ab}{\bullet}{\approxwsc}ab\rhd \varepsilon.\nil$ with
$(a,\bullet,a\rhd ab.\nil)
\vdash_{\infwsc,\WSCA} (ab, \bullet, ab\rhd ab.\nil)$; (2)
$ab\rhd ab.\nil \neq ab\rhd \varepsilon.\nil$, but (3) $ab\rhd ab.\nil \satbis ab\rhd
\varepsilon.\nil$. Thus, the symbolic transition $\gamma_1
\atr{ab}{\bullet}{\approxwsc}ab\rhd \varepsilon.\nil$ is
semantically redundant and it is the reason why
$\gamma_2=\varepsilon \rhd a.ab.\nil$ is not syntactically bisimilar
to $\gamma_1$ (i.e., $\gamma_1 \not \sim^W \gamma_2$) even if they
are saturated bisimilar (as discussed in Section \ref{sec:swc}).

As a further example consider 
$\<\onet{N_4},l\>$ (Figure \ref{fig:OpenNet}):
(1) $l\atr{\emptyset}{\alpha}{\approxON}q$ and
$l\atr{\$^3}{\alpha}{\approxON}m$ with 
$(\emptyset,\alpha,q)
\vdash_{\infON,\Alg{N}} (\$^3,\alpha,\$^3_{\Alg{N}}(q))$ and (2)
$\$^3_{\Alg{N}}(q)=q\$^3\neq m$, but $q\$^3 \sim^S m$. Now consider
$\<\onet{N_1},a\>$.
$a\atr{\emptyset}{\alpha}{\approxON}b$. 
Clearly $l \not \sim^W a$ but $l \satbis a$ (as shown in Section
\ref{openpetrinets}).

For the asynchronous $\pi$-calculus consider the symbolic
transitions of $p_1=\tau.\res{y}{\outp{y}{a}}+ \inp{a}{b}.\outp{a}{b}_1$
in Figure \ref{fig:processes}(B):
(1) $p_1\atr{-}{\tau}{\approxAsy}\res{y}{\outp{y}{a}}_1$ and $p_1\atr{-|
\outp{a}{a}}{\tau}{\approxAsy}\outp{a}{a}_1$ with $(-,\tau,\res{y}{\outp{y}{a}}_1)\vdash_{\infAsy,\Alg{A}} (-|\outp{a}{a},\tau, -
| \outp{a}{a}_{\Alg{A}}(\res{y}{\outp{y}{a}}_1))$; (2) 
$-|\outp{a}{a}_{\Alg{A}}(\res{y}{\outp{y}{a}}_1)
=\res{y}{\outp{y}{a}}| \outp{a}{a}_1 \neq \outp{a}{a}_1 $, but (3)
$\res{y}{\outp{y}{a}}| \outp{a}{a}_1 \sim^S \outp{a}{a}_1 $. Now the
process $\tau.\nil_1$ only performs $
\atr{-}{\tau}{\approxAsy}\nil_1$. Clearly
$\tau.\res{y}{\outp{y}{a}}+ \inp{a}{b}.\outp{a}{b}_1 \not \sim^W
\tau.\nil_1$, but they are saturated bisimilar (as shown in Section
\ref{sec:Asyn}).

\medskip

The above observation tells us that we have to remove not only the
redundant transition, i.e., those derivable from $\vdash_{\ded{T},
\Alg{X}}$, but also the \emph{semantically redundant} ones. But
immediately a problem arises. How can we decide which transitions
are semantically redundant, if semantic redundancy itself depends on
bisimilarity?

Our solution is the following: we define a category of coalgebras
without redundant transitions ($\coalg{\DMIN}$) and, as a result, a
final coalgebra contains no \emph{semantically} redundant
transitions.

\begin{defi}[Normalized Set and
Normalization]\label{def:normalization}Let $\Alg{X}$ be a $\Gamma
(\Cat{C})$-algebra.

A transition $(c',o',q')$ is \emph{equivalent} to $(c,o,q)$ in
$\ded{T}, \Alg{X}$ (written $(c',o',q') \equiv_{\ded{T},\Alg{X}}
(c,o,q)$) iff $(c',o',q')\vdash_{\ded{T},\Alg{X}} (c,o,q)$ and
$(c,o,q) \vdash_{\ded{T},\Alg{X}} (c',o',q')$.

A transition $(c',o',q')$ \emph{dominates} $(c, o,q)$ in $T,
\Alg{X}$ (written $(c',o',q') \prec _{\ded{T},\Alg{X}}(c,o,q) $) iff
$(c',o',q')\vdash_{\ded{T},\Alg{X}} (c,o,q)$ and $(c,o,q)
\nvdash_{\ded{T},\Alg{X}} (c',o',q')$.

Let $A \in \PD (X)$. A transition $(c,o,q)\in A$ is \emph{redundant
in $A$} w.r.t.\ $\ded{T}, \Alg{X}$ if $\exists (c',o',q')\in A$ such
that $(c',o',q')\prec_{\ded{T},\Alg{X}} (c,o,q)$.

The set $A$ is \emph{normalized} w.r.t.\ $\ded{T}, \Alg{X}$ iff it
does not contain redundant transitions and it is closed by
equivalent transitions. The set $\PN{\Alg{X}}(X)$ is the subset of
$\PD(X)$ containing all and only the normalized sets w.r.t.\
$\ded{T}, \Alg{X}$.

The \emph{normalization function} $\norm{X}: \PD(X) \to
\PN{\Alg{X}}(X)$ maps $A \in \PD(X)$ into $\{(c',o',q')$ $ \text{
s.t. }\exists (c,o,q)\in A \text{ s.t. }
(c',o',q')\equiv_{\ded{T},\Alg{X}} (c,o,q) \text{ and } (c,o,q)$
 not redundant in $A \text{ w.r.t.\ } \ded{T}, \Alg{X}\}$.
\end{defi}

\begin{figure}[t]
\begin{tabular}{c}
$\xymatrix@R=10pt@C=6pt{
\gamma_2 \ar[rrd]^(.4){a} \ar@(ld,lu)@{.>}[dddddddd] \\
\gamma_1 \ar[rr]^(.4){a} \ar[rrd]_(.4){ab}
\ar@(ld,lu)@{.>}[ddddddd]& & a\rhd ab.\nil \ar[rr]^{b}
\ar@(ld,lu)@{.>}[ddddddd] &&
ab \rhd \nil \ar@{.>}[ddddddd] &&  ab \rhd ab.\nil \ar[ll]_{\varepsilon} \ar@{.>}[ddddddd] & \text{(A)} \\
& & ab\rhd \varepsilon.\nil \ar[rru]_{\varepsilon} \ar@(r,u)@{.>}[ddddddrrrr]  \\
\\
\\
\gamma \ar[rr]^(.4){a} \ar@(dr,dl)[rrrrrr]_{ab}& &  \gamma'
\ar[rr]^{b} &&
\gamma'' &&  b_{\Alg{U}}(\gamma') \ar[ll]_{\varepsilon} & \text{(B)}\\
\\
\\
\gamma \ar[rr]^(.4){a} & &  \gamma' \ar[rr]^{b} && \gamma'' &&
b_{\Alg{U}}(\gamma') \ar[ll]_{\varepsilon}  & \text{(C)}}$
\end{tabular}
\caption{(A) Part of the normalized coalgebra $\<\WSCA,
\alpha_{\modelswc};norm_{T_{\modelswc},\WSCA}\>$. (B) Part of a not
normalized coalgebra $\<\Alg{U},\zeta\>$. (C) Part of a normalized
coalgebra $\<\Alg{U},\zeta;norm_{T_{\modelswc},\Alg{U}}\>$. The
dotted arrows represent a $\Fun{N}_{\Fun{T}_{\modelswc}}$-morphism
from $\<\WSCA, \alpha_{\modelswc};norm_{T_{\modelswc},\WSCA}\>$ to
$\<\Alg{U},\zeta;norm_{T_{\modelswc},\Alg{U}}\>$.}\label{fig:wordNormalization}
\end{figure}

\begin{figure}[t]
\begin{center}
\begin{tabular}{ccc}
$\xymatrix@R=5pt@C=7pt{ u \ar[rr]^{\emptyset, \alpha}
\ar@(rd,ld)[rrrrrrrr]|{\$^3,\alpha} & & v \ar@(r,u)_{\$,\beta} && w
\ar[ll]^{\emptyset, \beta} & & x \ar[ll]^{\emptyset, \beta} & &
\$^3_{\Alg{Y}}(v) \ar[ll]^{\emptyset, \beta}}$ & \hspace{1cm} &
$\xymatrix@R=5pt@C=7pt{ u \ar[rr]^{\emptyset, \alpha} & & v
\ar@(r,u)_{\$,\beta} && w \ar[ll]^{\emptyset, \beta} & & x
\ar[ll]^{\emptyset, \beta} & & \$^3_{\Alg{Y}}(v) \ar[ll]^{\emptyset,
\beta}}$\\(A)& &(B) \\
$\xymatrix@R=6pt@C=5.5pt{
&& q_2 && -|\outp{a}{b}_{\Alg{Z}}(q_1) \ar[ll]_{-, \outp{a}{b}} \\
p_1 \ar[rrd]_{-, \tau} \ar@(r,lu)[rrrrd]|{-|\outp{a}{a}, \tau} \ar@(ur,dl)[rrrru]|{-|\outp{a}{b}, \tau}&& \\
&& q_1 && -|\outp{a}{a}_{\Alg{Z}}(q_1) \ar[ll]_{-, \outp{a}{a}}  \\
}$ && $\xymatrix@R=6pt@C=5.5pt{
&& q_2 && -|\outp{a}{b}_{\Alg{Z}}(q_1) \ar[ll]_{-, \outp{a}{b}} \\
p_1 \ar[rrd]_{-, \tau} && \\
&& q_1 && -|\outp{a}{a}_{\Alg{Z}}(q_1) \ar[ll]_{-, \outp{a}{a}}  \\
}$\\ (C) && (D)
\end{tabular}
\caption{(A) Part of a not normalized coalgebra
$\<\Alg{Y},\gamma\>$. (B) Part of a normalized coalgebra
$\<\Alg{Y},\gamma;norm_{T_{\ded{\modelON}},\Alg{Y}}\>$. (C) Part of
a not normalized coalgebra $\<\Alg{Z},\delta\>$. (D) Part of a
normalized coalgebra
$\<\Alg{Z},\delta;norm_{T_{\ded{\modelAsy}},\Alg{Z}}\>$.
 }\label{fig:normalized}
\end{center}
\end{figure}
%
%

Recall $\modelswc = \modelwscAll$ and $T_{\ded{\modelswc}}$
(introduced in Section \ref{sec:CIS}). Consider the coalgebra
$\<\Alg{U},\zeta\>$ partially depicted in Figure
\ref{fig:wordNormalization}(B). Here we have that
$(a,\bullet,\gamma')\vdash_{T_{\ded{\modelswc}},
\Alg{U}}(ab,\bullet,b_{\Alg{U}}(\gamma'))$ but
$(ab,\bullet,b_{\Alg{U}}(\gamma'))\not \vdash_{T_{\ded{\modelswc}},
\Alg{U}}(a,\bullet,\gamma')$. Thus
$(a,\bullet,\gamma')\prec_{T_{\ded{\modelON}},
\Alg{U}}(ab,\bullet,b_{\Alg{U}}(\gamma'))$ and then the set
$\zeta(\gamma)$, i.e., the set of transitions of $\gamma$, is not
normalized (w.r.t.\ $\ded{T_{\modelswc}}, \Alg{U}$) since the
transition $(ab,\bullet,b_{\Alg{U}}(\gamma'))$ is redundant in
$\zeta(\gamma)$. By applying $norm_{T_{\ded{\modelswc},\Alg{U}}}$ to
$\zeta(\gamma)$, we get the normalized set of transitions
$\{(a,\bullet,\gamma')\}$ (Figure \ref{fig:wordNormalization}(C)).
It is worth noting that in swc, two transitions are equivalent iff
they are the same transition.

Now consider $\modelON = \modelONAll$, $T_{\ded{\modelON}}$ (introduced in Section \ref{sec:CIS}) and the
coalgebra $\<\Alg{Y},\gamma\>$ partially depicted in Figure
\ref{fig:normalized}(A). Here we have that
$(\emptyset,\alpha,v)\vdash_{T_{\ded{\modelON}},
\Alg{Y}}(\$^3,\alpha,\$^3_{\Alg{Y}}(v))$ but $(\$^3,\alpha,$
$\$^3_{\Alg{Y}}(v))\not \vdash_{T_{\ded{\modelON}},
\Alg{Y}}(\emptyset,\alpha,v)$. Thus
$(\emptyset,\alpha,v)\prec_{T_{\ded{\modelON}},
\Alg{Y}}(\$^3,\alpha,\$^3_{\Alg{Y}}(v))$ and then the set
$\gamma(u)$, i.e., the set of transitions of $u$, is not normalized
(w.r.t.\ $\ded{T_{\modelON}}, \Alg{Y}$) since the transition
$(\$^3,\alpha,\$^3_{\Alg{Y}}(v))$ is redundant in $\gamma(u)$. By
applying $norm_{\Alg{Y},T_{\ded{\modelON}}}$ to $\gamma(u)$, we get
the normalized set of transitions $\{(\emptyset,\alpha,v)\}$ (Figure
\ref{fig:normalized}(B)). Also in open Petri nets, two transitions
are equivalent iff they are the same transition.

Finally consider $\modelAsy = \modelAsyAll$, $T_{\ded{\modelAsy}}$ (introduced in Section \ref{sec:CIS})
and the coalgebra $\<\Alg{Z},\delta\>$ partially depicted in Figure
\ref{fig:normalized}(C). We have that
$(-,\tau,q_1)\vdash_{T_{\ded{\modelAsy}},
\Alg{Z}}(-|\outp{a}{a},\tau,- | \outp{a}{a}_{\Alg{Z}}(q_1))$ but
$(-|\outp{a}{a},\tau,- | \outp{a}{a}_{\Alg{Z}}(q_1))\not
\vdash_{T_{\ded{\modelAsy}}, \Alg{Z}}(-,\tau,q_1)$. The same holds
for $(-|\outp{a}{b},\tau,- | \outp{a}{b}_{\Alg{Z}}(q_1))$. Thus the
set $\delta(p_1)$, i.e., the set of transitions of $p_1$, is not
normalized (w.r.t. $\ded{T_{\modelAsy}}, \Alg{Z}$) since the
transitions $(-|\outp{a}{a},\tau,- | \outp{a}{a}_{\Alg{Z}}(q_1))$
and $(-|\outp{a}{b},\tau,- | \outp{a}{b}_{\Alg{Z}}(q_1))$ are
redundant in $\delta(p_1)$ (they are dominated by $(-,\tau,q_1)$).
By applying $norm_{\Alg{Z},T_{\ded{\modelAsy}}}$ to $\delta(p_1)$,
we obtain the normalized set of transitions $\{(-,\tau,q_1)\}$ (in
Figure \ref{fig:normalized}(D)). Also in the asynchronous
$\pi$-calculus, two transitions are equivalent iff they are the same
transition.

\begin{defi}\label{def:normalizedFunctor}
$\DMIN: \Cat{Alg_{\Sig{\Gamma(\Cat{C})}}} \to
\Cat{Alg_{\Sig{\Gamma(\Cat{C})}}}$ maps each $\Alg{X}=\langle X,
d^0_{\Alg{X}}, d^1_{\Alg{X}}, \dots \rangle \in
\Cat{Alg_{\Sig{\Gamma(\Cat{C})}}}$ into $\DMIN(\Alg{X}) =\langle
\PN{\Alg{X}}(X), d^0_{\DMIS(\Alg{X})}; \norm{X} ,
d^1_{\DMIS(\Alg{X})}; \norm{X}, \dots \rangle$. For all $h: \Alg{X}
\to \Alg{Y}$, let $\PDA'(h): \DMIN(\Alg{X})\to \PDA(\Alg{Y})$ be the restricion of $\PDA(h)$ to $\DMIN(\Alg{X})$. 
Then, $\DMIN(h): \DMIN(\Alg{X})\to \DMIN(\Alg{Y})$ is defined as $\PDA'(h) ; \norm{Y} $.
\end{defi}

Hereafter we will sometimes write $\PDA(h)$ to mean its restriction $\PDA'(h)$.

The coalgebra $\<\Alg{U},\zeta\>$ (partially depicted
in Figure \ref{fig:wordNormalization}) and $\<\Alg{Y},\gamma\>$,
$\<\Alg{Z},\delta\>$ (in Figure \ref{fig:normalized}(A)(C)) are not
normalized.
In order to get a normalized coalgebra for our running examples, we
can normalize their saturated coalgebra $\<\WSCA, \alpha_{\modelswc}
\>$, $\<\Alg{N},\alpha_{\modelON}\>$ and
$\<\Alg{A},\alpha_{\modelAsy}\>$ obtaining, respectively, $\<\WSCA,
\alpha_{\modelswc};norm_{\ded{T_{\modelswc}},\WSCA} \>$,
$\<\Alg{N},\alpha_{\modelON};norm_{\ded{T_{\modelON}},\Alg{N}}\>$
and
$\<\Alg{A},\alpha_{\modelAsy};norm_{\ded{T_{\modelAsy}},\Alg{A}}\>$.
For $\gamma_1$ and $\gamma_2$ in Figure
\ref{fig:wordsaturatedSymbolic}(C), for the nets in Figure
\ref{fig:OpenNet} and for the process $\tau.\res{y}{\outp{y}{a}}+
\inp{a}{b}.\outp{a}{b}_1$, this coincides with their \sts. Section
\ref{sec:Algo} discusses the exact relationship between a \sts\ and
the transition system that is obtained by normalizing
$\alpha_{\sys{I}}$.

\bigskip

The most important idea behind normalized coalgebra is in the
definition of $\DMIN(h)$: we first apply $\PDA(h)$ and then the
normalization $\norm{Y}$. Thus $\DMIN$-morphisms must preserve not
all the transitions of the source coalgebras, but only those that
are not redundant when mapped into the target.

For instance, consider the function $h$ from $\<\WSCA,
\alpha_{\modelswc};norm_{T_{\modelswc},\WSCA}\>$ to
$\<\Alg{U},\zeta;norm_{T_{\modelswc},\Alg{U}}\>$ that is partially depicted
in Figure \ref{fig:wordNormalization}. Note that the transition
$\gamma_1 \atr{ab}{\bullet}{}ab \rhd \varepsilon. \nil$ is not
preserved, but $h$ is however an $\DMIN$-morphisms because the
transition $(ab, \bullet, b_{\Alg{U}}(\gamma'))$ is removed by
$norm_{T_{\modelswc},\Alg{U}}$. Thus, $h$ forgets about the
transition $\gamma_1 \atr{ab}{\bullet}{}ab \rhd \varepsilon. \nil$
that is indeed semantically redundant.

For the asynchronous $\pi$, consider the coalgebra $\< \OINPA,
\alpha_{\modelON};norm_{\ded{T_{\modelON}},\Alg{N}}\>$. For the
state $l$, it coincides with the \sts\ $\eta$ (Figure
\ref{fig:OpenNet}(C)). Consider $\<\Alg{Y},
\gamma;norm_{\ded{T_{\modelON}},\Alg{Y}}\>$ (partially represented
in Figure \ref{fig:normalized}(B)) and the
$\Gamma(\onmssign)$-homomorphism $h:\ONPA \to \Alg{Y}$ that maps
$l,m,n,o$ into $u, \$^3_{\Alg{Y}}(v),x,w$ (respectively) and $p,q$
into $v$. The morphism is shown in Figure \ref{fig:morphism}. Note
that the transition $l\atr{\$^3}{\alpha}{\eta}m$ is not preserved
(i.e., $u \not\atr{\$^3}{\alpha}{\gamma}h(m)$), but $h$ is however a
$\DMIN$-morphism, because the transition $(\$^3, \alpha, h(m))\in
\PDA(h)(\eta(l))$ is removed by $norm_{\ded{T_{\modelON}},\Alg{Y}}$.
Indeed $h(m)=\$^3_{\Alg{Y}}(v)$ and
$(\emptyset,\alpha,v)\vdash_{\ded{T_{\modelON}},\Alg{Y}}(\$^3,\alpha,\$^3_{\Alg{Y}}(v))$.
Thus, we forget about $l\atr{\$^3}{\alpha}{\eta}m$ that is, indeed,
semantically redundant.

As a further example, consider the coalgebras $\< \Alg{A},
\alpha_{\modelAsy};norm_{\ded{T_{\modelAsy}},\Alg{A}}\>$. For the
state $\tau.\res{y}{\outp{y}{a}}+ \inp{a}{b}.\outp{a}{b}_1$, it
coincides with the \sts\ $\approxAsy$ (in Figure
\ref{fig:processes}(B)). Consider $\<\Alg{Z},
\delta;norm_{\ded{T_{\modelAsy}},\Alg{Z}}\>$ (partially represented
in Figure \ref{fig:normalized}(D)) and the
$\Gamma(\amssign)$-homomorphism $h:\Alg{A} \to \Alg{Z}$ shown in
Figure \ref{fig:morphism}. Note that for all $i\in \names$ the
transitions $\tau.\res{y}{\outp{y}{a}}+ \inp{a}{b}.\outp{a}{b}_1
\atr{-|\outp{a}{i}}{\tau}{\approxAsy}\outp{a}{i}_i$ are not
preserved (i.e., $p_1
\not\atr{-|\outp{a}{i}}{\tau}{\delta}h(\outp{a}{i}_i)$), but $h$ is
however a $\DMIN$-morphism, because the transitions $(-|\outp{a}{i},
\tau, h(\outp{a}{i}_i))\in
\PDA(g)(\approxAsy(\tau.\res{y}{\outp{y}{a}}+
\inp{a}{b}.\outp{a}{b}_1))$ are removed by
$norm_{\ded{T_{\modelAsy}},\Alg{Z}}$. Indeed $h(\outp{a}{i}_i)=-|
\outp{a}{i}_{\Alg{Z}}(q_1)$ and
$(-,\tau,q_1)\vdash_{\ded{T_{\modelAsy}},\Alg{Z}}(-|\outp{a}{i},\tau,-|\outp{a}{i}_{\Alg{Z}}(q_1))$.
Thus, we forget about all the transitions
$\atr{-|\outp{a}{i}}{\tau}{\approxAsy}\outp{a}{i}_i$ that are,
indeed, semantically redundant.

\subsection{Isomorphism Theorem}\label{sec:ISO}
Now we prove that $\coalg{\DMIN}$ is isomorphic to $\coalg{\DMIS}$.
\begin{defi}[Saturation]
Let $\Alg{X}$ be a $\Gamma(\Cat{C})$-algebra. The saturation function $\sat{\Alg{X}}: \PD(X) \to \PS(X)$ maps all sets of transitions $A \in
\PD(X)$ into the set $\{ (c',o',x') \text{ s.t. } (c,o,x)\in A \text{ and }
(c,o,x)\vdash_{\ded{T}, \Alg{X}} (c',o',x') \}$.
\end{defi}
Saturation is intuitively the opposite of normalization. Indeed
saturation adds to a set all the redundant transitions, while
normalization junks all of them. Thus, if we take a saturated set of
transitions, we first normalize it and then we saturate it, we
obtain the original set. Analogously for a normalized set.

However, in order to get such correspondence, we must add a
constraint to our theory. Indeed, according to the actual
definitions, there could exist a $\DMIS$-coalgebra
$\<\Alg{X},\gamma\>$ and an infinite descending chain like: $\dots
\prec_{\ded{T},\Alg{X}} p\atr{c_2}{o_2}{\gamma}p_2
\prec_{\ded{T},\Alg{X}} p\atr{c_1}{o_1}{\gamma}p_1$. In this chain,
all the transitions are redundant and thus if we normalize it, we
obtain an empty set of transitions.
\begin{defi}[Normalizable System]\label{def:Normalizable}
A context interactive system $\isys$ is \emph{normalizable w.r.t.\
$\ded{T}$} iff $\forall \Alg{X} \in \Cat{Alg_{\Gamma(\Cat{C})}}$,
$\prec_{\ded{T},\Alg{X}}$ is well founded, i.e., there are not
infinite descending chains of $\prec_{\ded{T},\Alg{X}}$.
\end{defi}
In Appendix \ref{appendixNormalizable}, we show that the context
interactive systems for open nets and asynchronous $\pi$ are
normalizable w.r.t.\ their inference systems.

\begin{lem}\label{lemmaPropNorm}
Let $\sys{I}$ be a normalizable system w.r.t.\ $\ded{T}$. Let
$\Alg{X}$ be $\Gamma(\Cat{C})$-algebra and $A \in \PD(X)$. Then
$\forall (d,o,x)\in A$, either $(d,o,x) \in \norm{X}(A)$ or $\exists (d',o',x') \in \norm{X}(A)$, such
that $(d',o',x')\prec_{\ded{T},\Alg{X}} (d,o,x)$.
\end{lem}
The above lemma guarantees that normalizing a set of transitions
produces a new set containing all the transitions that are needed to
retrieve the original one.
Hereafter, we always refer to normalizable systems.
\begin{prop}\label{prop:iso}
Let $\normaR$, respectively, $\satuR$ be the families of morphisms
$ \{\norm{X}:\DMIS(\Alg{X}) \to \DMIN(\Alg{X}), \;\forall \Alg{X}
\in | \Cat{Alg_{\Gamma(\Cat{C})}}| \} $ and $
\{\sat{X}:\DMIN(\Alg{X}) \to \DMIS(\Alg{X}), \;\forall \Alg{X} \in
|\Cat{Alg_{\Gamma(\Cat{C})}}| \}$. Then $\normaR : \DMIS \Rightarrow
\DMIN$ and $\satuR : \DMIN \Rightarrow \DMIS$ are natural
transformations. More precisely, they are natural isomorphisms, one
the inverse of the other.
\end{prop}
As for the case of the natural transformation $\iota$, we use the
fact that that any natural transformation between endofunctors
induces a functor between the corresponding categories of coalgebras
\cite{Rut96}. In the present case, $\normaR: \DMIS \Rightarrow
\DMIN$ induces the functor $\NormR\ : \coalg{\DMIS} \to
\coalg{\DMIN}$ that maps every coalgebra $\<\Alg{X}, \alpha \>$ in
$\< \Alg{X}, \alpha ; \normR{X}\>$ and every cohomomorphism $h$ in
itself. Analogously $\satuR: \DMIN \Rightarrow \DMIS$ induces
$\Fun{SAT_T}\ : \coalg{\DMIN} \to \coalg{\DMIS}$. These two functors
are one the inverse of the other.
\begin{thm}\label{thm:iso}
$\coalg{\DMIS}$ and $\coalg{\DMIN}$ are isomorphic.
\end{thm}
Thus $\coalg{\DMIN}$ has a final coalgebra $\final{\DMIN}$ and the
final morphisms from $\< \Alg{X}, \alpha_{\sys{I}} ; \normR{X}\>$
(that is $\NormR\ \< \Alg{X}, \alpha_{\sys{I}}\>$) still
characterizes $\satbis$. This is theoretically very interesting,
since the minimal canonical representatives of $\satbis$ in
$\coalg{\DMIN}$ do not contain any (semantically) redundant
transitions and thus they are much smaller than the (possibly
infinite) minimal representatives in $\coalg{\DMIS}$. Pragmatically,
it allows for an effective procedure for minimizing that we will
discuss in the next section. Notice that minimization is usually
unfeasible in $\coalg{\DMIS}$, since the saturated transitions
systems are usually infinite.
%
%

%
%
%
\section{From Normalized Coalgebras to Symbolic Minimization}\label{sec:Algo}
\noindent In \cite{ESOP09}, we have introduced a partition refinement
algorithm for symbolic bisimilarity. First, it creates a partition
$P_0$ equating all the states (with the same interface) of a
symbolic transition system $\beta$ and then, iteratively, refines
this partition by splitting non equivalent states. The algorithm
terminates whenever two subsequent partitions are equivalent.
It computes the partition $P_{n+1}$ as follows: $p$ and $q$ are
equivalent in $P_{n+1}$ iff whenever $p\atr{c}{o}{\sym}p_1$ is
\emph{not-redundant in $P_n$}, then $q\atr{c}{o}{\sym}q_1$ is
\emph{not-redundant in $P_n$} and $p_1,q_1$ are equivalent in $P_n$
(and viceversa). By ``not-redundant in $P_n$'', we mean that no
transition $p\atr{c'}{o'}{\sym}p_1'$ exists such that
$(c',o',p_1')\vdash_{\ded{T},\Alg{X}}(c,o,p_2')$ and $p_2',p_1$ are
equivalent in $P_n$.

Figure \ref{fig:wordPart} shows the partitions computed by the
algorithm for the symbolic transition system $\approxwsc$ of
$\gamma_1=\varepsilon \rhd a.ab.\nil + ab.\varepsilon. \nil$ and
$\gamma_2=\varepsilon \rhd a.ab.\nil$. In $P_0$ all the
configurations are equivalent since they all have the same interface
(more generally, in swc, all the configurations have the same
interface). Then in $P_1$, $\{ab \rhd \varepsilon.\nil, ab \rhd ab.
\nil\}$ are distinguished by all the other configurations because
they are the only ones that can perform a transition with
$\varepsilon$. Analogously, $ab \rhd \nil$ is different from all the
others, because it is the only that performs no transition, while $a
\rhd ab.\nil$ is distinguished because it can perform a $b$
transition. Note that $\gamma_1$ and $\gamma_2$ are equivalent in
$P_1$, because the transition $\gamma_1
\atr{ab}{\bullet}{\approxwsc}ab \rhd \varepsilon. \nil$ is redundant
in $P_0$. Indeed $(a,\bullet, a\rhd ab.\nil)\vdash_{T_{\modelswc},
\WSCA} (ab, \bullet,ab\rhd ab.\nil)$ and $ab \rhd \varepsilon. \nil$
is equivalent to $ab\rhd ab.\nil$ in $P_0$. The same holds for
$P_2$.

Figure \ref{fig:SeqPart} shows the partitions computed by the
algorithm for the \sts\ $\eta$ of the marked nets $\<\onet{N_1},a\>$
and $\<\onet{N_4},l\>$ of Figure \ref{fig:OpenNet}. Note that $a$
and $l$ are equivalent in the partition $P_1$, because the
transition $l\atr{\$^3}{\alpha}{\eta}m$ is redundant in $P_0$.
Indeed, $l\atr{\emptyset}{\alpha}{\eta}q$,
$(\emptyset,\alpha,q)\vdash_{\ded{T_{\modelON}},\Alg{N}}(\$^3,\alpha,q\$^3)$
and $m$ is equivalent to $q\$^3$ in $P_0$. Analogously, for the
other $P_i$.

Figure \ref{fig:SeqPartAsy} shows the partitions computed by the
algorithm for the \sts\ $\approxAsy$ of the asynchronous processes
$\tau.\res{y}{\outp{y}{a}}+ \inp{a}{b}.\outp{a}{b}_1$ and
$\tau.\nil_1$. Since the \sts\ of the former process is infinite,
our algorithm cannot work in reality. We discuss this issue in the
next section and for the time being, we imagine to have a procedure
that can manipulate this infinite \lts. First of all, note that all
the states with different interfaces are different in $P_0$ (while
in the case of swc and open nets, all the states have the same
interface). Moreover, $\tau.\res{y}{\outp{y}{a}}+
\inp{a}{b}.\outp{a}{b}_1$ and $\tau.\nil_1$ are equivalent in the
partition $P_1$, because for all $i\in \omega_0$, the transitions
$\tau.\res{y}{\outp{y}{b}}+ \inp{a}{b}.\outp{a}{b}_1
\atr{-|\outp{a}{i}}{\tau}{\approxAsy}\outp{a}{i}_i$ are redundant in
$P_0$. Indeed, $\tau.\res{y}{\outp{y}{a}}+
\inp{a}{b}.\outp{a}{b}_1\atr{-}{\tau}{\approxAsy}\res{y}{\outp{y}{a}}_1$,
$(-,\tau,\res{y}{\outp{y}{a}}_1)\vdash_{\ded{T_{\modelAsy}},\Alg{A}}(-|\outp{a}{i},\tau,\res{y}{\outp{y}{a}}|\outp{a}{i}_i)$
and $\outp{a}{i}_i$ is equivalent to
$\res{y}{\outp{y}{a}}|\outp{a}{i}_i$ in $P_0$. Analogously, for
$P_2$.

\begin{figure}[t]
\begin{tabular}{c}
$\xymatrix@R=10pt@C=6pt{
\gamma_2 \ar[rrd]^(.4){a} \\
\gamma_1 \ar[rr]^(.4){a} \ar[rrd]_(.4){ab}& &  a\rhd ab.\nil
\ar[rr]^{b} &&
ab \rhd \nil &&  ab \rhd ab.\nil \ar[ll]_{\varepsilon} \\
& & ab\rhd \varepsilon.\nil \ar[rru]_{\varepsilon} }$
\\
\begin{tabular}{l}
\\
\\
$P_0=\{\gamma_1,\gamma_2, a \rhd ab.\nil, ab \rhd \nil, ab \rhd
\varepsilon.\nil, ab \rhd ab. \nil\}$
\\
$P_1=\{\gamma_1,\gamma_2\}\{a \rhd ab.\nil\}\{ab \rhd \nil\} \{ab
\rhd \varepsilon.\nil, ab \rhd ab. \nil\}$
\\
$P_2=\{\gamma_1,\gamma_2\}\{a \rhd ab.\nil\}\{ab \rhd \nil\} \{ab
\rhd \varepsilon.\nil, ab \rhd ab. \nil\}$
\end{tabular}
\end{tabular}
\caption{The partitions computed for $\gamma_1=\varepsilon \rhd
a.ab.\nil + ab.\varepsilon. \nil$ and $\gamma_2=\varepsilon \rhd
a.ab.\nil$.}\label{fig:wordPart}
\end{figure}

\begin{figure}[t]
\begin{tabular}{c@{\hspace{10pt}}c}
$\xymatrix@R=6pt@C=3pt{a \ar[rr]^{\emptyset, \alpha} & & b
\ar@(r,u)_{\$,\beta} \\l \ar[rr]^{\$^3, \alpha} \ar[rrd]_{\emptyset,
\alpha}& & m \ar[rr]^{\emptyset, \beta} &&n\$^2\ar[rr]^{\emptyset,
\beta} && o\$ \ar[rr]^{\emptyset, \beta} &&p \ar@(r,u)|{\$,\beta}\\
& & q \ar@(r,d)|{\$,\beta}&& q\$^1 \ar[ll]_{\emptyset,\beta}&& q\$^2
\ar[ll]_{\emptyset,\beta} && q\$^3 \ar[ll]_{\emptyset,\beta}}$ &
\begin{tabular}{l}
\\
\\
$P_0=\{a,l,b,p,q,q\$^1,o\$,q\$^2,n,q\$^3,m\}$
\\
$P_1=\{a,l\}\{b,p,q\}\{q\$^1,o\$,q\$^2,n,q\$^3,m\}$
\\
$P_2=\{a,l\}\{b,p,q\}\{q\$^1,o\$ \}\{q\$^2,n,q\$^3,m\}$
\\
$P_3=\{a,l\}\{b, p,q\}\{q\$^1,o\$ \}\{q\$^2,n\}\{q\$^3,m\}$
\\
$P_4=\{a,l\}\{b, p,q\}\{q\$^1,o\$ \}\{q\$^2,n\}\{q\$^3,m\}$
\end{tabular}
\end{tabular}\caption{The partitions computed for the marked nets $\<\onet{N_1},a\>$
and $\<\onet{N_4},l\>$.}\label{fig:SeqPart}
\end{figure}

\bigskip
The terminal sequence $1 \tl{}\Fun{N}_{\Fun{T}}(1)\tl{}
\Fun{N}^2_{\Fun{T}}(1)\tl{}\dots$ (where $1$ is a final
$\Gamma(\Cat{C})$-algebra) induces a sequence of
\emph{approximations} of the final morphism from
$\<\Alg{X},\alpha_{\sys{I}};\norm{X}\>$ to $\final{\DMIN}$. The
0-approximation $!_0:\Alg{X}\to 1$ is the unique morphism in
$\Cat{Alg_\Gamma(\Cat{C})}$. The $n+1$-approximation
$!_{n+1}:\Alg{X}\to \Fun{N}^{n+1}_{\Fun{T}}(1)$ is defined as
$\alpha_{\sys{I}};\norm{X}; \DMIN(!_n)$.

%
%
%
%
$$\xymatrix{ \Alg{X}
\ar[r]^{!_n}\ar[d]|{\alpha_{\sys{I}};\norm{X}}\ar[rd]^{!_{n+1}}&
\Fun{N}^n_{\Fun{T}}(1)
\\ \DMIN(\Alg{X}) \ar[r]_{\DMIN(!_n)}&
\Fun{N}^{n+1}_\Fun{T}(1) \ar[u]}$$
In this section, we show that the kernel of the n-approximation
$!_n$ coincides with the partition $P_n$ computed by the algorithm.
Formally, $!_n(p)=!_n(q)$ iff $p$ is equivalent to $q$ in $P_n$.

\begin{prop}\label{prop:SYM->NORM}
Let $\isys$ be a context interactive system. Let $\ded{T}$ and
$\sym$ be, respectively, an inference system and \sts\ that are
sound and complete for $\sys{I}$. Then
$\alpha_{\sys{I}};\norm{X}=\sym;\norm{X}$.
\end{prop}

The above proposition states that the transition systems resulting
from the normalization of the saturated $\alpha_{\sys{I}}$ coincides
with the systems resulting from the normalization of the symbolic
$\sym$. Note that usually $\sym\neq \sym;\norm{X}$, because our
definition of symbolic transition system does not guarantee that
$\sym$ is normalized (according to our definition, also the \cts\ is
a symbolic transition system).
For instance, the symbolic transition system of $\gamma_1$ in Figure
\ref{fig:wordsaturatedSymbolic} (C) is normalized, while the one of
$\varepsilon \rhd u.p+v.p$ in Figure \ref{fig:wordsaturatedSymbolic}
(B) is not.

For all the nets in Figure \ref{fig:OpenNet}, the symbolic
transition system is normalized w.r.t.\ $\Alg{N}$ but, for the net
$S_2$ in Figure \ref{fig:OpenNets1s2}, it is not. Indeed both
$d\atr{\emptyset}{\beta}{\eta}z$ and $d\atr{y}{\beta}{\eta}zy$, and
the former dominates the latter in $\Alg{N}$.
Also in the case of asynchronous $\pi$, the symbolic transition
system $\approxAsy$ is not normalized. Consider the process
$\tau.\outp{c}{d}+ \inp{a}{b}.(\outp{c}{d}|\outp{a}{b})_4$. The
symbolic transition
$\atr{-|\outp{a}{b}}{\tau}{\approxAsy}\outp{c}{d}|\outp{a}{b}_4$ is
dominated by $\atr{-}{\tau}{\approxAsy}\outp{c}{d}_4$.

However, when computing the $n+1$-approximation $!_{n+1}$, we can
simply use $\sym$ instead of $\sym;\norm{X}$. Indeed,
$$\beta;\norm{X};\DMIN(!_n)=\beta;\norm{X};\PDA(!_n);norm_{T,\Fun{N}_{\Fun{T}}^n(1)}=\beta;\PDA(!_n);norm_{T,\Fun{N}_{\Fun{T}}^n(1)}$$
where the former equivalence follows from the definition of
$\DMIN(!_n)$ (Definition \ref{def:normalizedFunctor}) and the latter
follows from Lemma \ref{lemma:normGood}.2 in the Appendix. Thus,
$!_{n+1}=\beta;\PDA(!_n);norm_{T,\Fun{N}_{\Fun{T}}^n(1)}$.

\bigskip

Now we can show by induction that $!_{n+1}(p)=!_{n+1}(q)$ if and
only if $p$ and $q$ belongs to the same partition in $P_{n+1}$.

The base case trivially holds since $!_0:\Alg{X}\to 1$ maps all the
states (with the same interface) into the same element and $P_0$
equates all the states (with the same interface).

For the inductive case, note that by definition $\beta;\PDA(!_n)(p)$
is equal to the set of transitions $(c,o,!_n(p_1))$ such that
$p\atr{c}{o}{\sym}p_1$. Then applying
$norm_{T,\Fun{N}_{\Fun{T}}^n(1)}$ to this set, means to remove all
the transitions $(c,o,!_n(p_1))$ such that there exists a (non
equivalent) transition $(c',o',!_n(p_1'))\in \beta;\PDA(!_n)(p)$
such that $(c',o',!_n(p_1')) \vdash_{T,\Fun{N}_{\Fun{T}}^n(1)}
(c,o,!_n(p_1))$. By Lemma \ref{lemma:PresDer} and Lemma
\ref{lemma:existence} in the Appendix, the latter becomes: there
exists a (non equivalent) transition $(c',o',p_1')\in \beta(p)$ such
that $(c',o',p_1') \vdash_{T,\Alg{X}} (c,o,p_2')$ and
$!_n(p_1)=!_n(p_2')$. By inductive hypothesis, $!_n(p_1)=!_n(p_2')$
iff $p_1$ and $p_2'$ belongs to the same partition in $P_n$. Thus,
the normalization $norm_{T,\Fun{N}_{\Fun{T}}^n(1)}$ junks away all
the redundant transitions in $P_n$.
Summarizing $!_{n+1}(p)$ is equal to the set of transitions
$(c,o,!_n(p_1))$ such that $p\atr{c}{o}{\sym}p_1$ and the latter is
not redundant in $P_n$. Therefore, $!_{n+1}(p)=!_{n+1}(q)$, iff
whenever $p\atr{c}{o}{\sym}p_1$ is \emph{not-redundant in $P_n$},
then $q\atr{c}{o}{\sym}q_1$ is \emph{not-redundant in $P_n$} and
$p_1,q_1$ are equivalent in $P_n$.

\begin{figure}[t]
\begin{tabular}{c}
$\xymatrix@R=6pt@C=5.5pt{ & & \dots \\& & \outp{a}{b}_2 \ar[rr]^{-,
\outp{a}{b}} & & \nil_2 &&\\ \tau.\res{y}{\outp{y}{a}}+
\inp{a}{b}.\outp{a}{b}_1 \ar@(ur,l)[rru]^{-|\outp{a}{b},\tau}
\ar@(u,l)[rruu]^{\dots} \ar[rr]^(.7){-|\outp{a}{a},\tau}
\ar[rrd]_{-, \tau} & & \outp{a}{a}_1
\ar[rr]^{-, \outp{a}{a}} & & \nil_1 &&  \tau.\nil_1 \ar[ll]_{-, \tau} \\
& & \res{y}{\outp{y}{a}}_1 && \res{y}{\outp{y}{a}}|\outp{a}{a}_1 \ar[ll]_{-, \outp{a}{a}} & &  \\
& & \res{y}{\outp{y}{a}}_2 && \res{y}{\outp{y}{a}}|\outp{a}{b}_2
\ar[ll]_{-,\outp{a}{b}} }$\\
$\begin{array}{ll}
\\
P_0= &\{\tau.\nil_1, \tau.\res{y}{\outp{y}{a}}+
\inp{a}{b}.\outp{a}{b}_1,\nil_1,\res{y}{\outp{y}{a}}_1,\res{y}{\outp{y}{a}}|\outp{a}{a}_1,\outp{a}{a}_1\}\\
& \{\nil_2,\res{y}{\outp{y}{a}}_2,\res{y}{\outp{y}{a}}|\outp{a}{b}_2,\outp{a}{b}_2\} \{\nil_3,\res{y}{\outp{y}{a}}_3,\res{y}{\outp{y}{a}}|\outp{a}{c}_3,\outp{a}{c}_3\} \dots\\
P_1= &\{\tau.\nil_1, \tau.\res{y}{\outp{y}{a}}+
\inp{a}{b}.\outp{a}{b}_1\}\{\nil_1,\res{y}{\outp{y}{a}}_1\}\{\res{y}{\outp{y}{a}}|\outp{a}{a}_1,\outp{a}{a}_1\}\\
& \{\nil_2,\res{y}{\outp{y}{a}}_2\}\{\res{y}{\outp{y}{a}}|\outp{a}{b}_2,\outp{a}{b}_2\}\{\nil_3,\res{y}{\outp{y}{a}}_3\}\{\res{y}{\outp{y}{a}}|\outp{a}{c}_3,\outp{a}{c}_3\} \dots\\
P_2= &\{\tau.\nil_1, \tau.\res{y}{\outp{y}{a}}+
\inp{a}{b}.\outp{a}{b}_1\}\{\nil_1,\res{y}{\outp{y}{a}}_1\}\{\res{y}{\outp{y}{a}}|\outp{a}{a}_1,\outp{a}{a}_1\}\\
& \{\nil_2,\res{y}{\outp{y}{a}}_2\}\{\res{y}{\outp{y}{a}}|\outp{a}{b}_2,\outp{a}{b}_2\} \{\nil_3,\res{y}{\outp{y}{a}}_3\}\{\res{y}{\outp{y}{a}}|\outp{a}{c}_3,\outp{a}{c}_3\} \dots\\
\\
\end{array}$
\end{tabular}
\caption{The partitions computed for $\tau.\res{y}{\outp{y}{b}}+ \inp{a}{b}.\outp{a}{b}$ and $\tau.\nil$.}\label{fig:SeqPartAsy}
\end{figure}

\bigskip

We end up this section by showing ``in algorithmic terms'' why
normalized coalgebras are not bialgebras. By virtue of Proposition
\ref{prop:lifting}.2, minimization in bialgebras can be performed,
by first forgetting the algebraic structure, and then minimizing in
$\set$. Minimization in $\coalg{\DMIN}$, instead, heavily relies on
the algebraic structure. Indeed in Figure \ref{fig:wordPart}, the
algorithm needs $ab \rhd ab.\nil$ to compute the partition of
$\gamma_1$. Note that $\gamma_1$ cannot reach through the symbolic
transitions $ab \rhd ab.\nil$, but this is needed for checking if
$\gamma_1 \atr{ab}{\bullet}{\approxwsc}ab \rhd \varepsilon. \nil$ is
redundant. In Figure \ref{fig:SeqPart}, in order to compute the
partitions of $l$, the algorithm needs the state $q\$^3$ that is not
reachable from $l$. Also in Figure \ref{fig:SeqPartAsy}, we need the
state $\res{y}{\outp{y}{a}}|\outp{a}{a}_1$ that is not reachable
from $\tau.\res{y}{\outp{y}{a}}+ \inp{a}{b}.\outp{a}{b}_1$.

Summarizing, given a state $x$ of a normalized coalgebra $\<\Alg{X},\sym;\norm{X}\>$, in order to compute the partitions on the states reachable from $x$,
the algorithm needs ``some'' states that are not reachable 
but that are somehow connect via the algebraic structure $\Alg{X}$ (such as the states described above).
In \cite{ESOP09}, we have shown that the number of the needed ``extra states'' is finite in all the interesting cases
and it can be computed in the initialization phase of the algorithm. 
Moreover, it is important to remark here that $\Alg{X}$ is the only 
algebraic structure that is involved in the algorithm: as described above, 
the normalization $norm_{T,\Fun{N}_{\Fun{T}}^n(1)}$ (at iteration $n+1$) 
can be computed by just using the algebra $\Alg{X}$.

\section{Conclusions and related works}\label{sec:conclusion}
\noindent The paper introduces two coalgebraic models for context interactive
systems \cite{FOSSACS08}. In the first one, the saturated transition
system is an ordinary structured coalgebra
$\<\Alg{X},\alpha_{\sys{I}}\>$ and its final morphism induces
$\satbis$. The second model is the normalized coalgebra
$\<\Alg{X},\alpha_{\sys{I}};\norm{X}\>$ that is obtained by pruning
all the redundant transitions from the first one. The equivalence
induced by its final morphism is still $\satbis$, but this
characterization is much more convenient. Indeed, in the final
normalized coalgebra all the (semantically) redundant transitions
are eliminated. Moreover, minimization is usually feasible with
normalized coalgebras and coincides with the symbolic minimization
algorithm introduced in \cite{ESOP09}.

As a lateral result, we have obtained coalgebraic models for both
open Petri nets and asynchronous $\pi$-calculus.

Unfortunately, symbolic minimization is unfeasible in the case of
asynchronous $\pi$, because the symbolic transition system is
infinite. Indeed, in the definition of $\sim^1$ (Definition
\ref{def:1bis}), a process is put in parallel with all possible
outputs $\outp{a}{b}$. Our symbolic transition system eliminates all
those outputs whose subjects $a$ are not needed, but yet it
considers all the possible objects $b$. We could have defined a
different \sts\ that considers only those objects that are strictly
needed but, anyway, in the asynchronous $\pi$-calculus there are
several other sources of infiniteness. Amongst these, one always
appears when considering ``nominal calculi'' where systems are able
to generate and communicate names: every time that a system
generates a new name and extrudes it, the system goes in a new state
that is different from all the previous. HD-Automata \cite{HD} are
peculiar \lts s that allow to garbage collect names and avoid this
further source of infiniteness. As future work, we would like to
extend our framework to HD-Automata, so that we will be able to
handle systems that generates infinitely many names. In particular
we conjecture that the resulting minimization algorithm will
generalize both \cite{HDN} and \cite{PistoreSangiorgi} that provide
a partition refinement algorithm for asynchronous and open
bisimilarity. The reader is referred to \cite{ESOP09} for a more
detailed comparison with \cite{HDN} and \cite{PistoreSangiorgi}.

Concerning open bisimilarity \cite{Sang96}, a coalgebraic model has
been proposed in \cite{GhaniYV04}. However, this is the saturated
version, i.e., the one that takes into account all the possible
substitutions. In \cite{FOSSACS08}, we have given a context
interactive system for open $\pi$-calculus, and thus our work also
provides a coalgebraic model for the ``efficient characterization''
of open bisimilarity.

Besides open Petri nets, asynchronous and open $\pi$-calculus,
context interactive systems also generalize \emph{Leifer and
Milner's reactive systems} \cite{RobinCONCUR00}. The main novelty of
our framework consists in having observations and inference rules.
The latter generalize the notion of \emph{reactive contexts} of
\cite{RobinCONCUR00}. Indeed $c$ is reactive iff the following
inference rule holds.
$$\deduz{p\tr{\tau}q}{c(p)\tr{\tau}c(q)}$$
Concretely, the main advantage of our framework w.r.t.\
\cite{RobinCONCUR00} is that we do not need the existence of RPOs
and thus we can avoid those encodings into bigraphs
\cite{RobinCONCUR01} and borrowed contexts \cite{BARBARAFOSSACS04}.
The main disadvantage is that our framework does not provide a
constructive definition for the \lts : constructing a sound and
complete symbolic transition system is left to the ingenuity of the
researcher. We refer the reader to \cite{FOSSACS08,BonchiThesis} for
a detailed comparison between context interactive systems and
reactive systems.

In \cite{FOSSACS09}, the first author together with Gadducci and
Monreale has shown a reactive system for mobile ambients
\cite{Ambient}. Thus, the present work indirectly provides also a
coalgebraic semantics for mobile ambients. A coalgebraic model for
this calculus has been previously proposed in \cite{HausmannMS06}
but it characterizes \emph{action bisimilarity} that is strictly
included into \emph{reduction barbed congruence}
\cite{MerroNardelli}. Action bisimilarity is defined as the ordinary
bisimilarity on the symbolic transition system and thus it is an
instance of what we have called syntactic bisimilarity ($\sim^W$).

Besides their large applicability, normalized coalgebras are
interesting for a more theoretical reason: at our knowledge, these
are the first example in literature of structured coalgebras that
are not bialgebras. Indeed, both the definitions of saturated and
normalized set of transitions (Definition \ref{def:saturatedset} and
\ref{def:normalization}, respectively) strongly rely on the
underlying algebraic structures. 
This is evident in the minimization algorithm in
$\coalg{\DMIN}$ that heavily employs the algebraic structure.

\section*{Acknowledgement}

\noindent The authors would like to thank the anonymous referees (of this and
the preliminary version in \cite{CALCO09}) for the precious comments
that have improved the quality of the paper.


\appendix

\section{Normalizable Systems}\label{appendixNormalizable}

\noindent In this appendix we show that the context interactive system
$\modelON = \modelONAll$ is normalizable w.r.t.\ $\ded{T_{\modelON}}$
and that $\modelAsy=\modelAsyAll$ is normalizable w.r.t.\
$\ded{T_{\modelAsy}}$ (all these are defined in Section
\ref{sec:CIS}). Then we show an example of a not normalizable
systems.

\begin{prop}
$\modelON = \modelONAll$ is normalizable w.r.t.\
$\ded{T_{\modelON}}$.
\end{prop}
\begin{proof}
Recall that arrows of $\onmssign$ are multisets (over sets of input
places) and that $c;d=e$ if and only if $c\oplus d=e$. Then, for all
$\Gamma(\onmssign)$-algebra $\Alg{Y}$, $$(c_1,\lambda,
x_1)\prec_{\ded{T_{\modelON}},\Alg{Y}}(c_2,\lambda',x_2)$$ only if
the multiset $c_1$ is strictly included into the multiset $c_2$.
Since all multisets are finite also the descending chains must be
finite.
\end{proof}

\begin{prop}
$\modelAsy = \modelAsyAll$ is normalizable w.r.t.\
$\ded{T_{\modelAsy}}$.
\end{prop}
\begin{proof}
Recall that arrows of $\amssign$ are contexts representing parallel
output processes and that $c;d=e$ if and only if $e$ is the
syntactic composition of $c$ with $d$. Then, for all
$\Gamma(\amssign)$-algebra $\Alg{Y}$,
$$(c_1,o,
q_1)\prec_{\ded{T_{\modelAsy}},\Alg{Y}}(c_2,o',q_2)$$ only if the
context $c_2$ is the parallel composition of $c_1$ with some other
outputs. Since all contexts are finite then the descending chains
must be also finite.
\end{proof}

\begin{exa}
As an example of not normalizable context interactive system
consider the category $\Cat{NAT_{\geq}}$ defined as follow:
\begin{enumerate}[$\bullet$]
\item objects are natural numbers and $\infty$,
\item there is an arrow $n \to m$, if $n \geq m$ or if $n=\infty$.
\end{enumerate}
Since for any two objects $n,m$ there is only one arrow, we call
this arrow just as $n \to m$. Consider now a context interactive
system $\sys{NAT_{\geq}}=\< \Cat{NAT_{\geq}},\Alg{X},O,tr\>$ for
some $\Alg{X},O,tr$. Let $T$ be the tile system that states that all
contexts preserve transitions.

We have that $\sys{NAT_{\geq}}$ is not normalizable with respect to
$T$. Indeed, let $\Alg{F}$ be the final $\Gamma(\Cat{NAT_{\geq}})$.
In this algebra there is only one element $\star$ for each sort
(natural number), and an arrow $n \to m$ of $\Cat{NAT_{\geq}}$ is
interpreted in the function mapping $\star$ of sort $n$ into $\star$
of sort $m$. Since $\infty \to n$ can be decomposed in $\infty \to
n+1 \to n$, then $$(\infty \to n+1, l, \star)
\prec_{T,\Alg{F}}(\infty \to n, l, \star)\text{.}$$ This trivially
leads to an infinite descending chain.
\end{exa}

\section{Proofs of Section \ref{sec:CIS}}

\begin{oldprop}{prop:coarsest}
$\satbis$ is the coarsest bisimulation congruence.
\end{oldprop}
\begin{proof}
Since $\satbis$ is a saturated bisimulation, then it is also a
congruence: if $p\satbis q$, then for all contexts $c_{\Alg{X}}$, it
holds that $c_{\Alg{X}}(p)\satbis c_{\Alg{X}}(q)$.
%
%

In order to prove that it is the coarsest bisimulation congruence,
we prove that any bisimulation congruence $R$ is a saturated
bisimulation.

Suppose that $p \; R \; q$. Suppose that $c_{\Alg{X}}(p)\tr{o}p'$.
Since $R$ is a congruence, then $c_{\Alg{X}}(p)\; R\;
c_{\Alg{X}}(q)$. Since $R$ is a bisimulation
$c_{\Alg{X}}(q)\tr{o}q'$ and $p' R q'$. Thus $R$ is a saturated
bisimulation.
\end{proof}

\begin{oldprop}{prop:netsaturated}
Let $\<\onet{N_1},m_1\>$ and $\<\onet{N_2},m_2\>$ be two marked nets
both with interface $I$. Thus $\<\onet{N_1},m_1\> \sim^{N}
\<\onet{N_2},m_2\>$ iff $\<\onet{N_1},m_1\>
\sim^{S}_I\<\onet{N_2},m_2\>$.
\end{oldprop}
\begin{proof}
The definition of $\satbis$ instantiated to the context interactive
system $\modelON$, requires that $\<\onet{N_1},m_1\>$ and
$\<\onet{N_2},m_2\>$ (a) make the same transitions with the rule
$\rulelabel{tr}$ and (b) they are still equivalent when adding
multisets $i\in I^{\oplus}$. The definition of $\sim^{N}$ instead
requires that the two nets perform the same transitions with both
the rule $\rulelabel{tr}$ and the rule $\rulelabel{in}$. But the
latter rule just adds multisets $i\in I^{\oplus}$ and thus it is
just the same of point (b) above.

\end{proof}

\begin{oldprop}{prop:asysaturated}
Let $p,q$ be asynchronous $\pi$-processes, and let $n\geq max
\;\fn{p \cup q}$. Then $p\sim^{1} q$ iff $p_{n}\satbis_{n} q_{n}$.
\end{oldprop}
\begin{proof}
Let $R=\{(p,q)\; \mid\; p_n\satbis_{n} q_n\;\;n\geq \;max \;\fn{p
\cup q} \}$.
In order to prove that $p_n\satbis_{n} q_n$ implies $p\sim^{1} q$,
we prove that $R$ is an 1-bisimulation, i.e., an
$o\tau$-bisimulation closed under composition with output processes.
%
Suppose that $p \tr{\boutb{i}{j}}p'$ (the cases of $\tau$ and output
are easier). First of all observe that $p_n\satbis_n q_n$ implies
that $\forall m \geq n$ $p_m\satbis_m q_m$. Now since $j$ is fresh,
we have that $j-1\geq n$, and thus $p_{j-1}\satbis_{j-1} q_{j-1}$.
By definition of $\trAsy$, we have that
$p_{j-1}\asytr{\boutb{i}{}}p'_{j}$ and, since $p_{j-1} \satbis_{j-1}
q_{j-1}$, it follows that $q_{j-1}\asytr{\boutb{i}{}}q'_{j}$ and
$p'_j \satbis_j q'_j$ and then, $p'Rq'$. Again by definition of
$\trAsy$, we have that $q\tr{\boutb{i}{j}}q'$. This prove that $R$
is an $o\tau$-bisimulation. Now we have to prove that it is closed
under composition with output processes, but this is immediate since
$\satbis$ is a congruence w.r.t.\ composition with output processes.

Let $R$ be the $\nat$-sorted relation, such that $\forall n \in
\nat$, $R_n=\{(p_n,q_n)\; \mid\; p\sim^{1} q,\;\;n\geq \;max \;\fn{p
\cup q} \}$.
In order to prove that $p\sim^{1} q$ implies $p_n\satbis_{n} q_n$,
we prove that $R$ is a saturated bisimulation.
Let $c\in \amssign[n,m]$ and suppose that
$c(p_n)_m\asytr{\boutb{i}{}}p'_{m+1}$ (the case of $\tau$ and output
are easier). By definition of $\trAsy$, $c(p)\tr{\boutb{i}{m+1}}p'$.
Now, since $p \sim^{1}q$, by definition of 1-bisimulation, it
follows that $c(p) \sim^{1}c(q)$ because contexts $c$ are just
parallel output processes. Thus $c(q)\tr{\boutb{i}{m+1}}q'$ and
$p'\sim^{1}q'$. By definition of $\trAsy$, it follows that $c(q_n)
\asytr{\boutb{i}{}}q'_{m+1}$ and, by definition of $R$, that
$p'_{m+1}R_{m+1}q'_{m+1}$.
\end{proof}

\begin{oldthm}{theo:main}
Let $\sys{I}$ be a context interactive system, $\beta$ a context
transition system and $\infsys$ an inference system. If $\beta$ and
$\infsys$ are sound and complete w.r.t.\ $\sys{I}$, then $\symbis =
\satbis$.
\end{oldthm}
\begin{proof}
Let $R=\{R_i \subseteq X_i \times X_i \mid i \in \sort \}$ be the
$\sort$-sorted family of relations, such that $\forall j \in \sort$,
$$R_j=\{(c(p_i),c(q_i))\;\mid\;c \in \Cat{C}[i,j],\;\; p_i \symbis_i
q_i\}\text{.}$$ In order to prove that $\symbis \subseteq \satbis$
we prove that $R$ is a saturated bisimulation.
Suppose that $a_j R_j b_j$ thus there exists $c\in \Cat{C}[i,j]$
such that $c(p_i)=a_j$, $c(q_i)=b_j$ and $p_i \symbis_i q_i$.

Hereafter, in order to make lighter the notation, we avoid to
specify the sort of processes and contexts. Thus, $p,q,a,b$ stand
for, respectively, $p_i,q_i,a_j,b_j$.

Suppose that $d(a)=d(c(p)) \tr{l_1}p_1$ then, by definition of \cts,
$p \satr{c;d}{l_1} p_1$. By completeness of $\beta$ and $\infsys$,
we have that $p \atr{c_2}{l_2}{\sym} p_2$ such that $p
\atr{c_2}{l_2}{\sym} p_2 \vdash_{\infsys} p \atr{c;d}{l_1}{} p_1$,
i.e., $\exists e,e_1\in ||\Cat{C}||$ such that:
\begin{enumerate}[$\bullet$]
\item $e\UTtr{l_2}{l_1}e_1\in \Phi(\infsys)$,
\item $c_2;e=c;d$,
\item $e_1(p_2)=p_1$.
\end{enumerate}

Since $p \symbis q$, $q\atr{c_3}{l_3}{\sym}q_3 \vdash_{\infsys}
q\atr{c_2}{l_2}{}q_2$ and $p_2 \symbis q_2$. From the former we have
that $\exists f, f_1\in ||\Cat{C}||$ such that:
\begin{enumerate}[$\bullet$]
\item $f\UTtr{l_3}{l_2}f_1\in \Phi(\infsys)$,
\item $c_3;f=c_2$,
\item $f_1(q_3)=q_2$.
\end{enumerate}

Since $\Phi(\infsys)$ is closed by composition, then
$f;e\UTtr{l_3}{l_1}f_1;e_1\in \Phi(\infsys)$. Moreover
$c_3;f;e=c_2;e=c;d$. Thus $q\atr{c_3}{l_3}{\sym}q_3 \vdash_T
q\atr{c;d}{l_1}{}e_1(q_2)$. Since $\sym$ and $\infsys$ are sound, it
follows that $q\satr{c;d}{l_1}e_1(q_2)$, i.e.,
$d(b)\tr{l_1}e_1(q_2)$. Since $p_2 \symbis q_2$, then $e_1(p_2)\;R
\;e_1(q_2)$, i.e., $p_1\;R \;e_1(q_2)$.

For proving that $\satbis \subseteq \symbis$, take $p\satbis q$: if
$p \atr{c}{l}{\sym}p_1$ then also $p \satr{c}{l}p_1$ and, since
$p\satbis q$, $q \satr{c}{l}q_1$ with $p_1\satbis q_1$. By
completeness of $\sym$, we have that $q \atr{c_1}{l_1}{\sym}q_1
\vdash_T q \satr{c}{l}q_1$.
\end{proof}

\begin{oldprop}{prop:wordsoundcomplete}
$\approxwsc$ and $\infwsc$ are sound and complete w.r.t.\
$\modelswc$.
\end{oldprop}
\begin{proof}
Proving soundness is quite easy. Just observe that (1) if $u\rhd p
\atr{v}{\bullet}{\approxwsc} u'\rhd p'$ then $uv \rhd p \tr{\bullet}
u'\rhd p'$ and (2) the ``monotonicity property'' described in
Section \ref{sec:swc} holds.

For proving completeness we suppose that $w_{\WSCA}(\gamma)
\tr{\bullet}\gamma''$ (i.e., $\gamma \atr{w}{\bullet}{S}\gamma''$)
and we proceed by induction on the structure of the process of the
configuration $\gamma$. The inductive case for $\gamma= u \rhd
p_1+p_2$ is trivial. The base case is as follows.

Take $\gamma= u_1 \rhd u_2.p$ (thus $w_{\WSCA}(\gamma) = u_1w \rhd
u_2.p$). Note that $\gamma''$ must be equal to $u_1w \rhd p$ (by the
rules defining $\trwsc$). Since the configuration $u_1w \rhd u_2.p$
perform a transition then $u_2$ is a prefix of $u_1w$ (again by the
rules defining $\trwsc$). There are two possible cases: either $u_2$
is a prefix of $u_1$ or not.

In the former case, by the leftmost rule (defining $\approxwsc$),
$\gamma=u_1 \rhd u_2.p \atr{\varepsilon}{\bullet}{\approxwsc}u_1
\rhd p$ and by definition of $\vdash_{\infwsc}$, this transition
derives $\gamma \atr{w}{\bullet}{}w_{\WSCA}(u_1 \rhd p)=u_1w \rhd p
= \gamma''$.

If $u_2$ is not a prefix of $u_1$, then there exists $u,v\in A^*$
such that $w=vu$ and $u_2=u_1v$. By the central rule (defining
$\approxwsc$), we have that $\gamma=u_1 \rhd u_1v.p
\atr{v}{\bullet}{\approxwsc}u_1v \rhd p$ and, by definition of
$\vdash_{\infwsc}$, this transition derives $\gamma
\atr{vu}{\bullet}{}u_{\WSCA}(u_1v \rhd p)=u_1w \rhd p = \gamma''$.
\end{proof}

\begin{oldprop}{prop:netSymbolic}
Let $\<\onet{N_1},m_1\>$ and $\<\onet{N_2},m_2\>$ be two marked nets
both with interface $I$. Thus $\<\onet{N_1},m_1\> \sim^{NS}
\<\onet{N_2},m_2\>$ iff $\<\onet{N_1},m_1\>
\sim^{SYM}_I\<\onet{N_2},m_2\>$.
\end{oldprop}
\begin{proof}
The general condition of symbolic bisimilarity
\begin{enumerate}[$\bullet$]
\item if $p \atr{c}{o}{\sym} p'$,
then $q\atr{c_1}{o_1}{\sym}q_1'$ and
$q\atr{c_1}{o_1}{\sym}q_1'\vdash_{\infsys} q\atr{c}{o}{}q'$ and $p'
R_k q'$.
\end{enumerate}
for the context interactive system $\modelON$, the $\sts$ $\eta$ and
the inference system $\infON$, becomes
\begin{enumerate}[$\bullet$]
\item if $ \<N_1,m_1\> \atr{i}{\lambda}{\eta} \<N_1,m'_1\>$,
then $\<N_2,m_2\>\atr{j}{\lambda_1}{\eta}\<N_2,m_2'\>$ and
$\<N_2,m_2\>\atr{j}{\lambda_1}{\eta}\<N_2,m_2'\> \vdash_{\infON}
\<N_2,m_2\>\atr{i}{\lambda}{\eta}\<N_2,m_2''\>$ and $\<N_1,m_1'\>
R_I \<N_2,m_2''\>$.
\end{enumerate}
From the latter, we have that $\lambda_1=\lambda$ and there exists
$k\in I^{\oplus}$ such that $i=j\oplus k$ and $m_2''=m_2'\oplus k$.
These are the conditions of net-symbolic bisimilarity.
\end{proof}

\begin{oldprop}{prop:symboliOpenNet}
$\approxON$ and $\infON$ are sound and complete w.r.t.\ $\modelON$.
\end{oldprop}
\begin{proof}
We have to prove:
\begin{enumerate}[$\bullet$]
\item(completeness) if $\<N,m\> \satr{i}{\lambda}
\<N,m'\>$ then 
\[\<N,m\> \atr{i_1}{\lambda_1}{\approxON}\<N,m_1\>
\quad\hbox{and}\quad\<N,m\> \atr{i_1}{\lambda_1}{\approxON}\<N,m_1\>
\vdash_{\infON} \<N,m\> \atr{i}{\lambda}{} \<N,m'\>.
\]
\item(soundness) if $\<N,m\> \atr{i_1}{\lambda_1}{\approxON}\<N,m_1\>$
and $\<N,m\> \atr{i_1}{\lambda_1}{\approxON}\<N,m_1\>
\vdash_{\infON} \<N,m\> \atr{i}{\lambda}{} \<N,m'\>$ then $\<N,m\>
\satr{i}{\lambda} \<N,m'\>$.
\end{enumerate}

\noindent Let us prove completeness.
If $\onet{N},m\oplus i \ontr{\lambda}\onet{N},m'$, then there exists
a transition $t\in T$, such that $\lambda(t)=l$ and $m \oplus i
=\pre{t} \oplus c$ and $m'=\post{t}\oplus c$. We can take $c_1=m
\ominus (m \cap \pre{t})$ and $i_1=\pre{t}\ominus(m \cap \pre{t})$.
and apply the only rule of $\approxON$, and
$\onet{N},m\atr{i_1}{\lambda}{\approxON}\onet{N},\post{t}\oplus
c_1$. Note that $i_1 \subseteq i$, since by definition $i_1$ is the
smallest multiset that allow the transition $t$. Thus let $x = i
\ominus i_1$. We have
$\onet{N},m\atr{i_1}{\lambda}{\approxON}\onet{N},\post{t}\oplus c_1
\vdash_{\infON} N,m \atr{i}{\lambda}{\infON(\approxON)}\onet{N},m'$.
Indeed:
\begin{enumerate}[$\bullet$]
\item $i_1 \oplus x= i$;
\item $\post{t}\oplus c_1 \oplus x = m'$, because $c_1\oplus x =
m\ominus(\pre{t} \cap m) \oplus x = m\oplus \pre{t}\ominus( \pre{t}
\cap m) \oplus x \ominus \pre{t} = m\oplus i_1 \oplus x \ominus
\pre{t}= m\oplus i  \ominus \pre{t} =c$.
\end{enumerate}

\noindent For proving soundness observe that if $\<N,m\>
\atr{i_1}{\lambda_1}{\approxON}\<N,m_1\>$ then $\<N,m\oplus i_1\>
\tr{\lambda_1}{}\<N,m_1\>$. Moreover, if $\<N,m\>
\atr{i_1}{\lambda_1}{\approxON}\<N,m_1\> \vdash_{\infON} \<N,m\>
\atr{i}{\lambda}{} \<N,m'\>$, then $\lambda_1=\lambda$ there exists
$x \in I^{\oplus}$ such that $i_1\oplus x =i$ and $m_1\oplus x =m'$.

Thus, $\<N,m\oplus i_1\oplus x\> \tr{\lambda_1}{}\<N,m_1\oplus x\>$,
that means $\<N,m\> \satr{i}{\lambda} \<N,m'\>$.
\end{proof}

\begin{oldprop}{prop:AsynSymbolic}
Let $p,q$ be asynchronous $\pi$-processes, and let $n\geq max
\;\fn{p \cup q}$. Then $p\sim^{a} q$ iff $p_{n}\symbis_{n} q_{n}$.
\end{oldprop}
\begin{proof}
Here we prove that if $p_n\symbis_n q_n$ then $p \sim^a q$ (the
other implication is analogous).

Let $R=\{p,q \mid p_n \symbis_n q_n\}$ be a symmetric relation. We
prove that $R$ is an asynchronous bisimulation.

Take $p_n \symbis_n q_n$ and suppose that $p \tr{\boutb{i}{j}}p'$
and $j$ is fresh. First, observe that $\forall m \geq n$, $p_m
\symbis_m q_m$. Then, note that since $j$ is fresh, $j-1\geq n$ and
thus $p_{j-1} \symbis_{j-1} q_{j-1}$.

By definition of $\trAsy$, $p_{j-1}
\atr{-}{\boutb{i}{}}{\approxAsy}p'_{j}$. Now since $p_{j-1}
\symbis_{j-1} q_{j-1}$, $q_{j-1}$ must answer with a transition
$q_{j-1} \atr{c}{o}{\approxAsy} q''$ such that $q_{j-1}
\atr{c}{o}{\approxAsy} q'' \vdash_{\infAsy} q_{j-1}
\atr{-}{\boutb{i}{}}{} q'_{j}$ and $p'_{j}\symbis_{j}q'_{j}$. By
definition of $\infAsy$, the only such transition is $q_{j-1}
\atr{-}{\boutb{i}{}}{\approxAsy}q'_{j}$. Now, by definition of
$\approxAsy$, we have that $q \tr{\boutb{i}{j}}q'$ and, by
definition of $R$, $p'Rq'$.
We can proceed analogously in the case of output and $\tau$.

For the input, suppose that $p \tr{\inp{i}{j}}p'$. Then $p_n \atr{-|
\outp{i}{j}}{\tau}{\approxAsy} p'_{n'}$ where $n'= max \; \{j,n\}$.
Now since $p_n \symbis q_n$, $q_n$ must answer with a transition
$q_n \atr{c}{o}{\approxAsy} q''$ such that $q_n
\atr{c}{o}{\approxAsy} q'' \vdash_{\infAsy} q_n \atr{-|
\outp{i}{j}}{\tau}{\approxAsy} q'_{n'}$ and $p'_{n'}\symbis_{n'}
q'_{n'}$.

By definition of $\infAsy$ there are two possibilities:
\begin{enumerate}[$\bullet$]
\item $q_n \atr{-| \outp{i}{j}}{\tau}{\approxAsy}q'_{n'}$ and $p'_{n'}
\symbis_{n'} q'_{n'}$. Thus $q \tr{\inp{i}{j}}q'$ and $p'Rq'$.
\item $q_n \atr{-}{\tau}{\approxAsy}q''_n$ and by using the rule
$\rulelabel{tau}_{-| \outp{i}{j}}$, $q_n \atr{-}{\tau}{\approxAsy}q''_n
\vdash_{\infAsy} q_n \atr{-| \outp{i}{j}}{\tau}{\approxAsy}q''_n |
\outp{i}{j}$ and $p'_{n'} \symbis_{n'} (q''| \outp{i}{j})_{n'}$.
Thus $q \tr{\tau}q''$ and $p' R q'' | \outp{i}{j}$.
\end{enumerate}
Note that it is correct to write $q'_{n'}$, since
$\fn{q'}=\fn{q|\outp{i}{j}}\subseteq n'$. The same holds also for
$(q''| \outp{i}{j})_{n'}$: $\fn{q''}= \fn{q}\subseteq n$ and thus
$\fn{q''| \outp{i}{j}}\subseteq n'$.
\end{proof}

\begin{oldprop}{prop:AsynchronousSoundandComplete}
$\approxAsy$ and $\infAsy$ are sound and complete w.r.t.\
$\modelAsy$.
\end{oldprop}
\begin{proof}
We have to prove:
\begin{enumerate}[$\bullet$]
\item(completeness) if $p_n \satr{c}{\mu}
q_m$ then $p_n \atr{c'}{\mu'}{\approxAsy}q'_{m'}$ and
$p_n\atr{c'}{\mu'}{\approxAsy}q'_{m'} \vdash_{\infAsy} p_n
\satr{c}{\mu} q_m$.
\item(soundness) if $p_n \atr{c'}{\mu'}{\approxAsy}q'_{m'}$ and
$p_n\atr{c'}{\mu'}{\approxAsy}q'_{m'} \vdash_{\infAsy} p_n
\satr{c}{\mu} q_{m}$ then $p_n \satr{c}{\mu} q_m$.
\end{enumerate}

%
For soundness just observe that if $p\atr{c}{\mu}{\approxAsy}p'$
then $c(p)\tr{\mu}p'$ and that all the rules of $\infAsy$ are sound.
Let us prove completeness.
Suppose that $\mu=\tau$ (the other cases are easier): $p_n
\satr{c}{\tau} q_m$ implies that $c(p_n)_m \tr{\tau}q_m$ and $c\in
\amssign[n,m]$. By definition of $\trAsy$, it follows that
$c(p)\tr{\tau}q$. Since $c$ could be only the parallel composition
of outputs, by the definition of the operational semantics of
asynchronous $\pi$, it follows that either $p \tr{\tau} q'$ (such
that $q=c(q')$) or $c\tr{\outp{i}{j}}c'$ (where $c= -| \outp{i}{j} |
c'$) and $p\tr{\inp{i}{j}}q'$ (such that $q=c'| q'$).

In the former case, by definition of $\approxAsy$, we have that
$p_n\atr{-}{\tau}{\approxAsy}q'_n$ and using the rule
\rulelabel{tau}$_{c}$ of $\infAsy$, we have that
$p_n\atr{-}{\tau}{\approxAsy} q'_n \vdash_{\infAsy}
p_n\atr{c}{\tau}{\approxAsy}c(q'_n)_m=q_m$.
In the latter case, by definition of $\approxAsy$, we have that
$p_n\atr{-| \outp{i}{j}}{\tau}{\approxAsy}q'_{n'}$ where $n'=
max\{j,n\}$. Now, take $c'\in \amssign[n',m]$, by the rule
\rulelabel{tau}$_{c'}$ of $\infAsy$, we have that $p_n\atr{-|
\outp{i}{j}}{\tau}{\approxAsy}q'_{n'} \vdash_{\infAsy} p_n\atr{-|
\outp{i}{j}| c'}{\tau}{}c'(q'_{n'})_m=q_m$.
\end{proof}

\section{Proofs of Section \ref{sec:CIScoalgebra}}
\begin{oldthm}{theoDMcoalgebra}
Let $\isys$ be a context interactive system. Then $\<\Alg{X},
\alpha_{\sys{I}}\>$ is a $\PDA$-coalgebra.
\end{oldthm}
\begin{proof}
We have to prove that $\alpha_{\sys{I}}: \Alg{X} \to \PDA(\Alg{X})$
is a $\Gamma(\Cat{C})$-homomorphism, i.e., that $\forall x \in X$
and $\forall d \in \Gamma(\Cat{C})$,
$\alpha_{\sys{I}}(d_{\Alg{X}}(x))=
d_{\TPDA(\Alg{X})}(\alpha_{\sys{I}}(x))$.

Let $(c,l,y)\in \alpha_{\sys{I}}(d_{\Alg{X}}(x))$ be a saturated
transition of $d_{\Alg{X}}(x)$. Then by definition of
$\alpha_{\sys{I}}$, $(d;c,l,y) \in \alpha_{\sys{I}}(x)$. By
definition of $d_{\TPDA(\Alg{X})}$ and by $(d;c,l,y) \in
\alpha_{\sys{I}}(x)$, follows that $(c,l,y)\in
d_{\TPDA(\Alg{X})}(\alpha_{\sys{I}}(x))$.

Now let $(c,l,y)\in d_{\TPDA(\Alg{X})}(\alpha_{\sys{I}}(x))$. 
By definition of $d_{\TPDA(\Alg{X})}$ we have that $(d;c,l,y)\in
\alpha_{\sys{I}}(x)$ and, analogously to before, $(c,l,y)\in
\alpha_{\sys{I}}(d_{\Alg{X}}(x))$.
\end{proof}

\section{Proofs of Section \ref{sec:saturatedcoalgebra}}\label{proofSatu}
\noindent Before proving Lemma \ref{lemma:Icorrect} and Proposition
\ref{prop:covariety}, we prove some important lemmas about the
derivation relation $\vdash_{\ded{T}, \Alg{X}}^d$. Moreover, at the
end of this appendix we formally show the existence of the final
object in $\coalg{\DMIS}$.

\begin{lem}[composition of
$\vdash_{\ded{T},\Alg{X}}^d$]\label{lemma:compositionDer}\hfill

\noindent If $(c,l,x)\vdash_{\ded{T}, \Alg{X}}^d (c',l',x') \vdash_{\ded{T},
\Alg{X}}^e (c'',l'',x'')$ then $(c,l,x) \vdash_{\ded{T},
\Alg{X}}^{d;e} (c'',l'',x'')$.
\end{lem}
\begin{proof}
From the hypothesis we derives that there exists $d',d'',e',e'' \in
||\Cat{C}||$ such that $d;c'=c;d'$ and $e;c''=c';e'$ and $d' \Ttr{l }{l'}d''$
and $e' \Ttr{l'}{l''}e''$ such that $d''_{\Alg{X}}(x)=x'$ and
$e''_{\Alg{X}}(x')=x''$. From all this, we derive that
$(d;e);c''=c;(d';e')$ and that $d';e' \Ttr{l }{l''}d'';e''$ and that
$e''_{\Alg{X}}(d''_{\Alg{X}}(x))=x''$. Then the thesis immediately
follows.
\end{proof}

\begin{lem}[$\vdash_{\ded{T},\Alg{X}}^d$ is preserved by
homomorphisms]\label{lemma:PresDer}\hfill

\noindent Let $h: \Alg{X} \to \Alg{Y}$ be a $\Gamma (\Cat{C})$-homomorphism.
If $(c,l,x)\vdash_{\ded{T}, \Alg{X}}^d (c',l',x')$, then
$(c,l,h(x))\vdash_{\ded{T}, \Alg{Y}}^d (c',l',h(x'))$.
\end{lem}
\begin{proof}
If $(c,l,x)\vdash_{\ded{T}, \Alg{X}}^d (c',l',x')$, then there
exists $d'\in ||\Cat{C}||$ such that $d;c'=c;d'$ and $d' \Ttr{l}{l'}d''$ and
$d''_{\Alg{X}}(x)=x'$. Since $h$ is an homomorphism
$h(x')=h(d''_{\Alg{X}}(x))=d''_{\Alg{Y}}(h(x))$, and then
$(c,l,h(x))\vdash_{\ded{T}, \Alg{Y}}^d (c',l',h(x'))$.
\end{proof}

\begin{lem}[$\vdash_{\ded{T},\Alg{X}}^d$ is reflected by
homomorphisms]\label{lemma:existence}\hfill

\noindent Let $h: \Alg{X} \to \Alg{Y}$ be a $\Gamma (\Cat{C})$-homomorphism.
If $(c,l,h(x))\vdash_{\ded{T}, \Alg{Y}}^d (c',l',y')$, then $\exists
x' \in \Alg{X}$, such that $h(x')=y'$ and $(c,l,x)\vdash_{\ded{T},
\Alg{X}}^d (c',l',x')$.
\end{lem}
\begin{proof}
From the hypothesis we derive that there exists $f\in ||\Cat{C}||$ such
that $c;f=d;c'$ and $f\Ttr{l}{l'}f'$ and $f'_{\Alg{Y}}(h(x))=y'$.
Since $h$ is an homomorphism, $h(f'_{\Alg{X}}(x))=y'$. Then we have
that $(c,l,x)\vdash_{\ded{T}, \Alg{X}}^d (c',l',f'_{\Alg{X}}(x))$.
\end{proof}

\begin{prop}
$\DMIS:\Cat{Alg_{\Sig{\Gamma(\Cat{C})}}} \to
\Cat{Alg_{\Sig{\Gamma(\Cat{C})}}}$ is a functor.
\end{prop}
\begin{proof}
First of all, we have to show that $\forall \Alg{X}\in
|\Cat{Alg_{\Sig{\Gamma(\Cat{C})}}}|$, $\DMIS(\Alg{X})\in
|\Cat{Alg_{\Sig{\Gamma(\Cat{C})}}}|$. Notice that all the operators
$d_\DMIS(\Alg{X})$ are well defined, i.e., $\forall A \in \PS(X)$,
$d_\DMIS(\Alg{X})(A)$ is still a saturated set of transitions, i.e.,
it is closed w.r.t.\ $\vdash_{\ded{T},\Alg{X}}^{id}$. Then we have to
prove that $id_{\DMIS(\Alg{X})}$ coincides with the identity
function. This is trivial since $id_{\DMIS(\Alg{X})}(A)$ consists in
closing the set of transition $A$ w.r.t.\
$\vdash_{\ded{T},\Alg{X}}^{id}$. But since $A$ is saturated , it is
already closed. Finally we have to prove that
$(c;d)_{\DMIS(\Alg{X})}=c_{\DMIS(\Alg{X})};d_{\DMIS(\Alg{X})}$, but
this is trivial consequence of Lemma \ref{lemma:compositionDer}.

Then we have to prove that if $h:\Alg{X}\to \Alg{Y}$ in
$\Cat{Alg_{\Sig{\Gamma(\Cat{C})}}}$, then also
$\DMIS(h):\DMIS(\Alg{X})\to \DMIS(\Alg{Y})$. This follows easily by
Lemma \ref{lemma:PresDer} and Lemma \ref{lemma:existence}.

Then preservation of identity and arrow composition follows from the
fact that $\DMIS$ is defined as $\PDA$ on arrows and on the fact
that $\PDA$ is a functor.
\end{proof}

\begin{lem}\label{lemma:inclusion}
The inclusion $\iota_{\Alg{X}}: \PS(X) \to \PD(X)$ is a
$\Gamma(\Cat{C})$-homomorphism from the algebra $\DMIS(\Alg{X})$ to
$\PDA (\Alg{X})$.
\end{lem}
\begin{proof}
We have to prove that for all $A\in \PS(X)$ and $d \in
||\Cat{C}||$, $\iota_{\Alg{X}}(d_{\DMIS(\Alg{X})}(A))=
d_{\PDA(\Alg{X})}(\iota(A))$.

Let $(c,l,x)\in \iota_{\Alg{X}}(d_{\DMIS(\Alg{X})}(A))$, then there
exists $(c',l',x')\in A$ such that
$(c',l',x')\vdash_{\ded{T},\Alg{X}}^d (c,l,x)$. By definition of
$\vdash_{\ded{T},\Alg{X}}^d$, we also have that
$(c',l',x')\vdash_{\ded{T},\Alg{X}} (d;c,l,x)$ and since $A$ is
saturated, then $(d;c,l,x)\in A$. Since $\iota_{\Alg{X}}$ is simply
the inclusion, we also have that $(d;c,l,x)\in \iota_{\Alg{X}}(A)$
and thus, by definition of $d_{\PDA(\Alg{X})}$, $(c,l,x)\in
d_{\PDA(\Alg{X})}(\iota_{\Alg{X}}(A))$.

The other direction is analogous.
\end{proof}

\begin{oldlem}{lemma:inatural}
Let $\iota$ be the family of morphisms $\iota =
\{\iota_{\Alg{X}}:\DMIS(\Alg{X}) \to \PDA(\Alg{X}),  \forall
\Alg{X}\in |\Cat{Alg_{\Sig{\Gamma(\Cat{C})}}}|\}$. Then $\iota:\DMIS
\Rightarrow \PDA$ is a natural transformation.
\end{oldlem}
\begin{proof}
From Lemma \ref{lemma:inclusion}, it follows that each
$\iota_{\Alg{X}}$ is a morphism in
$\Cat{Alg_{\Sig{\Gamma(\Cat{C})}}}$. The fact that $\forall
h:\Alg{X}\to \Alg{Y}$, $\iota_{\Alg{Y}}; \PDA(h) = \DMIS(h);
\iota_{\Alg{Y}}$ follows from the fact that, by definition,
$\DMIS(h)=\PDA(h)$.
\end{proof}

\begin{oldlem}{lemma:Icorrect}
Let $\<\Alg{X}, \alpha\>$ be a $\PDA$-coalgebra. Then it is in
$|\coalg{\PDA^{I}}|$ iff it satisfies $\ded{T}$.
\end{oldlem}
\begin{proof}
Let $\<\Alg{X}, \alpha\>$ be a $\PDA$-coalgebra. If it satisfies
$\ded{T}$, then $\forall x \in \Alg{X}$, $\alpha(x)\in \PS(X)$. This
means that $\alpha$ factor through the inclusion
$\iota_{\Alg{X}}:\DMIS(\Alg{X}) \to \PDA(\Alg{X})$.

If $\ded{T}$ is not sound, then $\exists x \in \Alg{X}, d \in
||\Cat{C}||$ such that $(c,l,y)\in \alpha(x)$ and
$(c,l,y)\vdash_{\ded{T},\Alg{X}}^d(c',l',y')$ and $(c',l',y')\notin
\alpha(d_{\Alg{X}}(x))=d_{\PDA(\Alg{X})}(\alpha(x))$. From
$(c,l,y)\vdash_{\ded{T},\Alg{X}}^d(c',l',y')$, we have that
$(c,l,y)\vdash_{\ded{T},\Alg{X}}(d;c',l',y')$. From this setting
follows that $(d;c',l',y')\notin \alpha(c)$ because, otherwise, by
definition of $d_{\PDA(\Alg{X})}$, we would have that $(c',l',y')\in
d_{\PDA(\Alg{X})}(\alpha(x))$. Thus $\alpha(x)$ is not saturated,
i.e., $\alpha(x) \notin \PS(X)$.
\end{proof}

\begin{oldprop}{prop:covariety}
$|\coalg{\PDA^{I}}|$ is a covariety of $\coalg{\PDA}$, i.e., is
closed under:
\begin{enumerate}[\em(1)]
\item subcoalgebras,
\item homomorphic images,
\item sums.
\end{enumerate}
\end{oldprop}
\begin{proof}
A coalgebra $\<\Alg{X}, \alpha\>$ is a subcoalgebra of $\<\Alg{Y},
\beta\>$ if there is an arrow $m: \<\Alg{X}, \alpha\> \to \<\Alg{Y},
\beta\>$ that is mono in all its components (for a more formal
definition look at Appendix \ref{sec:factorization}). 

The fact that $|\coalg{\PDA^I}|$ is closed under subcoalgebras means
that whenever there is a subcoalgebra $m:\<\Alg{X}, \alpha\> \to
\<\Alg{Y}, \beta\>$ in $\coalg{\PDA}$ such that $\<\Alg{Y},
\beta\>\in |\coalg{\PDA^I}|$, then also $\<\Alg{X}, \alpha\> \in
|\coalg{\PDA^I}|$. This can be easily proved by employing Lemma
\ref{lemma:Icorrect}.

If $\<\Alg{Y}, \beta\> \in |\coalg{\PDA^I}|$, then it satisfies
$\ded{T}$. Suppose ab absurdum that $\<\Alg{X}, \alpha \>$ does not
satisfy $\ded{T}$. Then there exists $x\in |\Alg{X}|$,
$(c_1,l_1,x_1)\in \alpha(x)$ and $(c_2,l_2,x_2) \notin \alpha(x)$
such that $(c_1,l_1,x_1)\vdash_{\Alg{X},\ded{T}} (c_2,l_2,x_2)$.
Now, since $m$ is a cohomomorphism we have that $(c_1,l_1,m(x_1))\in
\beta(m(x))$. By Lemma \ref{lemma:PresDer}, it follows that
$(c_1,l_1,m(x_1))\vdash_{\Alg{Y},\ded{T}} (c_2,l_2,m(x_2))$. Since
$\<\Alg{Y}, \beta\>$ satisfies $\ded{T}$ then also
$(c_2,l_2,m(x_2))\in \beta(m(x))$. At this point, since $m$ is a
cohomomorphism then it must exist a $x_3 \in \Alg{X}$, such that
$(c_1,l_1,x_3)\in \alpha(x)$ and $m(x_3)=m(x_2)$. But since $m$ is
mono in all its components, then $x_2=x_3$ and thus
$(c_1,l_1,x_2)\in \alpha(x)$ against the hypothesis.

\medskip

Let $h: \<\Alg{X}, \alpha\> \to \<\Alg{Y}, \beta\>$ be an arrow in
$\coalg{\PDA}$. The homomorphic image of $\<\Alg{X}, \alpha\>$
through $h$, is the coalgebra $\<\Alg{I}, \gamma\>$ induced by the
unique factorization of $h=e;m$ (as shown below), where $e$ is an
arrow with all components epi and $m$ is an arrow with all
components mono (look at Appendix \ref{sec:factorization}).

\[
\xymatrix@R=10pt@C=10pt{
\Alg{X} \ar[rrrr]^{h} \ar[dddd]_{\alpha} \ar[rrd]^e& &   & &\Alg{Y}  \ar[dddd]^{\beta} \\
& & \Alg{I} \ar@{-->}[dddd]|{\gamma} \ar[rru]_m&&&\\
\\
\\
\PDA(\Alg{X}) \ar[rrrr]^{\PDA(h)} \ar[rrd]_{\PDA(e)}& &   &
&\PDA(\Alg{Y})\\
& & \PDA(\Alg{I}) \ar[rru]_{\PDA(m)} &&&\\
}
\]

The fact that $|\coalg{\PDA^I}|$ is closed under homomorphic images
means that whenever there is a cohomomorphism $h:\<\Alg{X}, \alpha\>
\to \<\Alg{Y}, \beta\>$ in $\coalg{\PDA}$ such that $\<\Alg{X},
\alpha\>\in |\coalg{\PDA^I}|$, then also $\<\Alg{I}, \gamma\> \in
|\coalg{\PDA^I}|$. This can be easily proved by employing Lemma
\ref{lemma:Icorrect}.

If $\<\Alg{X}, \alpha\> \in |\coalg{\PDA^I}|$, then it satisfies
$\ded{T}$. Suppose ab absurdum that $\<\Alg{I}, \gamma \>$ does not
satisfy $\ded{T}$. Then there exists an $i \in |\Alg{I}|$,
$(c_1,l_1,i_1)\in \gamma(i)$ and $(c_2,l_2,i_2) \notin \gamma(i)$
such that $(c_1,l_1,i_1)\vdash_{\Alg{I}, \ded{T}}(c_2,l_2,i_2)$.
Now, since $e$ is epi in all its components, there exists $x_1$,
such that $e(x_1)=i_1$ and since $e$ is a cohomomorphism there
exists $x\in \Alg{X}$ such that $h(x)=i$ and $(c_1,l_1,x_1)\in
\alpha(x)$.
By Lemma \ref{lemma:existence} and by $(c_1,l_1,i_1)\vdash_{\Alg{I},
\ded{T}}(c_2,l_2,i_2)$, it follows that there exists $x_2\in
\Alg(X)$ such that $e(x_2)=i_2$ and $(c_1,l_1,x_1) \vdash_{\Alg{X},
\ded{T}}(c_2,l_2,x_2)$. Now, since $\<\Alg{X}, \alpha\>$ satisfies
$\ded{T}$, then also $(c_2,l_2,x_2)\in \alpha(x)$. And now, since
$e$ is a cohomomorphism $(c_2,l_2,i_2)\in \gamma(i)$ against the
initial hypothesis.

In $\coalg{\PDA}$, all the colimits are defined as in
$\Cat{Alg_{\Sig{\Gamma(\Cat{C})}}}$ (for classical argument in
coalgebra theory). Recalling that
$\Cat{Alg_{\Sig{\Gamma(\Cat{C})}}}$ is isomorphic to
$\set^{\Cat{C}}$, it is easy to see that all colimits exists and
they are constructed as in $\set$. Thus, it is trivial to prove that
if $\<\Alg{X},\alpha\>$ and $\<\Alg{Y}, \beta\>$ satisfy $\ded{T}$,
also their sum, i.e., $\<\Alg{X}+\Alg{Y}, \alpha + \beta\>$,
satisfies $\ded{T}$.
\end{proof}

\begin{thm}\label{theo:finalsaturated}
$\coalg{\PDA^I}$ has final object $\final{\PDA^I}$.
\end{thm}
\begin{proof}
The proof is a standard argument in the theory of coalgebras.

Hereafter, we write ``$\PDA^I$-coalgebra'' as a short-hand for
``$\PDA$-coalgebra in $|\coalg{\PDA^I}|$''. In order to construct
$\final{\PDA^I}$, consider all the unique $\PDA$-cohomorphisms of
$\PDA^I$-coalgebras to $\final{\PDA}$ (the final object of
$\coalg{\PDA}$). Consider their homomorphic images through these
final morphisms. All of them are subobjects of $\final{\PDA}$ and
all of them are $\PDA^I$-coalgebras, because $|\coalg{\PDA^I}|$ is
closed under homomorphic images. Now, since these are subobjects of
$\final{\PDA}$, we can define
$\final{\PDA^I}$ as their union. 
In order to prove that $\final{\PDA^I}$ is final, it is important to
note that it is still a subcoalgebra of $\final{\PDA}$ (Corollary
1.4.14 of \cite{kurzThesis}), and thus we have a mono $m:
\final{\PDA^I} \to \final{\PDA}$\footnote{For this is important to
notice that all morphisms in $M_{\Cat{C}}$ (defined in Appendix
\ref{sec:factorization}) are also mono.}. Then for any
$\PDA^I$-coalgebra $\<\Alg{X}, \alpha\>$ there exists a morphism to
$\final{\PDA^I}$ since it is the union of all the images to
$\final{\PDA}$. Then, this morphism is unique since $m$ is mono.
Moreover, $\final{\PDA^I}$ satisfies $T$, since covarieties are also
closed by unions of subcoalgebras.

Another way of proving this theorem relies on Corollary 2.2.4 of
\cite{kurzThesis}. From such corollary and from Proposition
\ref{prop:covariety}, it follows that $\coalg{\PDA^I}$ is a
reflective subcategory of $\coalg{\PDA}$.
\end{proof}

\begin{cor}
$\coalg{\DMIS}$ has final object $\final{\DMIS}$.
\end{cor}
\begin{proof}
From the above theorem and from the fact that $\coalg{\DMIS}$ is
isomorphic to $\coalg{\PDA^I}$.
\end{proof}

\begin{cor}
Let $\<\Alg{X}, \alpha\>$ be a $\DMIS$-coalgebras. Let
$\unique{\PDA}{\<\Alg{X}, \alpha\>}$ be the unique morphism to
$\final{\PDA}$ and let $\unique{\DMIS}{\<\Alg{X}, \alpha\>}$ be the
unique morphism to $\final{\DMIS}$. Thus
$$\unique{\PDA}{\<\Alg{X}, \alpha\>}(x)=\unique{\PDA}{\<\Alg{X}, \alpha\>}(y) \text{ if and only if }
\unique{\DMIS}{\<\Alg{X}, \alpha\>}(x)=\unique{\DMIS}{\<\Alg{X},
\alpha\>}(y)\text{.}$$
\end{cor}
\begin{proof}
Note that $\final{\PDA^I} = \Fun{I}(\final{\DMIS})$ for $\Fun{I}$
being the functor described in Section \ref{sec:saturatedcoalgebra}.
Moreover, from the proof of the above theorem, we have that
$\final{\PDA^I}$ is a subobject of $\final{\PDA}$.
\end{proof}

\section{Proofs of Section \ref{sec:normalizedcoalgebra}}\label{proof2}

\noindent In this appendix we prove several lemmas that describe interesting
properties of the normalization function. In particular these
properties are useful to show that $\DMIN$ is a functor. Hereafter,
we will always implicitly assume to have a normalizable context
interactive system (Definition \ref{def:Normalizable}).
\begin{oldlem}{lemma:normalizedtrivial}
Let $\Alg{X}$ be a $\Gamma(\Cat{C})$-algebra. If
$(c_1,o_1,p_1)\vdash_{\ded{T}, \Alg{X}}(c_2,o_2,p_2)$ then
$p_2=e_{\Alg{X}}(p_1)$ for some $e\in ||\Cat{C}||$. Moreover
$\forall q_1\in Y$, $(c_1,o_1,q_1)\vdash_{\ded{T},
\Alg{X}}(c_2,o_2,e_{\Alg{X}}(q_1))$.
\end{oldlem}
\begin{proof}
Both observations trivially follows from the definition of
$\vdash_{\ded{T}, \Alg{X}}$ (Def. \ref{def:exDer}).
\end{proof}

\begin{oldlem}{lemmaPropNorm}
Let $\sys{I}$ be a normalizable system w.r.t.\ $\ded{T}$. Let
$\Alg{X}$ be $\Gamma(\Cat{C})$-algebra and $A \in \PD(X)$. Then
$\forall (d,o,x)\in A$, either $(d,o,x) \in \norm{X}(A)$ or $\exists (d',o',x') \in \norm{X}(A)$, such
that $(d',o',x')\prec_{\ded{T},\Alg{X}} (d,o,x)$.
\end{oldlem}
\begin{proof}
If there exists no $(d',o',x')$ with $(d',o',x') \prec_{\ded{T},\Alg{X}} (d,o,x)$, 
then $(d,o,x) \in \norm{X}(A)$. If it exists, then consider a chain $ \dots \prec_{\ded{T},\Alg{X}} (d_2,l_2,x_2)
\prec_{\ded{T},\Alg{X}} (d^1,l_1,x_1) \prec_{\ded{T},\Alg{X}}
(d,l,x)$. Since $\prec_{\ded{T},\Alg{X}}$ is well founded there
exists no infinite chains like this. Let $(d',l',x')\in A$ be the
last element of such a chain. Since it is the last, it is not
redundant and then $(d',l',x')\in \norm{X}(A)$. Moreover since
$\prec_{\ded{T},\Alg{X}}$ is transitive (as proved in the next
lemma), we have that $(d',l',x') \prec_{\ded{T},\Alg{X}}(d,l,x)$.
\end{proof}
\begin{lem}\label{lemma:PropertiesNormalization}
Let $\sys{I}$ be a context interactive system and $\ded{T}$ be an
inference system. Let $\Alg{X}$, $\Alg{Y}$ be
$\Gamma(\Cat{C})$-algebras.
\begin{enumerate}[\em(1)]
\item $\prec_{\ded{T},\Alg{X}}$ is transitive, \\(or better, if $(d'',l'',x'')\vdash_{\ded{T},\Alg{X}}(d',l',x')\prec_{\ded{T},\Alg{X}}(d,l,x)$ then
$(d'',l'',x'')\prec_{\ded{T},\Alg{X}}(d,l,x)$),
\item If $(d_0',l_0',x_0')\equiv_{\ded{T},\Alg{X}}(d_0,l_0,x_0) \prec_{\ded{T}, \Alg{X}}(d_1,l_1,x_1)\equiv_{\ded{T},\Alg{X}}
(d_1',l_1',x_1')$ then \\$(d_0',l_0',x_0') \prec_{\ded{T},\Alg{X}}
(d_1',l_1',x_1')$,
\item If $h: \Alg{X}\to \Alg{Y}$ and $(d,l,x)\equiv_{\ded{T}, \Alg{X}}
(d',l',x')$ then $(d,l,h(x)) \equiv_{\ded{T}, \Alg{Y}}
(d',l',h(x'))$.
\end{enumerate}
\end{lem}
\begin{proof}
Suppose that
$(d'',l'',x'')\vdash_{\ded{T},\Alg{X}}(d',l',x')\prec_{\ded{T},\Alg{X}}(d,l,x)$,
then we have both
\begin{center}
$(d'',l'',x'')\vdash_{\ded{T},\Alg{X}}(d',l',x')\vdash_{\ded{T},\Alg{X}}(d,l,x)$
and $(d,l,x)\nvdash_{\ded{T},\Alg{X}}(d',l',x')$.
\end{center}
%
We derive
$(d'',l'',x'')\vdash_{\ded{T},\Alg{X}}(d,l,x)$ by the former, and
$(d,l,x)\nvdash_{\ded{T},\Alg{X}}(d'',l'',x'')$ by the latter
(otherwise if $(d,l,x)\vdash_{\ded{T},\Alg{X}}(d'',l'',x'')$ then
also $(d,l,x)\vdash_{\ded{T},\Alg{X}}(d',l',x')$).

For the second point is sufficient to note that
\begin{center}
$(d_0',l_0',x_0')\vdash_{\ded{T},\Alg{X}}(d_0,l_0,x_0)
\vdash_{\ded{T}, \Alg{X}}(d_1,l_1,x_1)\vdash_{\ded{T},\Alg{X}}
(d_1',l_1',x_1')$,
\end{center}
and then $(d_0',l_0',x_0')\vdash_{\ded{T},\Alg{X}}(d_1',l_1',x_1')$.
Moreover
$(d_1',l_1',x_1')\nvdash_{\ded{T},\Alg{X}}(d_0',l_0',x_0')$, since
otherwise $(d_1,l_1,x_1)\vdash_{\ded{T},\Alg{X}} (d_0,l_0,x_0)$.

For the third point we use that $\vdash_{\ded{T},\Alg{X}}$ is
preserved by homomorphisms (Lemma \ref{lemma:PresDer}).
\end{proof}
\begin{lem}\label{lemma:BHO}
If $(d,l,x) \in \norm{X} ; c_{\DMIS(\Alg{X})} (A)$, then $(d,l,x)\in
c_{\DMIS(\Alg{X})} (A)$.
\end{lem}
\begin{proof}
If $(d,l,x) \in \norm{X} ; c_{\DMIS(\Alg{X})} (A)$, then by
definition of $c_{\DMIS(\Alg{X})}$, there exists  $(d',l',x')\in
\norm{X}(A)$ such that $(d',l',x')\vdash_{\ded{T}, \Alg{X}}^c
(d,l,x)$. Now by definition of normalization, there exists
$(d'',l'',x'')\in A$ such that $(d'',l'',x'') \equiv_{\ded{T},
\Alg{X}} (d',l',x')$. Then $(d'',l'',x'')\vdash_{\ded{T},
\Alg{X}}(d',l',x')\vdash_{\ded{T},\Alg{X}}^c(d,l,x)$, and then
$(d,l,x)\in c_{\DMIS(\Alg{X})} (A)$.
\end{proof}
\begin{lem}\label{lemma:normGood}
$\forall \Alg{X}, \Alg{Y} \in | \Cat{Alg_{\Gamma(\Cat{C})}} |$ and
$\forall h \in \Cat{Alg_{\Gamma(\Cat{C})}}[\Alg{X}, \Alg{Y}]$,
\begin{enumerate}[\em(1)]
\item $\norm{X}; d_{\DMIS(\Alg{X})}; \norm{X}= d_{\DMIS(\Alg{X})}; \norm{X}$,
\item $\norm{X}; \PDA(h); \norm{Y} = \PDA(h); \norm{Y}$,
\item $\norm{X}$ is idempotent.
\end{enumerate}
\end{lem}
\begin{proof}
For the first point we prove that $\forall A \in |\PDA(\Alg{X})|$
and $\forall c \in \Gamma$, $$c_{\DMIS(\Alg{X})}; \norm{X} (A) =
\norm{X} ; c_{\DMIS(\Alg{X})} ; \norm{X}(A)\text{.}$$
$$c_{\DMIS(\Alg{X})}; \norm{X} (A) \subseteq \norm{X} ;
c_{\DMIS(\Alg{X})} ; \norm{X}(A)$$
Suppose that $(e', l', x') \in
c_{\DMIS(\Alg{X})}; \norm{X} (A)$, then there exists $(e,l,x)\in
c_{\DMIS(\Alg{X})} (A)$ such that:
\begin{enumerate}[(1)]
\item $(e,l, x)\equiv_{\ded{T}, \Alg{X}} (e',l',x')$,
\item it is not redundant in $c_{\DMIS(\Alg{X})} (A)$.
\end{enumerate}
By definition of $c_{\DMIS(\Alg{X})}$, there exists
$(d_0,l_0,x_0)\in A$ such that $(d_0,l_0,x_0)\vdash_{\ded{T},
\Alg{X}}^c (e,l,x)$.

Now, by Lemma \ref{lemmaPropNorm}, there exists $(d_0',l_0',x_0')\in
\norm{X}(A)$ that either dominates $(d_0,l_0,x_0)$ or $(d_0',l_0',x_0')=(d_0,l_0,x_0)$. From definition of
$c_{\DMIS(\Alg{X})}$, it follows that $(e,l,x)\in \norm{X} ;
c_{\DMIS(\Alg{X})}(A)$. Now we have directly that $(e,l,x)\in
\norm{X} ; c_{\DMIS(\Alg{X})} ; \norm{X}(A)$. Indeed, suppose ab
absurdum that $(e,l,x)\notin \norm{X} ; c_{\DMIS(\Alg{X})} ;
\norm{X}(A)$, then there exists a $(e_1,l_1,x_1) \in \norm{X} ;
c_{\DMIS(\Alg{X})} (A)$ that dominates $(e,l,x)$. Now, by Lemma
\ref{lemma:BHO}, we have also that $(e_1,l_1,x_1) \in
c_{\DMIS(\Alg{X})} (A)$ that leads to absurd with
2.\\
Then $(e,l,x)\in \norm{X} ; c_{\DMIS(\Alg{X})} ; \norm{X}(A)$, and
also $(e',l',x')\in \norm{X} ; c_{\DMIS(\Alg{X})} \\; \norm{X}(A)$,
since the normalization function closes w.r.t.\ all equivalent
transitions.

$$\norm{X} ; c_{\DMIS(\Alg{X})} ; \norm{X}(A) \subseteq c_{\DMIS(\Alg{X})};
\norm{X} (A)$$
Suppose that $(e',l',x') \in \norm{X} ; c_{\DMIS(\Alg{X})} ;
\norm{X}(A)$, then there exists
\begin{center}$(e,l,x)\in \norm{X}
; c_{\DMIS(\Alg{X})}(A)$ such that:
\end{center}
\begin{enumerate}[(1)]
\item $(e,l,x)\equiv_{\ded{T}, \Alg{X}} (e',l',x')$,
\item it is not redundant in $\norm{X} ; c_{\DMIS(\Alg{X})}(A)$.
\end{enumerate}
Now, by Lemma \ref{lemma:BHO}, $(e,l,x)\in c_{\DMIS(\Alg{X})} (A)$.
Now we have that $(e,l,x)\in c_{\DMIS(\Alg{X})}; \\ \norm{X} (A)$.
Indeed, suppose ab absurdum that $(e,l,x)\notin c_{\DMIS(\Alg{X})} ;
\norm{X} (A)$, then there exists a $(e_1,l_1,x_1) \in
c_{\DMIS(\Alg{X})}(A)$ that dominates $(e,l,x)$. Now, by definition
of $c_{\DMIS(\Alg{X})}$, $(d_0'',l_0'', x_0'')\in A$ such that
$(d_0'',l_0'',x_0'')\vdash_{\ded{T},\Alg{X}}^c (e_1,l_1,x_1)$. Now,
by Lemma \ref{lemmaPropNorm}, and by $(d_0'',l_0'', x_0'')\in A$, it
follows that $(d_0''', l_0''', x_0''')\in \norm{X}(A)$ that either
dominates $(d_0'',l_0'',x_0'')$ or $(d_0''', l_0''', x_0''')=(d_0'',l_0'',x_0'')$. By definition of
$c_{\DMIS(\Alg{X})}$, $(e_1,l_1,x_1)\in \norm{X} ;
c_{\DMIS(\Alg{X})}(A)$ and this together with
2 leads to an absurd.\\
Thus $(e,l,x)\in c_{\DMIS(\Alg{X})} ; \norm{X} (A)$, and since
$(e,l,x)\equiv (e',l',x')$,
\begin{center}
$(e',l',x')\in c_{\DMIS(\Alg{X})} ; \norm{X} (A)$.
\end{center}\medskip

%
%

For the second point we prove that $\forall A \in \PDA(X)$,
$$\norm{X}; \PDA(h); \norm{Y}(A) = \PDA(h); \norm{Y}(A).$$
$$\norm{X}; \PDA(h); \norm{Y}(A) \subseteq \PDA(h); \norm{Y}(A)$$
Suppose that $(d',l',y') \in \norm{X}; \PDA(h); \norm{Y}(A)$. Then
there exists
\begin{center}
$(d,l,y)\in \norm{X}; \PDA(h)(A)$ such that
\end{center}
\begin{enumerate}[(1)]
\item $(d,l,y)\equiv_{\ded{T}, \Alg{Y}} (d',l',y')$,
\item it is not redundant in $\norm{X}; \PDA(h)(A)$.
\end{enumerate}
Then $\exists x \in X$ such that $h(x)=y$ and $(d,l,x)\in
\norm{X}(A)$ and then $\exists (d'',l'',x'')\in A$ such that
$(d,l,x) \equiv_{\ded{T}, \Alg{X}} (d'',l'',x'')$ and $(d'',l'',
h(x'')) \in \PDA(h)(A)$.

Now suppose ab absurdum that $(d'',l'', y'') \notin \PDA(h);
\norm{Y}(A)$ where $y''=h(x'')$. Then $\exists (d_0,l_0, y_0)\in
\PDA(h)(A)$ such that $(d_0,l_0,y_0) \prec_{\ded{T},\Alg{Y}}
(d'',l'', y'')$. However, if $(d_0,l_0,y_0)\in \PDA(h)(A)$, then
$(d_0,l_0,x_0)\in A$ such that $h(x_0)=y_0$ and by Lemma
\ref{lemmaPropNorm} there exists $(d_0',l_0',x_0')\in \norm{X}(A)$
that either dominates $(d_0,l_0, x_0)$ or  $(d_0',l_0',x_0')= (d_0,l_0,x_0)$.
By Lemma \ref{lemma:PresDer}, we have that $(d_0',l_0',h(x_0'))
\vdash_{\ded{T},\Alg{Y}} (d_0,l_0,h(x_0))
\prec_{\ded{T},\Alg{Y}}(d'',l'',y'')$ and, by Lemma
\ref{lemma:PropertiesNormalization}.1, $(d_0',l_0,h(x_0'))
\prec_{\ded{T},\Alg{Y}} (d'',l'',y'')\equiv_{\ded{T},\Alg{Y}}
(d,l,y)$. Since $(d_0',l_0',h(x_0')) \in \norm{X}; \PDA(h)(A)$, this
leads to an absurdum.

Now we have $(d'',l'', y'') \in \PDA(h); \norm{Y}(A)$ and 
$(d'',l'',y'')\equiv_{\ded{T},\Alg{Y}} (d,l,y)
\equiv_{\ded{T},\Alg{Y}} (d',l',y')$ and, since $\norm{Y}$ closes
w.r.t.\ all equivalent transitions, \begin{center}$(d',l',y')\in
\PDA(h); \norm{Y}(A)$.\end{center}

$$\PDA(h); \norm{Y}(A) \subseteq \norm{X}; \PDA(h); \norm{Y}(A)$$

Suppose that $(d',l', y') \in \PDA(h); \norm{Y}(A)$, then there
exists $(d,l,y)\in \PDA(h)(A)$, such that:
\begin{enumerate}[(1)]
\item $(d,l,y)\equiv_{\ded{T},\Alg{Y}} (d',l',y')$,
\item it is not redundant in $\PDA(h)(A)$.
\end{enumerate}
Then $\exists x \in \Alg{X}$, such that $h(x)=y$ and $(d,l,x)\in A$.

By Lemma \ref{lemmaPropNorm}, $\exists (d_0,l_0,x_0)\in \norm{X}(A)$
(and $(d_0,l_0,x_0)\in A$) that either dominates $(d,l,x)$ or $(d_0,l_0,x_0)=(d,l,x)$, and by Lemma
\ref{lemma:PresDer}
$(d_0,l_0,h(x_0))\vdash_{\ded{T},\Alg{Y}}(d,l,h(x))$. Now we have
two possible cases: or
$(d,l,h(x))\nvdash_{\ded{T},\Alg{Y}}(d_0,l_0,h(x_0))$, or
$(d,l,h(x))\vdash_{\ded{T},\Alg{Y}}(d_0,l_0,h(x_0))$. In the first
case we have that
$(d_0,l_0,h(x_0))\prec_{\ded{T},\Alg{Y}}(d,l,h(x))$, and this lead
to absurdum with 2. Then, only the latter is possible, i.e.,
$(d_0,l_0,h(x_0))\equiv_{\ded{T}, \Alg{Y}}(d,l,h(x))$.

Now suppose ab absurdum that $(d_0,l_0,h(x_0)) \notin \norm{X};
\PDA(h); \norm{Y}(A)$. Then $\exists (d_1,l_1,y_1)\in \norm{X};
\PDA(h)(A)$ that dominates $(d_0,l_0,h(x_0))$. Thus $\exists x_1 \in
\Alg{X}$ such that $h(x_1)=y_1$ and $(d_1,l_1,x_1)\in \norm{X}(A)$
and $(d_1',l_1',x_1')\in A$ such that
$(d_1',l_1',x_1')\equiv_{\ded{T}, \Alg{X}} (d_1,l_1,x_1)$.

Thus $(d_1',l_1',h(x_1')) \in \PDA(h) (A)$ and
\[(d_1',l_1',h(x_1'))\equiv_{\ded{T}, \Alg{Y}} (d_1,l_1,y_1)
\prec_{\ded{T},\Alg{Y}} (d_0,l_0,h(x_0))\equiv_{\ded{T}, \Alg{Y}}
(d,l,y),
\]
  i.e., $(d_1',l_1',h(x_1')) \prec_{\ded{T},\Alg{Y}}
  (d,l,y)$, against 2.

Then we have $(d_0,l_0,h(x_0))\in \norm{X}; \PDA(h);
\norm{Y}(A)$ and then also $(d',l',y')\in \norm{X}; \PDA(h);
\norm{Y}(A)$.

%
%
\medskip
For the third point we prove that $\forall A \in \DMIN(X)$,
$\norm{X}(A)=A$. This is trivial, since $\norm{X}$ junks away all
the redundant transitions and add all those equivalent. But since
$A$ is normalized, it does not contain any redundant transitions,
and it is still closed by equivalent transitions.
\end{proof}

\begin{prop}
$\DMIN:\Cat{Alg_{\Sig{\Gamma(\Cat{C})}}} \to
\Cat{Alg_{\Sig{\Gamma(\Cat{C})}}}$ is a functor.
\end{prop}
\proof
First of all we have to prove that $\forall \Alg{X}\in
\Cat{Alg_{\Sig{\Gamma(\Cat{C})}}}$, $\DMIN(\Alg{X})$ is a
$\Sig{\Gamma(\Cat{C})}$-algebra. In order to prove that, it is
enough to show that
$(c;d)_{\DMIN(\Alg{X})}=c_{\DMIN(\Alg{X})};d_{\DMIN(\Alg{X})}$ and
that $id_{\DMIN(\Alg{X})}$ is the identity function.

For the former, notice that
$(c;d)_{\DMIN(\Alg{X})}=(c;d)_{\DMIS(\Alg{X})};\norm{X}=c_{\DMIS(\Alg{X})};d_{\DMIS(\Alg{X})};\norm{X}$
since $\DMIS(\Alg{X})$ is a $\Sig{\Gamma(\Cat{C})}$-algebra. Now, By
Lemma \ref{lemma:normGood}.1, we have that it is equal to
\begin{center}$c_{\DMIS(\Alg{X})};\norm{X};d_{\DMIS(\Alg{X})};\norm{X}$,\end{center}
i.e., $c_{\DMIN(\Alg{X})};d_{\DMIN(\Alg{X})}$.

For the latter, notice that applying $id_{\DMIN(\Alg{X})}$ to a set
of transitions $A$, it is the same of closing $A$ w.r.t.\ the
derivation relation $\vdash_{\ded{T},\Alg{X}}$ and then normalizing
it. Now, if $A$ is normalized, one can close it w.r.t
$\vdash_{\ded{T},\Alg{X}}$, and then normalize it, obtaining the
same set $A$. This is formally proved by Proposition \ref{prop:iso}.

Now we prove that $\DMIN(h)$ is still a
$\Gamma(\Cat{C})$-homomorphism. Recall that $\PDA(h)=\DMIS(h)$,
$$c_{\DMIN(\Alg{X})}; \DMIN(h) = c_{\DMIS(\Alg{X})};\norm{X};\DMIN(h) =
c_{\DMIS(\Alg{X})};\norm{X};\DMIS(h) ; \norm{Y}=$$ (by Lemma
\ref{lemma:normGood}.2)
$$c_{\DMIS(\Alg{X})};\DMIS(h) ; \norm{Y} = \DMIS(h);
c_{\DMIS(\Alg{Y})}; \norm{Y} = $$ (by Lemma \ref{lemma:normGood}.1)
$$\DMIS(h) ; \norm{X}; c_{\DMIS(\Alg{X})} \norm{Y}=
\DMIN(h);c_{\DMIS(\Alg{Y})}; \norm{Y} = \DMIN(h);c_{\DMIN(\Alg{X})}
\text{.}$$

In order to prove that $\DMIN(id_{\Alg{X}})=id_{\DMIN(\Alg{X})}$ it
is enough to observe that
$\DMIN(id_{\Alg{X}})=\PDA(id{\Alg{X}});\norm{X}=id_{\PDA(\Alg{X})};\norm{X}$.
Since in $\DMIN(\Alg{X})$ all the elements are normalized, by Lemma
\ref{lemma:normGood}.3, normalization plays no role.

In order to prove that $\DMIN$ preserves composition we use Lemma
\ref{lemma:normGood}.2: $\forall h: \Alg{X}\to \Alg{Y}, g:\Alg{Y}\to
\Alg{Z}$,
$$\DMIN(h;g)=\PDA(h;g); \norm{Z}= \PDA(h) ; \PDA(g) ; \norm{Z}$$
$$=\PDA(h) ; \norm{Y} ; \PDA(g) ; \norm{Z}= \DMIN(h); \DMIN(g)\text{.}
  \eqno{\qEd}$$\medskip

\noindent At the end of the appendix we prove the main theorem. Note
that proof of Lemma \ref{lemmaPropNorm} is in Appendix \ref{proof2}.

\begin{lem}\label{lemma:NormSatHom}
$\norm{X}:\DMIS(\Alg{X}) \to \DMIN(\Alg{X})$ and
$\sat{X}:\DMIN(\Alg{X}) \to \DMIS(\Alg{X})$ are
$\Gamma(\Cat{C})$-homomorphisms.
\end{lem}
\begin{proof}
For all operators $c$, we have that
\begin{center}$c_{\DMIS(\Alg{X})}; \norm{X} = $ by Lemma \ref{lemma:normGood}.1 $= \norm{X} ; c_{\DMIS(\Alg{X})}; \norm{X} =
\norm{X}; c_{\DMIN(\Alg{X})}$.
\end{center}
For $\sat{X}$ we have that $c_{\DMIN(\Alg{X})}; \sat{X} =
c_{\DMIS(\Alg{X})} ; \norm{X}; \sat{X}$. Note that
\begin{center}$c_{\DMIS(\Alg{X})} ; \norm{X}; \sat{X} = c_{\DMIS(\Alg{X})} ;
\sat{X}$\end{center} since saturation adds everything that is
removed by normalization. At this point, it is enough to prove that
$c_{\DMIS(\Alg{X})} ; \sat{X}= \sat{X} ; c_{\DMIS(\Alg{X})}$.

We have to prove that $\forall A \in |\PDA(\Alg{X})|$,
$c_{\DMIS(\Alg{X})}; \sat{X} (A) = \sat{X}; c_{\DMIS(\Alg{X})}(A)$.

$$c_{\DMIS(\Alg{X})}; \sat{X} (A)\subseteq \sat{X}; c_{\DMIS(\Alg{X})}(A)$$
Suppose that $(e,l,x) \in c_{\DMIS(\Alg{X})}; \sat{X} (A)$, then
there exists $(e',l',x')\in c_{\DMIS(\Alg{X})}(A)$ such that
$(e',l',x') \vdash_{\ded{T}, \Alg{X}}(e,l,x)$, and by definition of
$c_{\DMIS(\Alg{X})}$, there exists $(e_0',l_0', x_0') \in A$ (and
then also in $\sat{X}(A)$) such that $(e_0',l_0',
x_0')\vdash_{\ded{T}, \Alg{X}}^c(e',l',x')\vdash_{\ded{T}
,\Alg{X}}(e,l,x)$. Then $(e_0',l_0', x_0')\vdash_{\ded{T},
\Alg{X}}^c(e,l,x)$, and then $(e,l,x)\in \sat{X};
c_{\DMIS(\Alg{X})}(A)$.

$$\sat{X}; c_{\DMIS(\Alg{X})}(A) \subseteq c_{\DMIS(\Alg{X})}; \sat{X} (A)$$
Suppose that $(e,l,x)\in \sat{X}; c_{\DMIS(\Alg{X})}(A)$, then there
exists $(d',l',x')\in \sat{X}(A)$ such that
$(d',l',x')\vdash_{\ded{T}, \Alg{X}}^c(e,l,x)$. Thus, by definition
of $\sat{X}$, there exists $(d'',l'',x'')\in A$ such that
$(d'',l'',x'')\vdash_{\ded{T}, \Alg{X}}(d',l',x')$. Then $(d'',l'',
x'')\vdash_{\ded{T}, \Alg{X}}^c(e,l,x)$ and then $(e,l,x)\in
c_{\DMIS(\Alg{X})}(A)$, and then $(e,l,x)\in c_{\DMIS(\Alg{X})};
\sat{X} (A)$.
\end{proof}

\begin{lem}\label{lemma:isomorphic}
$\normaR$ and $\satuR$ are isomorphisms, one the inverse of
the other.
\end{lem}
\begin{proof}
Since by Lemma \ref{lemma:NormSatHom}, $\norm{X}$ and $\sat{X}$ are
morphisms in $\Cat{Alg_{\Gamma (C)}}$, we have just to prove that
$\normR{X} ; \satR{X}=id_{\DMIS(\Alg{X})}$ and
$\satR{X} ; \normR{X}=id_{\DMIN(\Alg{X})}$.

$$\normR{X} ;\satR{X} (A) \subseteq A $$
If $(d,l,x)\in \normR{X} ;\satR{X} (A)$, then $(d',l',x')\in
\normR{X}(A)$ such that $(d',l',x')\vdash_{\ded{T},\Alg{X}}(d,l,x)$.
Thus $(d'',l'',x'')\in A$ such that $(d'',l'',x'')\equiv_{\ded{T},
\Alg{X}} (d',l',x')$. Then $(d'',l'',x'')\vdash_{\ded{T},
\Alg{X}}(d,l,x)$. Now, also $(d,l,x)\in A$, since $A$ is saturated.

$$A \subseteq \normR{X} ;\satR{X} (A)$$
If $(d,l,x) \in
A$ then, by Lemma \ref{lemmaPropNorm}, there exists $(d',l',x')\in
\normR{X}(A)$ that either dominates $(d,l,x)$ or $(d',l',x')=(d,l,x)$. Thus $(d,l,x) \in \normR{X}
;\satR{X} (A)$.

$$\satR{X} ; \normR{X} (A) \subseteq A$$
If $(d',l',x') \in \satR{X} ;
\normR{X} (A)$ then there exist $(d,l,x)\in \satR{X}(A)$ such that
\begin{enumerate}[(1)]
\item $(d,l,x)\equiv_{\ded{T}, \Alg{X}} (d',l',x')$,
\item it is not redundant in $\satR{X}(A)$.
\end{enumerate}
Then $\exists (d_0,l_0,x_0) \in A$ such that
$(d_0,l_0,x_0)\vdash_{\ded{T},\Alg{X}}(d,l,x)$. Now we have two
possibilities. Firstly, $(d,l,x)\nvdash_{\ded{T},
\Alg{X}}(d_0,l_0,x_0)$, then
$(d_0,l_0,x_0)\prec_{\ded{T},\Alg{X}}(d,l,x)$ and this is absurd
with 2. Secondly, $(d,l,x)\vdash_{\ded{T},\Alg{X}}(d_0,l_0,x_0)$ and
then $(d,l,x)\equiv_{\ded{T},\Alg{X}}(d_0,l_0,x_0)$, and since $A$
is normalized, $(d,l,x)\in A$.

$$A \subseteq \satR{X} ; \normR{X} (A)$$
If $(d,l,x)\in A$, then $(d,l,x) \in \satR{X}(A)$. Now suppose ab
absurdum that $(d,l,x)\notin \satR{X} ; \normR{X} (A)$ then there
exists a $(d',l',x')\in \satR{X}(A)$ that dominates $(d,l,x)$. Then,
by definition of $\satR{X}$, $(d'',l'',x'')\in A$ that dominates
$(d',l',x')$. But then $(d'',l'',x'')$ dominates also $(d,l,x)$,
against the hypothesis that $A$ is normalized.
\end{proof}

\begin{figure}
\begin{center}
   \begin{tabular}{ccc}
    \xymatrix@C=15pt@R=15pt
     {
      \ \Alg{X} \ \ar[rr]^{h} \ar[dd]|{\alpha}& & \ \Alg{Y} \ \ar[dd]|{\beta} \\
      \\
      \DMIS(\Alg{X}) \ar[rr]_{\DMIS(h)} & & \DMIS(\Alg{Y})
     }
     &
    \xymatrix@C=15pt@R=15pt
     {
      \DMIS(\Alg{X})\ar[rr]^{\DMIS(h)} \ar[dd]|{\norm{X}}& & \DMIS(\Alg{Y}) \ar[dd]|{\norm{Y}} \\
      \\
      \DMIN(\Alg{X}) \ar[rr]_{\DMIN(h)} & & \DMIN(\Alg{Y})
     }
     &
    \xymatrix@C=15pt@R=15pt
     {
      \DMIN(\Alg{X})\ar[rr]^{\DMIN(h)} \ar[dd]|{\sat{X}}& & \DMIN(\Alg{Y}) \ar[dd]|{\sat{Y}} \\
      \\
      \DMIS(\Alg{X}) \ar[rr]_{\DMIS(h)} & & \DMIS(\Alg{Y})
     }\\
     (i)&(ii)&(iii)
   \end{tabular}\caption{$\normaR$ and $\satuR$ are natural
   transformations.}\label{figureNaturalTransTrans}
\end{center}
\end{figure}

\begin{oldprop}{prop:iso}
Let $\normaR$, respectively, $\satuR$ be the families of morphisms
$ \{\norm{X}:\DMIS(\Alg{X}) \to \DMIN(\Alg{X}), \;\forall \Alg{X}
\in | \Cat{Alg_{\Gamma(\Cat{C})}}| \} $ and $
\{\sat{X}:\DMIN(\Alg{X}) \to \DMIS(\Alg{X}), \;\forall \Alg{X} \in
|\Cat{Alg_{\Gamma(\Cat{C})}}| \}$. Then $\normaR : \DMIS \Rightarrow
\DMIN$ and $\satuR : \DMIN \Rightarrow \DMIS$ are natural
transformations. More precisely, they are natural isomorphisms, 
one the inverse of the other.
\end{oldprop}
\begin{proof}
Since by Lemma \ref{lemma:isomorphic}, $\norm{X}$ and $\sat{X}$ are
one the inverse of the other, we have only to prove that they are natural transformation, i.e., that
diagrams (ii) and (iii) in Figure \ref{figureNaturalTransTrans}
commute. Notes that by definition, $\DMIN(h) = \PDA(h); \norm{Y}$
and thus diagram (ii) commutes by Lemma \ref{lemma:normGood}.2.

Then, by Lemma \ref{lemma:isomorphic}, also diagram (iii) commutes.
\end{proof}

\begin{oldthm}{thm:iso}
$\coalg{\DMIS}$ and $\coalg{\DMIN}$ are isomorphic.
\end{oldthm}
\begin{proof}
Let $\NormR\ : \coalg{\DMIS} \to \coalg{\DMIN}$ be the functor
sending an object $\<\Alg{X},\alpha\>$ into $\<\Alg{X},\alpha;\norm{X}\>$ and
any morphism $h$ to itself. Let $\SatR\ : \coalg{\DMIN} \to
\coalg{\DMIS}$ be the functor sending $\<\Alg{X},\alpha\>$ into
$\<\Alg{X},\alpha;\sat{X}\>$ and any morphism $h$ in itself. By
Proposition \ref{prop:iso}, these are clearly, one the inverse of
the other.
\end{proof}

\section{Factorization system for
$\PDA$-coalgebras}\label{sec:factorization}

\noindent The notions of \emph{subcoalgebra} and \emph{homomorphic image} have
been introduced in \cite{Rut96}, for coalgebras over $\set$. These
notions have been extended by Kurz in his thesis \cite{kurzThesis}
to coalgebras over a generic category $\Cat{C}$, by employing
factorization systems.

As subcolagebras and homomorphic images are fundamental for
proving that $|\coalg{\DMIS}|$ is a covariety of $\coalg{\PDA}$ (and
thus proving the existence of final system), we briefly report here
these definitions.

\begin{defi}[Factorization System]
Let $\Cat{C}$ be some category, and let $E,M$ be classes of
morphisms in $\Cat{C}$. Then $(E,M)$ is a factorization system for
$\Cat{C}$ if and only if
\begin{enumerate}[(1)]
\item $E,M$ are closed under isomorphism,
\item $\Cat{C}$ \emph{has $(E,M)$-factorizations}, i.e., every morphism $f$
in $\Cat{C}$ has a factorization $f= e;m$ for $e\in E$ and $m \in
M$,
\item $\Cat{C}$ \emph{has the unique $(E,M)$-diagonalisation property}, i.e.,
whenever the square
\[
\xymatrix{A \ar[d]^f \ar[r]^e & B \ar[d]^g \ar@{-->}[dl]|d \\
C \ar[r]_m& D}
\]
commutes for $m \in M$ and $e \in E$, then there is a unique
diagonal $d$ making the two triangle commute.
\end{enumerate}
\end{defi}\medskip

\noindent The theory of coalgebras has been mainly developed for coalgebras
over $\set$. In Section 1.4 of \cite{kurzThesis}, Kurz generalizes
this theory for coalgebras over a generic category $\Cat{C}$, by
providing four axioms relying on a factorization system for
$\Cat{C}$ and some properties of the endofunctor. These axioms
guarantees that the resulting category has all the good qualities of
coalgebras over $\set$, such as, for example, that the collection of
all subcoalgebras of a coalgebra is a complete lattice.

It can be easily proved (looking at
$\Cat{Alg_{\Sig{\Gamma(\Cat{C})}}}$ as $\set^{\Cat{C}}$) that the
endofunctor $\PDA$ satisfies these four axioms when considering the
following factorization system.

\begin{defi}
The factorization system for $\Cat{Alg_{\Sig{\Gamma(\Cat{C})}}}$ is
$(E_\Cat{C}, M_{\Cat{C}})$, where $E_{\Cat{C}}$ is the class of
$|\Cat{C}|$-indexed homomorphism having all components epi, while
$M_{\Cat{C}}$ is the class of $|\Cat{C}|$-indexed homomorphism
having all components mono.
\end{defi}


Here, we want to show that the forgetful functor $\Fun{U}:
\coalg{\PAS} \to \Cat{Alg_{\Sig{\Gamma(\Cat{C})}}}$ creates
factorizations with respect to $(E_{\Cat{C}}, M_{\Cat{C}})$ (Axiom
1.2). This means that if $h:(\Alg{X}, \alpha) \to (\Alg{Y}, \beta)$
is a morphism in $\coalg{\PAS}$ and $h=e;m$ is a factorization in
$(E_{\Cat{C}},M_{\Cat{C}})$, then it is also a factorization in
$\coalg{\PAS}$, i.e., $e,m$ are also cohomomorphisms. This is
graphically depicted below.

\[
\xymatrix@R=10pt@C=10pt{
\Alg{X} \ar[rrrr]^{h} \ar[dddd]_{\alpha} \ar[rrd]^e& &   & &\Alg{Y}  \ar[dddd]^{\beta} \\
& & \Alg{I} \ar@{-->}[dddd]|{\gamma} \ar[rru]_m&&&\\
\\
\\
\PDA(\Alg{X}) \ar[rrrr]^{\PDA(h)} \ar[rrd]_{\PDA(e)}& &   &
&\PDA(\Alg{Y})\\
& & \PDA(\Alg{I}) \ar[rru]_{\PDA(m)} &&&\\
}
\]
If the back square commutes and $h=e;m$ is factorization with
respect to $(E_{\Cat{C}}, M_{\Cat{C}})$, then also $\PDA(e)$ is in
$E_{\Cat{C}}$ and $\PDA(m)$ is in $M_{\Cat{C}}$. The unique arrow
$\gamma$ comes from the diagonalization property noting that:
\[
\xymatrix{\Alg{X} \ar[r]^e \ar[d]_{\alpha;\PDA(e)} & \Alg{I}
\ar[d]^{m; \beta} \ar@{-->}[ld]|{\gamma}
\\
\PDA(\Alg{I}) \ar[r]_{\PDA(m)}& \PDA(\Alg{Y})}
\]

At this point we can define subcoalgebra and homomorphic image.

\begin{defi}[Subcoalgebra]
Let $m :\<\Alg{X},\alpha\> \to \<\Alg{Y}, \beta\>$ be an arrow of
$\coalg{\PDA}$. Then $\<\Alg{X},\alpha\>$ is said a subcoalgebra of
$\<\Alg{Y}, \beta\>$ if $m \in M_{\Cat{C}}$.
\end{defi}

\begin{defi}[Homomorphic Image]
Let $f :\<\Alg{X},\alpha\> \to \<\Alg{Y}, \beta\>$ be an arrow of
$\coalg{\PDA}$. The homomorphic image of $\<\Alg{X},\alpha\>$
through $f$ is the coalgebra $\<\Alg{I}, \gamma\>$ shown in the
diagram above.
\end{defi}
\vspace{-20 pt}
\end{document}

%% file: macro.tex
\newcommand{\Fun}[1]{{\bf #1}}
\newcommand{\Cat}[1]{{\bf #1}}
\newcommand{\Sig}[1]{#1}
\newcommand{\Alg}[1]{\mathbb{#1}}

\newcommand{\alge}[1]{\Cat{Alg_{#1}}}
\newcommand{\Pow}{\Fun{P}}

\newcommand{\Pc}{\Pow_c}
\newcommand{\set}{\Cat{Set}}
\newcommand{\coalg}[1]{\Cat{Coalg_{#1}}}

\newcommand{\final}[1]{F_{#1}}

\newcommand{\subs}[2]{\{^{#1}/_{#2}\}}
\newcommand{\rulelabel}[1]{(\textsc{#1})}

\newcommand{\nil}{\mathbf{0}}

\newcommand{\nat}{{\omega}}
\newcommand{\deduz}[2]{\frac{\displaystyle #1}{\displaystyle #2}}

 \newcommand{\names}{\mbox{$\mathcal{N}$}}
 \newcommand{\inp}[2]{#1(#2)}
 \newcommand{\outp}[2]{\overline{#1}#2}
 \newcommand{\boutb}[2]{\overline{{#1}}({#2})}
 \newcommand{\fn}[1]{\mathrm{fn}(#1)}
 \newcommand{\bn}[1]{\mathrm{bn}(#1)}
 \newcommand{\nm}[1]{\mathrm{nm}(#1)}
 \newcommand{\res}[2]{\nu #1. #2 }

\newcommand{\Ttr}[2]{\xymatrix{ \ar[r]^{#1}_{#2} & }}
\newcommand{\UTtr}[2]{\Ttr{#1}{#2}}
\newcommand{\atr}[3]{\stackrel{#1, #2}{\rightarrowfill}_{#3}}
\newcommand{\satr}[2]{\stackrel{#1, #2}{\rightarrowfill}_{S}}
\newcommand{\satbis}{\sim^{S}}

\newcommand{\cts}{\textsc{satts}}
\newcommand{\sts}{\textsc{scts}}

\newcommand{\<}{\langle}
\renewcommand{\>}{\rangle}

\newcommand{\sort}{|\Cat{C}|}

\newcommand{\mssign}{\Cat{C}}
\newcommand{\sys}[1]{\mathcal{#1}}
\newcommand{\isys}{ \sys{I}=\<\Cat{C}, \Alg{X}, O, tr\>}

\newcommand{\sym}{\beta}
\newcommand{\symbis}{\sim^{SYM}}

\newcommand{\lts}{\textsc{lts}}

\newcommand{\infsys}{\ded{T}}

\newcommand{\modelAsy}{\mathcal{A}}

\newcommand{\amssign}{\Cat{Out}}
\newcommand{\aspec}{\Sig{\Gamma(\amssign)}}
\newcommand{\APA}{\Alg{A}}
\newcommand{\actAsy}{O_{\modelAsy}}
\newcommand{\trAsy}{tr_{\modelAsy}}
\newcommand{\modelAsyAll}{\< \amssign, \APA, \actAsy, \trAsy \>}
\newcommand{\asytr}[1]{\tr{#1}_{\modelAsy}}
\newcommand{\infAsy}{\ded{T}_{\modelAsy}}
\newcommand{\approxAsy}{\alpha}
\newcommand{\asyproc}{A}

\newcommand{\modelswc}{\mathcal{W}}

\newcommand{\wscmssign}{\Cat{Wor}}

\newcommand{\WSCA}{\Alg{W}}
\newcommand{\actwsc}{O_{\modelswc}}
\newcommand{\trwsc}{tr_{\modelswc}}
\newcommand{\modelwscAll}{\< \wscmssign, \WSCA, \actwsc, \trwsc \>}

\newcommand{\infwsc}{\ded{T}_{\modelswc}}
\newcommand{\approxwsc}{\omega}

\newcommand{\modelON}{\mathcal{N}}
\newcommand{\onsort}{|\onmssign|}

\newcommand{\onmssign}{\Cat{Tok}}

\newcommand{\ONPA}{\Alg{N}}
\newcommand{\OINPA}{\Alg{N}}
\newcommand{\ONSET}{N}
\newcommand{\actON}{\Lambda}
\newcommand{\trON}{tr_{\modelON}}
\newcommand{\ONtr}{tr_{\modelON}}
\newcommand{\modelONAll}{\< \onmssign, \ONPA, \actON, \trON \>}
\newcommand{\ontr}[1]{\tr{#1}_{\modelON}}
\newcommand{\infON}{\ded{T}_{\modelON}}
\newcommand{\approxON}{\eta}

\newcommand{\pr}{pre}
\newcommand{\pre}[1]{\ensuremath{\!~^\bullet{#1}}}
\newcommand{\po}{post}
\newcommand{\post}[1]{\ensuremath{{#1} {^\bullet}}}
\newcommand{\onet}[1]{#1}

\def\tr#1{\stackrel{#1}{\rightarrowfill}}
\def\tl#1{\stackrel{#1}{\leftarrow}}

\makeatletter
\def \rightarrowfill{\m@th\mathord{\smash-}\mkern-6mu%
  \cleaders\hbox{$\mkern-2mu\mathord{\smash-}\mkern-2mu$}\hfill
  \mkern-6mu\mathord\rightarrow}
\makeatother \makeatletter
\def \Rightarrowfill{\m@th\mathord{\smash-}\mkern-6mu%
  \cleaders\hbox{$\mkern-2mu\mathord{\smash-}\mkern-2mu$}\hfill
  \mkern-6mu\mathord\Rightarrow}
\makeatother \makeatletter
\def \mapstofill{\m@th\mathord{\smash-}\mkern-6mu%
  \cleaders\hbox{$\mkern-2mu\mathord{\smash-}\mkern-2mu$}\hfill
  \mkern-6mu\mathord\longmapsto}
\makeatother

\newcommand{\ded}[1]{#1}
\newcommand{\clos}[1]{\Phi(#1)}

\newcommand{\DMIN}{\Fun{N_{\ded{T}}}}
\newcommand{\DMIS}{\Fun{{S_\ded{T}}}}
\newcommand{\PAS}{\DMIS}

\newcommand{\PD}{\Fun{G}} 
\newcommand{\PS}{\Fun{S^{T}_{\Alg{X}}}} 
\newcommand{\PN}[1]{\Fun{N^{T}_{#1}}} 
\newcommand{\PDA}{\Fun{H}} 

\newcommand{\TPDA}{\PDA}
\newcommand{\unique}[2]{!^{#1}_{#2}}

\newcommand{\norm}[1]{norm_{\ded{T},\Alg{#1}}}
\newcommand{\sat}[1]{sat_{\ded{T},\Alg{#1}}}
\newcommand{\normR}[1]{\norm{#1}}
\newcommand{\satR}[1]{\sat{#1}}
\newcommand{\NormR}[1]{\Fun{NORM_\ded{T}}}
\newcommand{\SatR}[1]{\Fun{SAT_\ded{T}}}

\newcommand{\normaR}{norm_\ded{T}}
\newcommand{\satuR}{sat_\ded{T}}